%% file: main.tex
\documentclass[11pt]{article}
\usepackage{styling}
\usepackage{hyperref}
\hypersetup{
    colorlinks=true,
    linkcolor=violet,
    filecolor=magenta,      
    urlcolor=cyan,
    }
\usepackage[utf8]{inputenc}
\usepackage{booktabs}

\newcommand{\A}{\mathbf{A}}

\renewcommand{\norm}[2]{\left\|#1\right\|_{#2}}

\newcommand{\median}{\textnormal{median}}
\newcommand{\linf}[1]{\|#1\|_{\infty}}
\newcommand{\lp}[1]{\|#1\|_{p}}
\renewcommand{\epsilon}{\varepsilon}

\newcommand{\Proj}{\mathbb{P}}
\newcommand{\levSample}{\textsf{senSample}}
\newcommand{\start}{\text{start}}

\renewcommand{\SC}{\textnormal{\textbf{\textsf{SC}}}}
\newcommand{\PL}{\textnormal{\textbf{\textsf{PL}}}}
\newcommand{\textLARGE}{\textnormal{\textsc{Large}}}
\newcommand{\textSMALL}{\textnormal{\textsc{Small}}}
\newcommand{\textGOOD}{\textnormal{\textsc{Good}}}
\newcommand{\textBAD}{\textnormal{\textsc{Bad}}}
\newcommand{\Est}{\textnormal{\textsf{Est}}}
\newcommand{\keys}{\textnormal{keys}}

\newcommand{\sk}{{\rm{sk}}}
\newcommand{\key}{{\rm{key}}}
\newcommand{\val}{{\rm{val}}}

\newcommand{\B}{{\mathcal{B}}}
\newcommand{\vals}{{\rm{vals}}}
\renewcommand{\A}{{\mathcal{A}}}

\newcommand{\Ignore}{\textnormal{\textsc{Ignore}}}

\usepackage[ruled,linesnumbered]{algorithm2e}

\begin{document}

\allowdisplaybreaks

\begin{titlepage}
\author{
Hossein Esfandiari \\ Google Research \\ \texttt{esfandiari@google.com}{}
\and
Praneeth Kacham \\ CMU \\ \texttt{pkacham@cs.cmu.edu}
\and
Vahab Mirrokni \\ Google Research \\ \texttt{mirrokni@google.com}
\and
David P. Woodruff \\ CMU \\ \texttt{dwoodruf@cs.cmu.edu}
\and
Peilin Zhong \\ Google Research \\ \texttt{peilinz@google.com}
}
  
\title{Optimal Communication for Classic Functions in the Coordinator Model and Beyond}
\date{}

\maketitle

\begin{abstract}
\input{abs}
\end{abstract}
\setcounter{page}{0}
\thispagestyle{empty}

\end{titlepage}
\tableofcontents
\setcounter{page}{0}
\thispagestyle{empty} 
\clearpage

\input{intro}

\input{prelims}
\input{generalized_sampling}
\input{new_two_round}
\input{lowerbounds}
\input{framework}
\input{composable_leverage}
\\\\
{\bf Acknowledgments:} Praneeth Kacham and David P. Woodruff were supported in part by Google Research and a Simons Investigator Award. Part of this work was also done while David P. Woodruff was visiting the Simons Institute for the Theory of Computing.

\bibliographystyle{alpha}
\bibliography{ref}

\appendix
\input{appendix.tex}
\end{document}

%% file: abs.tex
In the coordinator model of communication with $s$ servers, given an arbitrary non-negative function $f$, we study the problem of approximating the sum $\sum_{i \in [n]}f(x_i)$ up to a $1 \pm \varepsilon$ factor. Here the vector $x \in \R^n$ is defined to be $x = x(1) + \cdots + x(s)$, where $x(j) \ge 0$ denotes the non-negative vector held by the $j$-th server. A special case of the problem is when $f(x) = x^k$ which corresponds to the well-studied problem of $F_k$ moment estimation in the distributed communication model. We introduce a new parameter $c_f[s]$ which captures the communication complexity of approximating $\sum_{i\in [n]} f(x_i)$ and for a broad class of functions $f$ which includes $f(x) = x^k$ for $k \ge 2$ and other robust functions such as the Huber loss function, we give a two round protocol that uses total communication $c_f[s]/\varepsilon^2$ bits, up to polylogarithmic factors. For this broad class of functions, our result improves upon the communication bounds achieved by  Kannan, Vempala, and Woodruff (COLT 2014) and Woodruff and Zhang (STOC 2012), obtaining the optimal communication up to polylogarithmic factors in the minimum number of rounds. We show that our protocol can also be used for approximating higher-order correlations.  
Our results are part of a broad framework for optimally sampling from a joint distribution in terms of the marginal distributions held on individual servers. 

Apart from the coordinator model, algorithms for other graph topologies in which each node is a server have been extensively studied. We argue that directly lifting protocols from the coordinator model to other graph topologies will require some nodes in the graph to send a lot of communication. Hence, a natural question is the type of problems that can be efficiently solved in general graph topologies. We address this question by giving communication efficient protocols in the so-called personalized CONGEST model for solving linear regression and low rank approximation by designing \emph{composable sketches}. Our sketch construction may be of independent interest and can implement any importance sampling procedure that has a \emph{monotonicity} property.

%% file: intro.tex
\section{Introduction}
In modern applications data is often distributed across multiple servers and communication is a bottleneck. This motivates minimizing the communication cost for solving classical functions of interest. A standard model of distributed computation is the {\it coordinator} or {\it message-passing} model, in which there are $s$ servers, each with an input, and a coordinator with no input. All communication goes through the coordinator, who decides who speaks next. This models arbitrary point-to-point communication up to a multiplicative factor of $2$ and an additive $\lceil \log_2 (\text{\# of servers})\rceil$ bits per message, since the coordinator can forward a message from server $i$ to server $j$ provided $i$ indicates which server should receive the message. The coordinator model is also useful in distributed functional monitoring \cite{CMY11}. Numerous functions have been studied in the coordinator model, such as bitwise operations on vectors \cite{PVZ16}, set-disjointness \cite{BEOPV13}, graph problems \cite{WZ17}, statistical problems \cite{WZ17}, and many more.

In the coordinator model, we measure the efficiency of a protocol by looking at the following: (i) the overall number of bits of communication required by the protocol and (ii) the number of rounds of communication in the protocol. In each round of communication, each of the servers sends a message to the coordinator based on their input and messages from the coordinator in previous rounds. Based on the messages received from all the servers in this round and earlier rounds, the coordinator sends a possibly distinct message to each of the servers. Thus, in a protocol with one round, each of the servers sends a message to the coordinator based \emph{only} on their inputs and the coordinator has to compute the output based only on these messages. We additionally assume that all the servers and the coordinator have access to a shared source of randomness which they can use to sample shared random variables. 

We revisit classical entrywise function approximation, which includes the $F_k$ moment estimation as a special case. Surprisingly, despite the simplicity of the coordinator model and the optimal bounds known for such functions in related models such as the streaming model, the optimal communication complexity of computing or approximating such functions in the coordinator model is open.  In the $F_k$ moment estimation
problem there are $s$ players, the $j$-th of which holds a non-negative vector\footnote{If the vectors are allowed to have negative entries, then there is an $\Omega(n^{1-2/k})$ bit lower bound on the amount of communication when $k > 2$ \cite{BJKS04}.} $x(j) \in \mathbb{R}^n$, and the goal is to, with constant probability, output a $(1+\epsilon)$-multiplicative approximation to $\|x\|_k^k = \sum_{j=1}^n |x_i|^k$, where $x = \sum_{j=1}^s x(j)$.
A large body of work has studied  $F_k$-moment estimation in a  stream, originating with work of Alon, Matias, and Szegedy \cite{alon1999space}. 

In the coordinator model, Cormode et al. \cite{CMY11} initiated the study of this problem and gave a protocol achieving $\tilde{O}(n^{1-2/k} \poly(s/\epsilon))$ bits of total communication for $k \geq 2$, where we use $\tilde{O}$ to suppress polylogarithmic factors in $n$. They optimize their bound for $k = 2$ and achieve a quadratic dependence on $s$. They also achieve protocols for $k \in \{0,1\}$ with a linear dependence on $s$. We note that their algorithms hold in the more general distributed functional monitoring framework. The upper bound was improved by Woodruff and Zhang to $\tilde{O}(s^{k-1}/\epsilon^{\Theta(k)})$ bits in \cite{WZ12}, which showed that a polynomial dependence on $n$ is not needed for $k > 2$. Unfortunately the $1/\epsilon^{\Theta(k)}$ multiplicative factor is prohibitive, and Kannan, Vempala, and Woodruff \cite{kannan2014principal} claimed an improved bound of $\tilde{O}((s^{k-1} + s^3)/\epsilon^3)$ bits. However, there appears to be a gap in their analysis which is not clear how to fix \cite{pers}. We describe this gap in Appendix \ref{app:bug}. Their algorithm for general function approximation can be used to obtain an algorithm which uses $\tilde{O}(s^k/\varepsilon^2)$ bits of communication. Thus the current state of the art is a $\min(\tilde{O}(s^{k-1}/\epsilon^{\Theta(k)}), \tilde{O}(s^k/\varepsilon^2))$ upper bound from \cite{WZ12, kannan2014principal} and the $\Omega(s^{k-1}/\epsilon^2)$ lower bound of \cite{WZ12}. Even if the work of \cite{kannan2014principal} can be fixed, it would not match the existing lower bounds, and an important open question is:
\begin{center}
    {\bf Question 1:} What is the complexity of $F_k$-estimation in the coordinator model?
\end{center}
As $F_k(x) = \sum_{i=1}^n |x_i|^k$ is just one example of an entrywise function $\sum_{i=1}^n f(x_i)$ for a non-negative function $f:\mathbb{R}_{\geq 0} \rightarrow \R_{\geq 0}$, it is natural to ask what the complexity of approximating $\sum_{i=1}^n f(x_i)$ is in terms of $f$. Indeed, a wide body of work originating with that of Braverman and Ostrovsky \cite{BO10} does exactly this for the related data stream model. In the coordinator model, Kannan, Vempala, and Woodruff \cite{kannan2014principal} attempt to characterize the complexity of $f$ by defining a parameter they call $c_{f,s}$, which is the smallest positive number for which 
$
f(y_1 + \cdots + y_s) \leq c_{f,s} (f(y_1) + \cdots + f(y_s))$ for all $y_1, \ldots, y_s \geq 0.$

For general functions $f$, they give a protocol which uses $O(s^2 c_{f,s}/\epsilon^2)$ bits of communication up to polylogarithmic factors. Assuming that the function $f$ is super-additive, i.e., $f(y_1 + y_2) \ge f(y_1) + f(y_2)$ for all $y_1, y_2 \ge 0$, their upper bound can be further improved to $O(s c_{f,s}/\varepsilon^2)$. They also give an $\Omega(c_{f,s}/\epsilon)$ communication lower bound. 

We note that a number of interesting entrywise functions have been considered in optimization contexts, such as the M-Estimators (see, e.g., \cite{CW15}), and a natural such estimator is the Huber loss function $f(x) = x^2/(2\tau)$ for $|x| \leq \tau$, and $f(x) = |x|-\tau/2$ otherwise. It is not hard to show $c_{f,s} = s$ for the Huber loss function, and so the best known upper bound is $\tilde{O}(s^2/\epsilon^2)$ bits while the lower bound is only $\Omega(s/\epsilon)$. Given the gap between the upper and lower bounds, it is unclear if the parameter $c_{f,s}$ captures the complexity of approximating the sum $\sum_i f(x_i)$. We observe that the communication complexity of the problem is better captured by a new parameter $c_f[s]$ defined as the smallest number for which
\begin{align}
    f(y_1 + \cdots + y_s) \le \frac{c_f[s]}{s}(\sqrt{f(y_1)} + \cdots + \sqrt{f(y_s)})^2\quad \text{for all}\quad y_1, \ldots, y_s \ge 0.
    \label{eqn:cfs-definition}
\end{align}
Since, $(\sqrt{f(y_1)} + \cdots + \sqrt{f(y_s)})^2 \ge f(y_1) + \cdots + f(y_s)$, we obtain $c_f[s] \le c_{f, s} \cdot s$ and using the Cauchy-Schwarz inequality, we can show that $c_{f,s} \le c_f[s]$. We consider the following question:
\begin{center}
    {\bf Question 2:} What is the complexity of entrywise function approximation in the coordinator model? Can one characterize the complexity completely in terms of $c_{f}[s]$?
\end{center}

Beyond the coordinator model, a number of works have looked at more general network topologies, such as \cite{CRR14,CLL17}.  One challenge in a more general network topology and without a coordinator is how to formally define the communication model. 
One of the most popular distributed frameworks in the past few decades is the CONGEST model~\cite{peleg2000distributed}.
In this model, each server is a node in the network topology and in each round it can send and simultaneously receive a possibly distinct message of bounded size\footnote{Usually $O(\log n)$ bits where $n$ is the number of nodes in the graph.} to and from its neighbors, as defined by the edges in the network, which is an unweighted undirected graph. Given the restriction on the size of the messages that can be sent in each round, efficiency of a protocol in this model is in general measured by the number of rounds required by an algorithm. 

Efficient CONGEST algorithms have been developed for shortest paths~\cite{ghaffari2018improved,haeupler2018faster}, independent sets~\cite{luby1986simple,ghaffari2019distributed}, matchings~\cite{ahmadi2018distributed,barenboim2016locality}, minimum spanning trees~\cite{kutten1998fast,ghaffari2017distributed}, and so on. A related distributed computation model is the on-device public-private model~\cite{eem19} that provides a framework for distributed computation with privacy considerations.

It is not hard to see that one cannot estimate $F_k$ of the sum of vectors as efficiently in the CONGEST model as one can in the coordinator model with $s$ servers, even if one has a two-level rooted tree where each non-leaf node has $s$ children. Indeed, the root and the $s$ nodes in the middle layer can have no input, at which point the problem reduces to the coordinator model with $s^2$ servers, and for which a stronger $\Omega((s^2)^{k-1})$ lower bound holds \cite{WZ12} and hence the average communication per node in the tree must be $\Omega(s^{2k-4})$ bits, as opposed to an average of $O(s^{k-2})$ bits per server in the coordinator model with $s$ servers. A natural question is which functions can be estimated efficiently with a more general network topology. Inspired by connections between frequency moment algorithms and randomized linear algebra (see, e.g., \cite{woodruff2013subspace}), we study the feasibility of communication efficient algorithms for various linear algebra problems in this setting. 

We assume that each server $v$ in the graph holds a matrix $A_v \in \R^{n_v \times d}$. Accordingly, we restrict the size of messages in each round to be $\poly(d, \log n, \log \max_v n_v)$ bits where $n$ is the number of nodes in the graph. 

We define a generalization of the CONGEST model called the ``{personalized CONGEST model}''. In this model, given a distance parameter $\Delta$, we would like, after $\Delta$ communication rounds, for each node to compute a function of all nodes reachable from it in at most $\Delta$ steps, which is referred to as its $\Delta$-hop neighborhood. Note that taking $\Delta$ to be the diameter of the graph, we obtain the CONGEST model. Such personalized solutions are desired in several applications such as recommendation systems and online advertisements.
For instance, in an application that wants to recommend a restaurant to a user, the data from devices in a different country with no relation to that particular user may not provide useful information and may even introduce some irrelevant bias.
Distributed problems with personalization have been studied in a line of work~\cite{park2007location,ceklm15,eem19}; however, to the best of our knowledge none of the prior work studies linear algebraic problems.

Ideally one would like to ``lift" a communication protocol for the coordinator model to obtain algorithms for the personalized CONGEST model. However, several challenges arise. The first is that if you have a protocol in the coordinator model which requires more than $1$-round, one may not be able to compute a function of the $\Delta$-hop neighborhood in only $\Delta$ rounds. Also, the communication may become too large if a node has to send different messages for each node in say, its $2$-hop neighborhood. Another issue is that in applications, one may be most interested in the maximum communication any node has to send, as it may correspond to an individual device, and so cannot be used for collecting a lot of messages and forwarding them. More subtly though, a major issue arises due to multiple distinct paths between two nodes $u$ and $v$ in the same $\Delta$-hop neighborhood. Indeed, if a server is say, interested in a subspace embedding of the union of all the rows held among servers in its $\Delta$-hop neighborhood, we do not want to count the same row twice, but it may be implicitly given a different weight depending on the number of paths it is involved in. 

Distributed algorithms for problems in randomized linear algebra, such as regression and low rank approximation, are well studied in the coordinator model ~\cite{balcan2015distributed,boutsidis2016optimal,kannan2014principal,feldman2020turning,balcan2014improved}, but they do not work for communication networks with a general topology such as social networks, mobile communication networks, the internet, and other networks that can be described by the CONGEST model. Hence we ask:
\begin{center}
{\bf Question 3:} For which of the problems in numerical linear algebra can one obtain algorithms in the personalized CONGEST model? 
\end{center}

\subsection{Our Results}
To answer Question 2, we give a two round protocol for approximating $\sum_i f(x_i)$ for non-negative functions which have an ``approximate-invertibility'' property.
\begin{definition}[Approximate Invertibility]
    We say that a function $f$ satisfies approximate invertibility with parameters $\theta, \theta', \theta'' > 1$ if all the following properties hold: (i) \emph{super-additivity}: $f(x + y) \ge f(x) + f(y)$ for all $x, y \ge 0$,
    (ii) for all $y \ge 0$, $f(\theta' y) \ge \theta \cdot f(y)$, and
    (iii) for all $y \ge 0$, $f(y/4 \cdot \sqrt{\theta} \cdot \theta') \ge f(y)/\theta''$.
\end{definition}
The super-additivity and the fact that $f(y) \ge 0$ for all $y$ implies that $f(0) = 0$. We note that any increasing convex function $f$ with $f(0) = 0$ satisfies the super-additivity property and hence it is not a very strong requirement. It is satisfied by $f(x) = x^k$ for $k \ge 1$ and the Huber loss function with any parameter. For such functions, we prove the following theorem:
\begin{theorem}[Informal, Theorem~\ref{thm:main-estimation}]
Let there be $s$ servers with the $j$-th server holding a non-negative $n$-dimensional vector $x(j)$, and define $x = x(1) + \cdots + x(s)$. Given a function $f$ which satisfies the approximate invertibility property with parameters $\theta, \theta', \theta'' > 1$, our two round protocol approximates $\sum_i f(x_i)$ up to a $1 \pm \varepsilon$ factor with probability $\ge 9/10$. Our protocol uses a total communication of $O_{\theta, \theta', \theta''}(c_f[s]/\varepsilon^2)$ bits up to polylogarithmic factors in the dimension $n$.
\end{theorem}
For $f(x) = x^k$, we see that $c_{f}[s] = s^{k-1}$ and can take $\theta = 2$, $\theta' = 2^{1/k}$ and $\theta'' = 2 \cdot 8^{k/2}$. Hence our algorithm uses $s^{k-1}/\varepsilon^2$ bits of total communication up to multiplicative factors depending on $k$ and $\log n$,  thus matching the known lower bounds from \cite{WZ12}. We additionally show that any one round algorithm must use $\Omega(s^{k-1}/\varepsilon^k)$ bits of communication and hence our protocol achieves the optimal communication bounds using the fewest possible number of rounds, thus resolving Question 1 completely. We summarize the results for $F_k$-moment estimation in Table \ref{fig:1}. 
\begin{table}[t]
    \centering
    \begin{tabular}{l  c  c}\toprule
         $F_k$ estimation algorithm                 & Upper Bound                         & Lower Bound            \\ \midrule
         2-round algorithm & $\tilde{O}(s^{k-1}/\varepsilon^2)$ (Corollary~\ref{cor:fk-estimation})  & $\Omega(s^{k-1}/\varepsilon^2)$ \cite{WZ12}\\
         1-round algorithm  & $\tilde{O}(s^{k-1}/\varepsilon^{\Theta(k)})$  \cite{WZ12} & $\tilde{\Omega}(s^{k-1}/\varepsilon^k)$ (Theorem~\ref{thm:Fk-one-round-lowerbound})\\
         \bottomrule
    \end{tabular}
    \caption{Upper and lower bounds on the total communication for $F_k$ approximation in the coordinator model. As mentioned, the claimed $\tilde{O}(\varepsilon^{-3}(s^{k-1} + s^3))$ upper bound for $F_k$ in \cite{kannan2014principal} has a gap.}
    \label{fig:1}
\end{table}

We can also use our protocol to approximate ``higher-order correlations'' studied by Kannan, Vempala and Woodruff \cite{kannan2014principal}. In this problem, each server $j$ holds a set of non-negative vectors $W_j$ and given functions $f : \R_{\ge 0} \rightarrow \R_{\ge 0}$ and $g : \R^k_{\ge 0} \rightarrow \R_{\ge 0}$, the correlation $M(f, g)$ is defined as
\begin{align}
    M(f, g, W_1, \ldots, W_k) \coloneqq \sum_{i_1, \ldots, i_k\, \text{distinct}}f\left(\sum_{j}\sum_{v \in W_j}g(v_{i_1}, \ldots, v_{i_k})\right).
    \label{eqn:defn-correlation}
\end{align}
Kannan, Vempala and Woodruff \cite{kannan2014principal} note that this problem has numerous applications and give some examples. Our protocol for estimating $\sum_i f(x_i)$ in the coordinator model extends in a straightforward way to the problem of estimating higher-order correlations. We show the following result:
\begin{theorem}[Informal, Theorem~\ref{thm:higher-order}]
Let there be $s$ servers with the $j$-th server holding a set of $n$-dimensional non-negative vectors $W_j$. Given a function $f$ that has the approximate invertibility property with parameters $\theta, \theta', \theta'' > 1$ and a function $g : \R^{k}_{\ge 0} \rightarrow \R_{\ge 0}$, our randomized two round protocol approximates $M(f, g, W_1, \ldots, W_k)$ up to a $1 \pm \varepsilon$ factor with probability $\ge 9/10$. The protocol uses a total of
$
    O_{\theta, \theta', \theta''}\left({c_f[s]}\poly(k, \log n)/\varepsilon^2\right)
$
bits of communication.
\end{theorem}

Our algorithm for approximating $\sum_i f(x_i)$ in the coordinator model is inspired by a one round protocol for sampling from an ``additively-defined distribution'' in the coordinator model, which can also approximate the sampling probability of the index that was sampled. This protocol lets us sample, in one round, from very general distributions such as the leverage scores. Our result is stated in the following theorem.
\begin{theorem}[Informal, Theorem~\ref{thm:additive-sampling}]
Given that each server $j$ has a non-negative vector $p(j) \in \R^{n}$, define $q_i \coloneqq \sum_j p_i(j)$. There is a randomized algorithm which outputs FAIL with probability $\le \frac{1}{\poly(n)}$ and conditioned on not outputting FAIL, it outputs a coordinate $\hat{i}$ along with a value $\hat{q}$ such that for all $i \in [n]$
\begin{align*}
    \Pr[\hat{i} = i\ \text{and}\ \hat{q} \in (1 \pm O(\varepsilon))\frac{q_i}{\sum_i q_i}] = (1 \pm O(\varepsilon))\frac{q_i}{\sum_i q_i} \pm \frac{1}{\poly(n)}.
\end{align*}
The algorithm uses one round and has a total communication of $O(s\polylog(n)/\varepsilon^2)$ words.
\end{theorem}

For Question 3, in the personalized CONGEST model, we show how to compute $\Delta$-hop subspace embeddings, approximately solve $\ell_p$-regression, and approximately solve low rank approximation efficiently. For example, for $\ell_p$-subspace embeddings and regression, we achieve $\tilde{O}(\Delta^2 n d^{p/2+2}/\epsilon^2)$ words of communication, which we optimize for $p = 2$ to $\tilde{O}(\Delta^2 n d^2)$ communication, where $n$ is the total number of nodes in the graph. Our algorithms are also efficient in that each node sends at most $\tilde{O}(\Delta \cdot d^{\max(p/2+2, 3)})$ communication to each of its neighbors in each of the rounds, which we optimize to $\tilde{O}(\Delta \cdot d^2)$ for $p = 2$. That is, the maximum communication per server is also small. We remark that in a round, the information sent by a node to all its neighbors is the same. Finally, our protocols are efficient, in that the total time, up to logarithmic factors is proportional to the number of non-zero entries across the servers, up to additive poly$(d/\epsilon)$ terms. Our results hold more broadly for sensitivity sampling for any optimization problem, which we explain in Section~\ref{subsec:techniques}. Our results in the CONGEST model are summarized in Table \ref{tab:congest}.
\begin{small}
\begin{table}[t]
    \centering
    \begin{tabular}{l l}\toprule
        Problem & Per node Communication in each round\\ \midrule
         $\ell_p$ subspace embeddings ($p \ne 2$) &  $\tilde{O}_{\Delta}(d^{\max(p/2+2, 3)}\varepsilon^{-2})$ (Theorem~\ref{thm:main-theorem-congest-embeddings})\\
         $\ell_p$ regression ($p \ne 2$) &  $\tilde{O}_{\Delta}(d^{\max(p/2+2, 3)}\varepsilon^{-2})$ (Section~\ref{subsec:congest-regression})\\ 
         $\ell_2$ subspace embeddings  &  $\tilde{O}_{\Delta}(d^2\varepsilon^{-2})$ (Theorem~\ref{thm:main-theorem-congest-embeddings})\\
         $\ell_2$ regression  &  $\tilde{O}_{\Delta}(d^{2}\varepsilon^{-2})$ (Section~\ref{subsec:congest-regression})\\ 
         Rank-$k$ Frobenius LRA & $\tilde{O}_{\Delta}(kd\varepsilon^{-3})$ (Section~\ref{subsec:congest-lra})\\ \bottomrule
    \end{tabular}
    \caption{Per node communication in each of the $\Delta$ rounds to solve the problems over data in a $\Delta$ neighborhood of each node in the CONGEST model}
    \label{tab:congest}
\end{table}    
\end{small}

\subsection{Our Techniques}\label{subsec:techniques}
\paragraph{Sampling from ``Additively-Defined'' Distributions.} At the heart of our results for approximating $\sum_{i=1}^n f(x_i) = \sum_{i=1}^n f(\sum_{j=1}^s x_i(j))$ is a general technique to sample from an ``additively-defined'' distribution using only one round of communication. To obtain our tightest communication bounds, our algorithm for approximating $\sum_{i=1}^n f(x_i)$ does not use this protocol in a black-box way, but the techniques used are similar to the ones used in this protocol. In this setting, the $j$-th server holds a non-negative vector $p(j) \in \R^{n}_{\ge 0}$ and the coordinator wants to sample from a distribution over $[n]$ with the probability of sampling $i \in [n]$ being proportional to 
$
    q_i = p_i(1) + \cdots + p_i(s).
$

Additionally if the coordinate $i$ is sampled by the coordinator, the coordinator also needs to be able to estimate the probability with which $i$ is sampled. This is an important requirement to obtain (approximately) unbiased estimators in applications involving importance sampling. If the coordinator just wants to sample from the distribution, it can run the following simple protocol: first, each server $j$ samples a coordinate $\bi_j$ from its local distribution, i.e., $\Pr[\bi_j = i] = p_i(j)/\sum_i p_i(j)$. Each server sends the coordinate $\bi_j$ along with the quantity $\sum_i p_i(j)$ to the coordinator. The coordinator then samples a random server $\bj$ from the distribution $\Pr[\bj = j] = (\sum_i p_i(j))/\sum_{j'}(\sum_i p_i(j'))$ and then takes $\bi_{\bj}$, i.e., the coordinate sent by the server $\bj$, to be the sample. We note that
\begin{align*}
    \Pr[\bi_\bj = i] = \sum_j \Pr[\bj = j] \Pr[\bi_j = i] = \sum_j \frac{\sum_{i'}p_{i'}(j)}{\sum_j \sum_{i'}p_{i'}(j)} \cdot \frac{p_i(j)}{\sum_{i'}p_{i'}(j)} = \frac{\sum_j p_i(j)}{\sum_j\sum_{i'}p_{i'}(j)} = \frac{q_i}{\sum_{i'}q_{i'}}.
\end{align*}
Hence the distribution of $\bi_{\bj}$ is \emph{correct}. But notice that there is no easy way for the coordinator to estimate the probability of sampling $\bi_{\bj}$ since it may not receive any information about this coordinate from the other servers. Therefore, to obtain the probability with which $\bi_{\bj}$ was sampled, the coordinator needs another round of communication, which we wish to avoid.

We will now give a protocol that can also approximate the sampling probabilities with only one round of communication. The protocol we describe here is a simpler version of the full protocol in Section~\ref{sec:generalized_sampler}. Consider the following way of sampling from the distribution in which $i$ has a probability proportional to $q_i$. Let $\be_1, \ldots, \be_n$ be independent standard exponential random variables. Let
$
    i^* = \argmax_{i \in [n]} \be_i^{-1}q_i = \argmax_{i \in [n]}\be_i^{-1}(p_i(1) + \cdots + p_i(s)).
$
By standard properties of the exponential random variable, we have
$
    \Pr[i^* = i] = \frac{\sum_{j}p_i(j)}{\sum_j \sum_{i'} p_{i'}(j)} = \frac{q_i}{\sum_{i'}q_{i'}}.
$

Hence, the random variable $i^*$ also has the \emph{right} distribution. The advantage now is that we can additionally show that with a high probability,
$
    \sum_i\be_i^{-1}\left(\sum_j p_i(j)\right) \le (C\log^2 n) \cdot \be_{i^*}^{-1}\sum_{j}p_{i^*}(j).
$
In other words, if we define $s$ vectors $r(1), \ldots, r(s)$, one at each of the servers, such that $r_i(j) = \be_i^{-1}p_i(j)$ and define $r = \sum_{j=1}^s r(j)$, then the coordinate $r_{i^*}$ is a $1/(C\log^2 n)$ $\ell_1$ \emph{heavy hitter}, as in $r_{i^*} \ge \|r\|_1/(C\log^2 n)$. 

Suppose a deterministic sketch matrix $S \in \R^{m \times n}$ is an $\alpha$-incoherent matrix for $\alpha = \varepsilon/(4C\log^2 n)$. Here we say that a matrix $S$ is $\alpha$-incoherent if all the columns of $S$ have unit Euclidean norm and for any $i \ne i' \in [n]$, $|\la S_{*i}, S_{*i'}\ra| \le \alpha$. Nelson, Nguyen and Woodruff \cite{nelson2014deterministic} give constructions of such matrices with $m = O(\log(n)/\alpha^2)$ rows and show that for any vector $x$,
$
    \|x - \T{S}Sx\|_{\infty} \le \alpha \|x\|_1.
$

Now suppose that each server $j$ computes the vector $S \cdot r(j)$ and uses $O(\polylog(n)/\varepsilon^2)$ words of communication to send the vector $S \cdot r(j)$ to the coordinator. The coordinator receives the vectors $S \cdot r(1), \ldots, S \cdot r(s)$ and  computes the vector $S \cdot r = S \cdot r(1) + \cdots + S \cdot r(s)$ and can then compute a vector $r' = \T{S}Sr$ satisfying
$
    \linf{r - r'} \le \alpha\|r\|_1.
$

Conditioned on the event that $\|r\|_1 \le (C\log^2 n) \cdot r_{i^*}$, as $\alpha$ is set to be $\varepsilon/(4C\log^2 n)$, we obtain that
$
    r'_{i^*} = (1 \pm \varepsilon/4) \cdot r_{i^*}
$
and for all $i \ne i^*$, we have $r'_i \le r_i + (\varepsilon/4) \cdot r_{i^*}$. Using properties of exponential random variables, we can also show that with probability $\ge 1 - O(\varepsilon)$, $\max_i r_i \ge (1 + \varepsilon) \cdot\text{second-max}_i\, r_i$. Conditioned on this event as well, we obtain that for all $i \ne i^*$, $r'_i < r_{i^*}(1 - \varepsilon/4) \le r'_{i^*}$. Hence the largest coordinate in the vector $r'$ is exactly $i^*$ and the value $r_{i^*}$ can be recovered up to a $1 \pm \varepsilon/4$ factor. Overall, the coordinator can find a coordinate $\hat{i}$ and a value $\hat{q} = \be_{\hat{i}} \cdot (r_{\hat{i}})'$ such that for all $i \in [n]$,
\begin{align}
    \Pr[\hat{i} = i\ \text{and}\ \hat{q} \in (1 \pm \varepsilon/4)q_{i}] = \frac{q_i}{\sum_{i'}q_{i'}} \pm O(\varepsilon).
    \label{eqn:overview-generalized-distribution}
\end{align}
Thus this procedure lets us sample from a distribution close to that of the desired distribution while at the same time lets us approximate the probability of the drawn sample. In Section~\ref{sec:generalized_sampler}, we give a protocol, which instead of using incoherent matrices, uses an ``$\ell_1$ sampling-based" algorithm to compute the coordinate $i^*$ and approximate the value $q_{i^*}$. We end up obtaining a sample $\hat{i} \in [n]$ and a value $\hat{q}$ for which 
\begin{align}
    \Pr[\hat{i} = i\, \text{and}\, \hat{q} = (1\pm \varepsilon)\frac{q_i}{\sum_{i'}q_{i'}}] = (1 \pm \varepsilon)\frac{q_i}{\sum_{i'}q_{i'}} \pm \frac{1}{\poly(n)}.
\end{align}
Additionally, computing the coordinate $\hat{i}$ does not require $\Omega(n)$ time at the coordinator using the protocol in Section~\ref{sec:generalized_sampler}.
\paragraph{Function Sum Approximation.}
Our aim is to obtain an algorithm which approximates the sum $\sum_{i=1}^n f(x_i) = \sum_{i=1}^n f(\sum_{j=1}^s x_i(j))$ up to a $1 \pm \varepsilon$ factor for a non-negative, super-additive function $f$. The protocol described above for sampling from additively-defined distributions shows that correlating randomness across all the servers using exponential random variables is a powerful primitive in this context.

In the function sum approximation problem, each server holds a non-negative vector $x(1), \ldots, x(s)$,  respectively, and the coordinator wants to approximate $\sum_i f(x_i) = \sum_i f(\sum_j x_i(j))$. Suppose that each server $j$ defines a vector $p(j) \in \R^n$ such that $p_i(j) = f(x_i(j))$. Then the above described protocol can be used to sample approximately from the distribution in which $i$ has probability
    $(1\pm \varepsilon)\frac{f(x_i(1)) + \cdots + f(x_i(s))}{\sum_{i'}\sum_j f(x_{i'}(j))} \pm \frac{1}{\poly(n)}.$
Using super-additivity and the definition of the parameter $c_f[s]$, we obtain that the above probability is at least
$
    (1\pm \varepsilon)\frac{f(x_i)}{c_f[s] \cdot \sum_{i'}f(x_{i'})} \pm \frac{1}{\poly(n)}.
$
This distribution is off by a multiplicative $c_f[s]$ factor from the distribution we need to sample coordinates from in order to estimate $\sum_i f(x_i)$ with a low variance. Thus, we need $O(c_f[s]/\varepsilon^2)$ samples from the above distribution, and this overall requires a communication of $O(s \cdot c_f[s]/\varepsilon^4)$ bits, which is more than the total communication required by the protocol of \cite{kannan2014principal}. 

Additionally, when the coordinate $i$ is sampled, the protocol lets us estimate $f(x_i(1)) + \cdots + f(x_i(s))$, but the quantity we want to construct an estimator for is the value $f(x_i(1) + \cdots + x_i(s))$, which requires an additional round of communication and defeats the point of obtaining a protocol that can approximate the sampling probability in the same round. Overall, a protocol based on this procedure requires $O(s \cdot c_f[s] \cdot \polylog(n)/\varepsilon^4)$ bits of communication and two rounds of communication.

Thus, we need a different technique to obtain algorithms which can approximate $\sum_i f(x_i)$ more efficiently. We continue to use exponential random variables to correlate the randomness across the servers but instead heavily use the max-stability property  to design our protocol. Let $\be_1, \ldots, \be_n$ be independent standard exponential random variables. The max-stability property asserts that for any $f_1, \ldots, f_n \ge 0$, the random variable $\max(f_1/\be_1,\ldots,f_n/\be_n)$ has the same distribution as $(\sum_{i=1}^n f_i)/\be$ where $\be$ is also a standard exponential random variable. Now using the \emph{median} of $O(1/\varepsilon^2)$ independent copies of the random variable $(\sum_{i=1}^n f_i)/\be$, we can compute an approximation of $\sum_{i=1}^n f_i$ up to a $1 \pm \varepsilon$ factor with high probability. Thus, in the coordinator model, if there is a protocol that can find the value of the random variable $\max_i f(x_i)/\be_i$, then we can use it to compute a $1 \pm \varepsilon$ approximation to $\sum_i f(x_i)$ by running the protocol for $O(1/\varepsilon^2)$ independent copies of the exponential random variables. From here on, we explain how we construct such a protocol.

Given the exponential random variables $\be_1, \ldots, \be_n$, define
$
 i^* \coloneqq \argmax_{i \in [n]}\be_i^{-1}f(x_i).   
$
As mentioned above we would like to find the value of $\be_{i^*}^{-1}f(x_{i^*})$. The ``heavy-hitter'' property we used previously shows that with probability $\ge 1 - \frac{1}{\poly(n)}$,
\begin{align}
    \sum_i \be_i^{-1} f(x_i) \le (C\log^2 n) \cdot \max_i \be_i^{-1}f(x_i).
    \label{eqn:intro-max-is-large}
\end{align}
This is the main property that leads to a communication efficient algorithm that can identify the max coordinate $i^*$ and the value $x_{i^*}$. From here on, condition on the above event. 

Note that we are shooting for a protocol that uses at most $O(c_f[s] \cdot \polylog(n))$ bits of total communication and succeeds in computing the value $\be_{i^*}^{-1}f(x_{i^*})$ with a probability $\ge 1 - \frac{1}{\poly(n)}$. Fix a server $j \in [s]$. Consider the random variable $\bi$ which takes values in the set $[n]$ according to the distribution
$
    \Pr[\bi = i] = {\be_i^{-1}f(x_i(j))}/{\sum_{i \in [n]}\be_i^{-1}f(x_i(j))}.
$
Since the server $j$ knows all the values, $x_1(j), x_2(j), \ldots, x_n(j)$, it can compute the above probability distribution and can sample $N$ (for a value to be chosen later) independent copies $\bi_1, \ldots, \bi_N$ of the random variable $\bi$. Let $\SC_j \coloneqq \set{\bi_1, \ldots, \bi_N}$\footnote{We use the notation $\SC_j$ since it denotes the ``\textsf{\textbf{S}ampled \textbf{C}oordinates}'' at server $j$.} be the set of coordinates sampled by server $j$. Server $j$ then sends the set $\SC_j$ along with the values $x_i(j)$ for $i \in \SC_j$. Note that the coordinator can compute $f(x_i(j))$ since it knows the definition of the function $f$. The total communication from all the servers to the coordinator until this point is $O(s \cdot N)$ words. 

Now define $\SC \coloneqq \bigcup_j\,  \SC_j$ to be the set of coordinates that is received by the central coordinator from all the servers. We will first argue that $i^* \in \SC$ with a large probability if the number of sampled coordinates at each server $N = \Omega(c_f[s] \cdot \polylog(n)/s)$. To prove this, we use the fact that $\be_{i^*}^{-1}f(x_{i^*})$ is significantly large since we conditioned on the event in \eqref{eqn:intro-max-is-large}, and then apply the definition \eqref{eqn:cfs-definition} of the parameter $c_f[s]$.

Conditioned on the event that the coordinate $i^* \in \SC$, a simple algorithm to determine the value of $\be_{i^*}^{-1}f(x_{i^*})$ would be for the coordinator to query the value of $x_i(j)$ for all $i \in \SC$ from all the servers $j$. Unfortunately, this requires a total communication of $\Omega(s \cdot c_f[s] \cdot \polylog(n))$ bits since the set $\SC$ could have a size as large as $s \cdot N = \Omega(c_f[s] \cdot \polylog(n))$. As we are aiming for a protocol that uses about $O(c_f[s] \cdot \polylog(n))$ bits of total communication, the coordinator cannot ask for the values of all the coordinates in the set $\SC$.

If the coordinator finds a smaller subset $\PL\footnote{We use $\PL$ to denote that these set of coordinates have ``\textsf{\textbf{P}robably \textbf{L}arge}'' values of $\be_i^{-1}f(x_i)$.} \subseteq \SC \subseteq [n]$ that contains $i^*$, with a size $|\PL|$ of about $\polylog(n)$, the coordinator can then query for $x_i(j)$ for $i \in \PL$ for all $j$ using only a communication of $O(s \cdot \polylog(n)) = O(c_f[s] \cdot \polylog(n))$ bits since $c_f[s] \ge s$. 

From here on condition on the event that $i^* \in \SC$. To find such a small subset $\PL$, our strategy is to construct $\hat{x}_i$ for each $i \in \SC$ so that the following properties are simultaneously satisfied with probability $\ge 1 - \frac{1}{\poly(n)}$:
    (i) for all $i \in \SC$, $\hat{x}_i \le x_i$ and
    (ii) $\be_{i^*}^{-1}f(\hat{x}_{i^*}) \ge \alpha \cdot \be_{i^*}^{-1}f(x_{i^*})$ for some value $\alpha < 1$.
Define $\Est_i \coloneqq \be_{i}^{-1}f(\hat{x}_i)$ for $i \in \SC$. If the constructed values $\hat{x}_i$ for $i \in \SC$ satisfy these two properties, we have that $\Est_i \le \be_i^{-1}f(x_i)$ for all $i$ and $\Est_{i^*} \ge \alpha \cdot \be_{i^*}^{-1}f(x_{i^*})$ by monotonicity of $f$. 

Recall that we conditioned on the event $\sum_i \be_{i}^{-1}f(x_i) \le (C \log^2 n) \cdot \be_{i^*}^{-1}f(x_{i^*})$ which then implies that the number of coordinates $i$, with $\Est_i \ge \Est_{i^*}$ is at most $(C\log^2 n)/\alpha$. Now, if we define $\PL$ to be the set of coordinates $i \in \SC$ with the $C\log^2 n/\alpha$ largest values, then we have that $i^* \in \PL$. The coordinator can then determine $\be_{i^*}^{-1}f(x_{i^*})$ after a second round of communication from the servers in which it asks for the values of $x_i(j)$ from all the servers $j$ only for the coordinates $i \in \PL$.

Fix a coordinate $i \in \SC$. We will briefly describe how $\hat{x}_i$ is computed by the coordinator using only the information it receives in the first round of communication from the servers. We say $x_i(j)$ is the    \emph{contribution} of server $j$ to $x_i$. If $i \ne i^*$, then we can safely ignore the contribution from any number of servers to $x_i$ when we are trying to construct the estimator $\hat{x}_i$. However, the coordinator does not know what $i^*$ is and cannot arbitrarily drop the contribution from servers when trying to compute the estimator $\hat{x}_i$.

We first define two disjoint subsets of servers $\textLARGE_i$ and $\textSMALL_i$: we put $j \in \textLARGE_i$ if the probability of the coordinate $i$ being sampled at $j$ is \emph{very high} and we put $j \in \textSMALL_i$ if the probability is \emph{very low}. Since the probability of $i$ being sampled at the servers in $\textLARGE_i$ is very large, we can union bound over all $i \in [n]$ and all servers $j \in \textLARGE_i$ and assume that it does happen, and can therefore estimate the contribution of all the servers in $\textLARGE_i$ exactly.

We then argue that the contribution from all the servers $j \in \textSMALL_i$ can be ``ignored'': by ignoring the contribution of a set of servers $S_i$, we mean that $\be_{i^*}^{-1}f(x_{i^*} - \sum_{j \in S_{i^*}}x_{i^*}(j))$ is a significant portion of $\be_{i^*}^{-1}f(x_{i^*})$. Note that we only need to care about how excluding the contribution of $S_{i^*}$ to $x_{i^*}$ affects our ability in obtaining $\hat{x}_{i^*}$ which satisfies the above property and we need not care about the effects of excluding the contribution to $x_i$ from the set of servers $S_i$ for all other $i \ne i^*$ since we are only trying to underestimate such $x_i$.

We then need to estimate the contribution to $x_i$ from the ``intermediate'' servers, i.e., those that are neither \textLARGE{} nor \textSMALL{}. We bucket the servers $j$ based on the values $\sum_i \be_i^{-1}f(x_i(j))$ and $\be_i^{-1}f(x_i(j))/\sum_i \be_i^{-1}f(x_i(j))$ (the probability that a given sample at server $j$ is equal to $i$). We show that the \emph{size} of a bucket is enough to approximate the contribution from all the servers in the bucket towards $x_i$ and argue that if the size is not very large, then the contribution from the bucket can be ignored in the sense described above. Of course, the coordinator can not determine the size of a bucket but given that a server $j$ sampled the coordinate $i$, the coordinator can compute which bucket the server $j$ belongs to. We show that when the size of the bucket is large enough, the number of servers in the bucket that sample the coordinate $i$ is concentrated enough that we can estimate its size. We thus identify which buckets have a large size based on the number of servers in the bucket that sample $i$, and for each bucket that we identify as large, we approximate the size up to constant factor. This can then be used to approximate the contribution of the bucket to $x_i$, and hence obtain the estimator $\hat{x}_i$. 

This wraps up our protocol for computing $\max_i \be_i^{-1}f(x_i)$, with a high probability, and using a total of $O(c_f[s] \cdot \polylog(n))$ bits of total communication across two rounds. Running this protocol concurrently for $O(1/\varepsilon^2)$ copies of the exponential random variables, we can then obtain a $1 \pm \varepsilon$ approximation to $\sum_i f(x_i)$ with high probability.

\paragraph{The Personalized CONGEST Model.}
We focus on constructing $\ell_2$-subspace embeddings, though the arguments here are analogous for $\ell_p$-subspace embeddings by using the $\ell_p$-sensitivities instead of the leverage scores, defined below.
Recall that given an $n \times d$ matrix $A$, we say that a matrix $M$ is $1/2$ subspace embedding for $A$ if
for all vectors $x$,
    $\opnorm{Mx}^2 = (1 \pm 1/2)\opnorm{Ax}^2$.
Subspace embeddings have numerous applications in obtaining fast algorithms for problems such as linear regression, low rank approximation, etc. 

Our main technique is to use the \emph{same} uniform random variables across all the servers to \emph{coordinate} the random samples across all the servers in a useful way. 
For simplicity, assume that there is a node $\alpha$ connected to $s$ neighbors (servers) such that the $j$-th neighbor holds a matrix $A^{(j)} \in \R^{n_j \times d}$, which does not have any duplicate rows. Further assume that we want to compute a subspace embedding for the matrix $A$ obtained by the union of the rows of the matrices $A^{(j)}$, i.e., if a row $v$ is present in say both $A^{(1)}$ and $A^{(2)}$, it appears only once in the matrix $A$. 

Given a matrix $A$, we recall the standard definition of the leverage scores for each of the rows of $A$, as well as the standard construction to obtain a subspace embedding from the leverage scores. If $v$ is a row of the matrix $A$, define the leverage score $\tau_A(v)$ of v to be:
    $\tau_A(v) \coloneqq \max_{x : Ax \ne 0}\frac{|\la v, x\ra|^2}{\opnorm{Ax}^2}$.
For convenience, we define $\tau_A(v) = 0$ if $v$ is not a row of the matrix $A$.
One has that the sum of the leverage scores of all the rows in a matrix $A$ is at most $d$. Now construct the $n \times n$ random diagonal matrix $\bD$ as follows: for each $i \in [n]$ independently set $\bD_{i,i}$ to be $1/\sqrt{p_i}$ with probability $p_i = \min(1, q_i)$ where $q_i \ge C\tau_A(a_i)\log d$ and $a_i$ is the $i$-th row of $A$, and set it to $0$ otherwise. We note that since $\sum_{i \in [n]}\tau_A(a_i) \le d$, the random matrix $\bD$ has at most $O(d\log d)$ nonzero entries with a large probability. One can now show that if $C$ is large enough, then with probability $\ge 99/100$, the matrix $\bD A$ is a $1/2$ subspace embedding for $A$. We note that the matrix $\bD A$ has at most $O(d \log d)$ nonzero rows with a large probability. This algorithm is known as \emph{leverage score sampling}.

Going back to our setting, we want to implement leverage score sampling on the matrix $A$ which is formed by the {\it union} of the rows of the matrices $A^{(1)}, \ldots, A^{(s)}$, i.e., we only count a row once even if it appears on multiple servers. As in the coordinator model where we used the same exponential random variables across servers, a key idea we use in the CONGEST model is \emph{correlated randomness}.  This time, for each possible row $v$ that could be held by any server, all the servers choose the same threshold $h(v)$ uniformly at random from the interval $[0,1]$. We treat $h(\cdot)$ as a fully random hash function mapping row $v$ to a uniform random number from the interval $[0,1]$.

Each server $j$ now computes the $\ell_2$ leverage score of each of the rows in its matrix $A^{(j)}$. Note if $v$ is held by two different servers $j \neq j'$, then it could be that $\tau_{A^{(j)}}(v) \ne \tau_{A^{(j')}}(v)$. 
 
Server $j$ sends all its rows $v$ that satisfy $h(v) \le C\tau_{A^{(j)}}(v)\log d$ to node $\alpha$. Additionally, assume that the server sends the value $\tau_{A^{(j)}}(v)$ along with the row $v$. Since $h(v)$ is picked uniformly at random from the interval $[0,1]$, the probability that a row $v$ is sent to the node $\alpha$ by the $j$-th server is $\min(1, C\tau_j(v)\log d)$. Hence, each server is implementing leverage score sampling of its own rows and sending all the rows that have been sampled to node $\alpha$.

Now we note that if $v$ is a row of the matrix $A^{(j)}$, then $\tau_{A^{(j)}}(v) \ge \tau_{A}(v)$ which directly follows from the definition of leverage scores. For any $v$ that is a row of the matrix $A$, we have $\tau_A(v) \le \max_{j \in [s]}\tau_{A^{(j)}}(v)$. Since $h(v)$ is same across all the servers, then the probability that a row $v$ of the matrix $A$ is sent to the node $\alpha$ is exactly, $\min(1, \max_{j \in [s]}C\tau_{A^{(j)}}(v)\log d)$. For all the rows $v$ that are received by node $\alpha$, it can also compute $\min(1, \max_{j \in [s]}C\tau_{A^{(j)}}(v)\log d)$ since it also receives the values $\tau_{A^{(j)}}(v)$ from all the servers that send the row $v$. Now using the fact that $\tau_A(v) \le \max_{j \in [s]}\tau_{A^{(j)}}(v)$ for all rows $v$ of $A$, the union of rows that are received by the node $\alpha$ correspond to a leverage score sampling of the matrix $A$. Since the node $\alpha$ can also compute the probability that each row it receives was sampled with, it can appropriately scale the rows and obtain a subspace embedding for the matrix $A$. In the above procedure, each server $j$ sends at most $O(d\log d)$ rows and therefore the subspace embedding constructed by the node $\alpha$ for matrix $A$ has at most $O(sd\log d)$ rows.

Even though we described a procedure to compute a subspace embedding of the union of neighboring matrices at a single node $\alpha$, if the nodes send the rows that are under the threshold $h(v)$ to all their neighbors, this procedure can simultaneously compute a subspace embedding at each node for a matrix that corresponds to the union of neighbor matrices of that node. This solves the $1$-neighborhood version of the more general $\Delta$-neighborhood problem we introduced.

Now consider how we can compute a subspace embedding for the distance $2$ neighborhood matrix. Note that we cannot run the same procedure on $1$-neighborhood subspace embeddings to obtain $2$-neighborhood subspace embeddings. We again use the monotonicity of leverage scores. Consider the node $\alpha$ and the matrix $A$ as defined before. Let $v$ be a row that it receives from one of its neighbors in the first round. Suppose the node $\alpha$ can compute $\tau_{A}(v)$, the leverage score of $v$ with respect to the matrix $A$. Now the node $\alpha$ forwards the row $v$ to its neighbors if $h(v) \le C\tau_A(v)\log d$. Suppose $\beta$ is neighbor of the node $\alpha$. Thus, after the second round, the rows received by $\beta$ then correspond to performing a leverage score sampling of the distance-$2$ neighborhood matrix for the node $\beta$ and it can then compute a subspace embedding for that matrix!

Now the main question is how can the node $\alpha$ compute the leverage scores $\tau_A(v)$? By definition of a subspace embedding, we note that if $M$ is a subspace embedding for $A$, then $M$ can be used to approximate $\tau_A(v)$. Since we already saw that the node $\alpha$ can compute a subspace embedding for the $1$-neighborhood matrix, it can also approximate $\tau_A(v)$ for all the rows $v$ that it receives. However an issue arises where we are using the set of rows that $\alpha$ receives in the first round themselves to approximate their leverage scores and therefore their sampling probabilities in the second round. This leads to correlations and it is unclear how to analyze leverage score sampling with such correlations. To solve for this issue, we use two independent hash functions $h_1(\cdot)$ and $h_2(\cdot)$. Using the sample of rows received by the node $\alpha$ when the $1$-neighborhood procedure from above is run using hash function $h_1(\cdot)$, it computes a subspace embedding for the matrix $A$ and then uses this subspace embedding to approximate the leverage scores of the rows that it receives when the $1$-neighborhood procedure run using hash function $h_2(\cdot)$. The node $\alpha$ then uses these approximate leverage scores to decide which of the rows that it received are to be forwarded to its neighbors. This decouples the probability computation and the sampling procedure and the proof of leverage score sampling goes through.

This procedure is similarly extended to compute subspace embeddings for the $\Delta$-neighborhood matrices at each node in the graph. In each round, we use a fresh subspace embedding and use it to compute approximate leverage scores and filter out the rows and then forward them to the neighbors. This way of decorrelating randomness is similar to the {\it sketch switching} method for adversarial streams in \cite{BJWY22}, though we have not seen it used in this context. 

This general procedure of collecting data from neighbors, \emph{shrinking} the collected data and transferring the data to all the neighbors is called ``graph propagation''. Any procedure such as ours above which can handle duplicates can be readily applied in this framework so that each node in the graph can simultaneously learn some statistic/solve a problem over the data in its neighborhood.

\nocite{jayaram2019towards,jayaram2019weighted}

%% file: prelims.tex
\section{Preliminaries}
\subsection{Notation}
For a positive integer $n$, we use the notation $[n]$ to denote the set $\set{1,\ldots, n}$. For two non-negative numbers $a$ and $b$ and a parameter $0 < \varepsilon < 1$, we use $a = (1 \pm \varepsilon)b$ if $(1-\varepsilon)b \le a \le (1+\varepsilon)b$. Given a vector $x \in \R^n$ and $k \ge 1$, $\|x\|_k = (\sum_{i}|x_i|^k)^{1/k}$ denotes the $\ell_k$ norm of the vector $x$. We use the notation $F_k(x)$ for $\|x\|_k^k$. Given a matrix $A$, $\norm{A}{F}$ denotes the Frobenius norm defined as $(\sum_{i,j}A_{ij}^2)^{1/2}$ and for an index $i$, $A_i$ denotes the $i$-th row of matrix $A$. We use the notation $\tilde{O}(f)$ to hide the factors that are polylogarithmic in the parameters of the problem such as the dimension or $1/\varepsilon$ where $\varepsilon$ is an accuracy parameter, etc. For a parameter $k$, we use the notation $O_k(f)$ to hide the multiplicative factors that depend purely on $k$.
\subsection{Exponential Random Variables}\label{subsec:exponentials}
We use the following properties of exponential random variables extensively throughout the paper.
\begin{enumerate}
    \item If $\be$ is a standard exponential random variable, then $\Pr[\be \ge C\log n] = 1/n^C$ and $\Pr[\be \le t] \le t$ for any $C, t \ge 0$.
    \item If $f_1,\ldots,f_n \ge 0$ are arbitrary and $\be_1,\ldots,\be_n$ are standard exponential random variables, then $\max_i \be_i^{-1}f_i$ has the same distribution as $\be^{-1}(\sum_i f_i)$. This property is referred to as ``max-stability''. This shows that with probability $\ge 1 - 1/n^C, \max_i \be_i^{-1}f_i \ge (\sum_i f_i)/C\log n$.
    \item In the same setting, if $i^* = \argmax_i \be_i^{-1}f_i$, then $\Pr[i^* = i] = f_i/\sum_j f_j$.
    \item Let $\be_1,\ldots,\be_t$ are $t$ independent standard exponential random variables. If $t = \Omega(1/\varepsilon^2)$,  then
    \begin{align*}
        \Pr[(1-\varepsilon)\ln(2) \le \median(\be_1,\ldots,\be_t) \le (1+\varepsilon)\ln(2)] \ge 9/10.
    \end{align*}
    Hence for any $F > 0$, $\median(F/\be_1,\ldots,F/\be_t) \cdot \ln(2) \in [(1-\varepsilon)F, (1+\varepsilon)F]$ with probability $\ge 9/10$.
\end{enumerate}
We use the following lemma extensively which shows that with a probability $\ge 1 - 1/\poly(n)$, $\sum_i \be_i^{-1}f_i \le (C\log^2 n)\max_i \be_i^{-1}f_i$. 
\begin{lemma}
Given any $f_1,\ldots,f_n \ge 0$, let $F = \sum_{i=1}^n f_i$. With probability $\ge 1 - 1/\poly(n)$,
\begin{align*}
    \frac{\max f_i/\be_i}{\sum_{i=1}^n f_i/\be_i} \ge \frac{1}{C\log^2 n}. 
\end{align*}
\label{lma:max-is-significant}
\end{lemma}
\begin{proof}
Condition on the event that $1/n^3 \le \be_i \le 3\log n$ for all $i$. The event has a probability of at least $1 - 2/n^2$ by a union bound. Let the event be denoted $\calE$. Conditioned on $\calE$, we have $\max_i f_i/\be_i \le n^3F$. Now consider the interval $I = [F/(Cn\log n), n^3F]$. The values $f_i/\be_i$ that are smaller than $F/(Cn\log n)$ contribute at most $F/(C\log n)$ to the sum $\sum_{i=1}^n f_i/\be_i$.

Partition the interval $I$ into $I_1,\ldots,$ such that $I_j = [2^{j-1}(F/(Cn\log n)), 2^{j}(F/(Cn\log n)))$. Note that there are at most $O(\log n)$ such intervals. We will use the fact that conditioned on $\calE$, the random variables $\be_1,\ldots,\be_n$ are still independent. 

Let $\bX_j$ be the number of indices $i$ such that $f_i/\be_i \in I_j$. We have
\begin{align*}
    \sum_{i} f_i/\be_i \le \frac{F}{C\log n} + \sum_{j=0}^{O(\log n)}2^j\bX_j \frac{F}{Cn\log n}.
\end{align*}
Let $\bY_{ij} = 1$ if $f_i/\be_i \ge 2^{j-1}(F/(Cn\log n))$ and $\bY_{ij}= 0$ otherwise. Note that $\bX_j \le \sum_{i=1}^n \bY_{ij}$ and that for a fixed $j$, the random variables $\bY_{ij}$ are mutually independent given $\calE$. We have
\begin{align*}
    \Pr[\bY_{ij} = 1 \mid \calE] = \Pr\left[\be_i \le \frac{Cn\log n}{2^{j-1}} \frac{f_i}{F} \mid \calE\right] \le \frac{2Cn\log n}{2^{j-1}}\frac{f_i}{F}
\end{align*}
since the p.d.f. of the exponential distribution is bounded above by $1$ and $\Pr[\calE] \ge 1/2$. By linearity of expectation, $\E[\sum_{i} \bY_{ij} \mid \calE] \le 2Cn\log n/(2^{j-1})$. By Bernstein's inequality, we get
\begin{align*}
    \Pr[\sum_{i} \bY_{ij} \ge (\E[\sum_{i} \bY_{ij} \mid \calE] + t) \mid \calE] \le \exp\left(-\frac{t^2/2}{2Cn\log n/(2^{j-1} )+ t/3}\right).
\end{align*}
Setting $t = 2Cn\log n/(2^{j-1}) + 6\log n$, we obtain that
\begin{align*}
    \Pr[\sum_{i} \bY_{ij} \ge (\E[\sum_{i} \bY_{ij} \mid \calE] + t) \mid \calE] \le \frac{1}{n^2}.
\end{align*}
We union bound over all $j$ and obtain that for all $j$,
\begin{align*}
    \bX_j \le \frac{2Cn\log n}{2^{j-1}} + 6\log n.
\end{align*}
Additionally, we note that if $\bX_j \ne 0$, then $2^{j}(F/(Cn\log n)) \le 2\max_i f_i/\be_i$. Now,
\begin{align*}
    \sum_i f_i/\be_i \le \frac{F}{C\log n} + \sum_{j : \bX_j \ne 0} 2^j \bX_j \frac{F}{Cn \log n} &\le \sum_{j: \bX_j \ne 0}\left( \frac{2Cn\log n}{2^{j-1}} \frac{2^j F}{Cn\log n} + (6\log n)\max_i f_i/\be_i\right)\\
    &\le O(F \log n) + O(\log^2 n) \max_i f_i/\be_i.
\end{align*}
As $\max_i f_i/\be_i \ge F/(3\log n)$ with probability $\ge 1 - 1/n^2$ conditioned on $\calE$, we get that with probability $\ge 1 - 1/\poly(n)$,
\begin{align*}
    \frac{\max_i f_i/\be_i}{\sum_i f_i/\be_i} \ge \frac{1}{C\log^2 n}.&\qedhere
\end{align*}
\end{proof}
If $\max_i f_i/\be_i \ge (\sum_i f_i/\be_i)/C\log^2 n$, we obtain that for any parameter $T \ge 1$,
\begin{align*}
    |\setbuilder{i}{f_i/\be_i \ge (1/T)\max_i f_i/\be_i}| \le TC\log^2 n
\end{align*}
which shows that there are only a small number of indices for which $f_i/\be_i$ is comparable to $\max_i f_i/\be_i$.

%% file: generalized_sampling.tex
\section{Sampling from Additively-Defined Distributions}\label{sec:generalized_sampler}
Consider the setting of $s$ servers with a coordinator. Assume that the $j$-th server has a non-negative vector $p(j) \in \R^{n}$ and the coordinator wants $N$ independent samples from the distribution supported on $[n]$ with the probability of $i \in [n]$ being
\begin{align}
    q_i \coloneqq \frac{\sum_j p_i(j)}{\sum_{i,j}p_i(j)}.
    \label{eqn:general-distribution}
\end{align}
For each sample, we would ideally also want an estimate to the probability of obtaining that sample. The sampling probabilities are necessary to scale the statistics appropriately to obtain unbiased estimators. A simple one round algorithm for sampling from such a distribution is for all the servers $j$ to send the value $\sum_{i}p_i(j)$ and an index $i$ with probability proportional to $p_i(j)$. The coordinator then picks a server $j$ with probability proportional to $\sum_i p_i(j)$ and chooses the index $i$ that the server $j$ sampled. Note that the coordinate $i$ sampled this way has the desired distribution. While this sampling procedure only requires one round of communication, note that the coordinator does not have the value $\sum_j p_i(j)$ nor even a way to estimate it. Hence if the index $i$ was sampled by the coordinator, it is not clear how to compute (or even approximate) $q_i$ without further rounds of communication. If another round of communication is allowed, the coordinator can send the sampled coordinate $i$ to all the servers which in turn report the values $p_i(j)$ from which the coordinator can compute $q_i$ exactly.

Using exponential random variables, we show that there is a protocol using which the coordinator can sample a coordinate $i$ and approximate the probability $q_i$ with only one round of communication.

Let $\be_1,\ldots,\be_n$ be independent exponential random variables that all the servers (and the coordinator) sample using the shared randomness. Each server $j$ locally computes $\be_i^{-1}p_i(j)$. Let $i^* \coloneqq \argmax_i \be_i^{-1}\sum_j p_i(j)$. As we have seen, the probability distribution of $i^*$ is exactly as in \eqref{eqn:general-distribution} and the advantage now is that with probability $\ge 1-1/\poly(n)$,
\begin{align*}
    \frac{\be_{i^*}^{-1}\sum_{j}p_{i^*}(j)}{\sum_{i,j}\be_i^{-1}p_i(j)} \ge \frac{1}{C\log^2 n}.
\end{align*}
We can further show the following lemma which is helpful to isolate the max coordinate $i^*$.
\begin{lemma}
Let $f_1,\ldots, f_n \ge 0$. and $i^* = \argmax_{i \in [n]}\be_i^{-1}f_i$ where $\be_1,\ldots,\be_n$ are independent standard exponential random variables. We have for all $i \in [n]$ that
\begin{align*}
    \frac{f_i}{\sum_{i'} f_{i'}} \ge \Pr\left[i^* = i\ \text{and}\ \be_{i^*}^{-1}f_{i^*} \ge (1+\varepsilon)\max_{i' \ne i^*}\be_i^{-1}f_{i'}\right] \ge \frac{f_i}{(1+\varepsilon)\sum_{i'}f_{i'}}.
\end{align*}
\end{lemma}
\begin{proof}
    Fix $i\in [n]$. The probability in the theorem statement is equivalent to the probability of $\be_i^{-1} f_i \ge (1+\varepsilon)\max_{i' \ne i}
    \be_{i'}^{-1}f_{i'}$. The probability is clearly at most $f_i/\sum_{i'}f_{i'}$.
    
    By min-stability of exponential random variables, $\max_{i' \ne i}\be_{i'}^{-1}f_{i'}$ is distributed as $\be^{-1}\sum_{i' \ne i}f_{i'}$. Note that the exponential random variable $\be$ is independent of $\be_i$. By standard arguments, 
    \[
    \Pr[\be_i^{-1}f_i \ge \be^{-1}(1+\varepsilon)\sum_{i' \ne i}f_i']
  = \frac{f_i}{f_i + (1+\varepsilon)\sum_{i' \ne i}f_{i'}} \ge \frac{f_i}{(1+\varepsilon)\sum_{i'} f_{i'}}.
        \qedhere\]
\end{proof}
Thus, $\be_{i^*}^{-1}\sum_{j}p_{i^*}(j)$ in addition to capturing a \emph{significant} portion of the interval $\sum_{i,j}\be_i^{-1}p_i(j)$ is also at least a $1+\varepsilon$ factor larger than the second largest value. We use these properties to prove the following theorem:
\begin{theorem}
Given that each server $j$ has a non-negative vector $p(j) \in \R^{n}$, define $q_i \coloneqq \sum_j p_i(j)$. Algorithm~\ref{alg:sampling-general} outputs FAIL only with probability $O(\varepsilon)$ and conditioned on not failing, for all $i \in [n]$
\begin{align*}
    \Pr[\hat{i} = i\ \text{and}\ \hat{q} \in (1 \pm O(\varepsilon))\frac{q_i}{\sum_i q_i}] = (1 \pm O(\varepsilon))\frac{q_i}{\sum_i q_i} \pm 1/\poly(n).
\end{align*}
The algorithm uses only one round and has a total communication of $O(s\polylog(n)/\varepsilon^2)$ words.
\label{thm:additive-sampling}
\end{theorem}
\begin{proof}
All the random variables used in the proof are as defined in the algorithm. The algorithm fails to sample if $\max_i \bX_i \le S/(2C\log^2 n)$ or if $\bX_{(1)} \le (1+\varepsilon/2)\bX_{(2)}$. Let $\calE$ denote the event that $\max_i \be_i^{-1}q_i \ge \sum_i \be_i^{-1}q_i/4C\log^2 n$ and $(\max_i \be_i^{-1}q_i) \ge (1+\varepsilon/4)\cdot \text{second-max}_{i}\, \be_i^{-1}q_i$. We now bound $\Pr[\calE \mid \lnot \text{FAIL}]$. By Bayes' theorem, we have
\begin{align*}
    \Pr[\calE \mid \lnot \text{FAIL}] = \frac{\Pr[\lnot\text{FAIL}\mid \calE]\Pr[\calE]}{\Pr[\lnot\text{FAIL}\mid \calE]\Pr[\calE] + \Pr[\lnot\text{FAIL}\mid \lnot\calE]\Pr[\lnot\calE]}.
\end{align*}
We have $\Pr[\calE] \ge 1/(1+\varepsilon/4) - 1/\poly(n)$ from the above lemma. Let $\calE'$ denote the event that $\max_i \be_i^{-1}q_i \ge \sum_i \be_i^{-1}q_i/C\log^2 n$ and $(\max_i \be_i^{-1}q_i) \ge (1+\varepsilon)\cdot \text{second-max}_{i}\, \be_i^{-1}q_i$. We note that $\calE' \subseteq \calE$ and that $\Pr[\calE'] \ge 1/(1+\varepsilon) - 1/\poly(n)$. Now,
\begin{align*}
    \Pr[\lnot \text{FAIL} \mid \calE] \ge \Pr[\lnot \text{FAIL} \mid \calE']\Pr[\calE' \mid \calE] \ge \Pr[\lnot \text{FAIL} \mid \calE']\Pr[\calE'].
\end{align*}
Condition on $\calE'$ and let $i^* = \argmax \be_i^{-1}q_i$. By a Chernoff bound, since $S \ge O(C^2\log^5 n/\varepsilon^2)$ we get that with probability $\ge 1 - 1/\poly(n)$
\begin{align*}
    \bX_{i^*} \ge S \frac{\be_{i^*}^{-1}q_{i^*}}{\sum_i \be_i^{-1}q_i}(1-\varepsilon/5)> \frac{S}{2C\log^2 n}.
\end{align*} 
Additionally, for all other indices $i$, we get by a Bernstein bound that
\begin{align*}
    \Pr[\bX_i \ge S \frac{\be_i^{-1}q_i}{\sum_i \be_i^{-1}q_i} + \frac{\varepsilon}{8}\frac{S}{C\log^2 n}] \le 1/\poly(n).
\end{align*}
Taking a union bound over all the indices $i \ne i'$, we get that with probability $\ge 1 - 1/\poly(n)$ for all $i \ne i'$,
\begin{align*}
    \bX_i < S \frac{\be_i^{-1}q_i}{\sum_i \be_i^{-1}q_i} + \frac{\varepsilon}{4}\frac{S}{C\log^2 n} \le S\frac{\be_{i^*}^{-1}q_{i^*}}{\sum_{i}\be_{i}^{-1}q_i}(1/(1+\varepsilon) + \varepsilon/8) \le \frac{1-\varepsilon/5}{1+\varepsilon/2}S\frac{\be_{i^*}^{-1}q_{i^*}}{\sum_{i}\be_{i}^{-1}q_i}
\end{align*}
when $\varepsilon$ is a small enough constant. Hence $\Pr[\lnot \text{FAIL} \mid \calE'] \ge 1 - 1/\poly(n)$ and we get
\begin{align*}
    \Pr[\lnot \text{FAIL} \mid \calE] \ge (1 - 1/\poly(n))(1/(1+\varepsilon) - 1/\poly(n)).
\end{align*}
Now similarly, we show that $\Pr[\text{FAIL} \mid \lnot \calE]$ is large. Condition on $\lnot\calE$ and again let $i^* = \max_i \be_i^{-1}q_i$. If $\be_{i^*}^{-1}q_{i^*}/\sum_i \be_i^{-1}q_i \le 1/4C\log^2 n$, then again by Bernstein's bound, we get $\max_i \bX_i < S/(2C\log^2 n)$ with probability $\ge 1 - 1/\poly(n)$. Suppose $\be_{i^*}^{-1}q_{i^*}/\sum_i \be_i^{-1}q_i > 1/4C\log^2 n$ but $\be_{i^*}^{-1}q_{i^*}/\sum_i \be_i^{-1}q_i < (1+\varepsilon/4)\cdot \text{second-max}_{i}\, \be_i^{-1}q_i$. By using a Chernoff bound, we get $\bX_{(1)} < (1+\varepsilon/2)\bX_{(2)}$ with a probability $\ge 1 - 1/\poly(n)$. Hence, $\Pr[\text{FAIL} \mid \lnot \calE] \ge 1 - 1/\poly(n)$. Thus, overall
\begin{align*}
    \Pr[\calE \mid \lnot \text{FAIL}] = 1 - 1/\poly(n).
\end{align*}
Now, conditioned on $\lnot \text{FAIL}$ and $\calE$, with probability $\ge 1 - 1/\poly(n)$, $\hat{i} = i^*$ and 
\begin{align*}
    \bY_{\hat{i}}/S = \bY_{i^*}/S = (1 \pm \varepsilon/2)\frac{\be_{i^*}^{-1}q_{i^*}}{\sum_{i}\be_i^{-1}q_i}
\end{align*}
and hence, $\hat{q} = (1\pm \varepsilon/2)q_{i^*}/\sum_i q_i$. Finally,
\begin{align*}
    \Pr[\hat{i} = i \mid \lnot \text{FAIL}] &= \Pr[\hat{i} = i \mid \lnot \text{FAIL}, \calE]\Pr[\calE \mid \lnot \text{FAIL}] + \Pr[\hat{i} = i \mid \lnot \text{FAIL}, \lnot \calE]\Pr[\lnot \calE \mid \lnot \text{FAIL}]\\
    &= \Pr[i^* = i \mid \calE] \pm 1/\poly(n)\\
    &= \frac{q_i}{\sum_i q_i}(1 \pm O(\varepsilon)) \pm 1/\poly(n),
\end{align*}
where the last inequality is from the previous lemma. 
\end{proof}
\begin{algorithm}
    \caption{Sampling from General Distributions}\label{alg:sampling-general}
    \DontPrintSemicolon
    \KwIn{Each server has a nonnegative vector $p(j) \in \R^{n}$ and a parameter $\varepsilon$}
    \KwOut{Samples approximately from the distribution with the probability of $i$ being $\sum_j p_i(j)/\sum_{i,j} p_{i}(j)$}
    All the servers agree on independent exponential random variables $\be_1,\ldots,\be_n$ using public randomness\;
    $S \gets O(\varepsilon^{-2}\log^5 n)$\;
    For each $j=1,\ldots,s$, the $j$-th server samples $2S$ independent copies of the random variable $\bi$ defined by $\Pr[\bi = i]=\be_i^{-1}p_i(j)/\sum_i \be_i^{-1}p_i(j)$\;
    Each server $j$ communicates $\sum_i p_i(j)$, $\sum_i \be_i^{-1}p_i(j)$ and the $2S$ sampled coordinates to the coordinator\;

    The coordinator, using the communication from the servers, samples $2S$ independent copies of the random variable $(\bj, \bi)$ defined by
    \begin{align*}
        \Pr[(\bj, \bi) = (j,i)] = \frac{\be_i^{-1}p_i(j)}{\sum_{i,j}\be_i^{-1}p_i(j)}
    \end{align*}\;
    $\bX_i \gets $ The number of times $(*, i)$ is sampled in the first $S$ trials\;
    $\bY_i \gets $ The number of times $(*, i)$ is sampled in the second $S$ trials\;
    $\bX_{(1)}, \bX_{(2)} \gets$ top two among $\bX_1,\ldots,\bX_n$\;
    \If{$\bX_{(1)} < S/(2C\log^2 n)$ or $\bX_{(1)} \le (1+\varepsilon/2)\bX_{(2)}$}{\Return{FAIL}}
    $\hat{i} \gets \argmax_i \bY_i$\;
    $\hat{q} \gets \be_{\hat{i}}(\bY_{\hat{i}}/S)\sum_{i,j}\be_i^{-1}p_i(j)/\sum_{i,j}p_i(j)$\;
    \Return{$(\hat{i}, \hat{q})$}
\end{algorithm}
As the above theorem shows, Algorithm~\ref{alg:sampling-general} can sample approximately from very general ``additively-defined'' distributions. As a simple application, we show that it can be used to obtain a leverage score sample from the deduplicated matrix. In this setting, each server $j$ has a matrix $A^{(j)}$ with $d$ columns. Assume that each row $a_i$ of the matrix $A^{(j)}$ is associated with a tag $t_i$. Assume that across all the servers, a tag $t$ is always associated with the same row. Let $A$ be the matrix formed by the rows that correspond to distinct tags across the servers. For example, suppose a pair of row and tag $(a, t)$ is present in the servers $1$, $2$ and $3$. In this deduplicated case, the matrix $A$ has only one copy of the row $a$. Suppose server $1$ has a pair $(a, 1)$ and server $2$ has a pair $(a, 2)$ i.e., the same row is associated with different tags at different servers, then the matrix $A$ has both the copies of the row. We want to sample rows of $A$ from the associated $\ell_2$ leverage score distribution. For a row $a$ of the matrix $A$, the leverage score of $a$ with respect to $A$ is defined as
\begin{align*}
    \tau_A(a)\coloneqq \max_x \frac{|\la a, x\ra|^2}{\opnorm{Ax}^2}.
\end{align*}
If the row $a$ is not in the matrix $A$, we define $\tau_A(a) = 0$. An important property of the leverage score distribution is that $\sum_{a \in A}\ell_A(a) = \text{rank}(A) \le d$ since $A$ has $d$ columns.
Given a tag $t$, let $a_t$ be the row associated with the tag $t$. As the matrix $A$ is formed by the ``union'' of all the matrices $A^{(1)}, \ldots, A^{(t)}$, for any $x$ we have  $\opnorm{Ax}^2 \ge \opnorm{A^{(j)}x}^2$. Hence if $a$ is a row of the matrix $A^{(j)}$, then $\tau_{A^{(j)}}(a) \ge \tau_{A}(a)$. Let the tags be drawn from a set $T$. Each server $j$ can now define a $|T|$-dimensional vector $p(j)$ such that for each $t \in T$
\begin{align*}
    p_t(j) = \tau_{A^{(j)}}(a_t). 
\end{align*}
Here we note that if the tag $t$ is not present in the matrix $A^{(j)}$, we have $p_t(j) = 0$. We now have for any tag $t$,
\begin{align*}
    \frac{\sum_{j=1}^s p_t(j)}{\sum_{j=1}^s\sum_{t \in T}p_t(j)} = \frac{\sum_{j=1}^s \tau_{A^{(j)}}(a_t)}{\sum_{j=1}^s \sum_{t\in T}\tau_{A^{(j)}}(a_t)} \ge \frac{\tau_{A}(a_t)}{\sum_{j=1}^s \text{rank}(A^{(j)})} \ge \frac{\tau_{A}(a_t)}{s \cdot \text{rank}(A)}.
\end{align*}
Here, we used the monotonicity of leverage scores which directly follows from the definition. Such a sampling distribution is an $s$-approximate leverage score distribution and has numerous applications such as computing subspace embeddings \cite{w14}. This application shows that Algorithm~\ref{alg:sampling-general} can sample from useful distributions and importantly give approximations to the probabilities of the samples using only one round of communication.

While this very general sampling procedure can be used to construct a subspace embedding for the deduplicated matrix in the coordinator model, it is unclear how to extend these algorithms to more general topologies to obtain algorithms in the personalized CONGEST model in which we want the protocols to have additional nice properties as discussed in the introduction. In later sections, we define the concept of composable sketches and how they are useful to obtain protocols in the personalized CONGEST model.


%% file: new_two_round.tex
\section{A Two Round Protocol for Sum Approximation}\label{sec:new_alg}
Recall that we say a non-negative function $f$ satisfies ``approximate invertibility'' with parameters $\theta, \theta', \theta'' > 1$ if the following hold:
\begin{enumerate}
    \item for all $x_1, x_2 \ge 0$, it holds that $f(x_1) + f(x_2) \le f(x_1 + x_2)$ which additionally implies that $f(0) = 0$, 
    \item for all $x$, $f(\theta' x) \ge \theta f(x)$, and
    \item for all $x$, $f(x/(4 \cdot \sqrt{\theta} \cdot \theta')) \ge f(x)/\theta''$.
\end{enumerate}
Note that plugging $x = f^{-1}(z)$ in the second property above, we get $\theta' f^{-1}(z) \ge f^{-1}(\theta z)$.
We now define $\varepsilon_1 \coloneqq 1/\theta''$ and $\varepsilon_2 \coloneqq 1 - 1/\theta''$. We can show that for all values of $0 \le x_2 \le x_1$, if $f(x_2) \le \varepsilon_1 f(x_1)$, then $f(x_1 - x_2) \ge (1-\varepsilon_2)f(x_1)$. This essentially shows that if $f(x_2)$ is very small as compared to $f(x_1)$, then $f(x_1 - x_2)$ can not be very small when compared to $f(x_1)$. These properties make the function $f$ ``approximately invertible'', meaning that good approximations for $f(x)$ will let us approximate the preimage $x$ as well. We say that a function $f$ that satisfies the above properties is ``approximately invertible'' with parameters $\theta, \theta', \theta''$.

Note that for an integer $s \ge 1$, we defined $c_{f}[s]$ to be the smallest number such that for all $x_1, \ldots, x_s \ge 0$,
\begin{align*}
    f(x_1 + \cdots + x_s) \le \frac{c_f[s]}{s}\left(\sqrt{f(x_1)} + \cdots + \sqrt{f(x_s)}\right)^2.
\end{align*}
Using the Cauchy-Schwarz inequality, we additionally have 
\[f(x_1 + \cdots + x_s) \le c_f[s]\left(f(x_1) + \cdots + f(x_s)\right)\] 
for all $x_1, \ldots, x_s \ge 0$.
We show that the parameter $c_f[s]$ as a function of $s$ cannot grow arbitrarily. Consider arbitrary integers $s, t \ge 1$. The following lemma shows that $c_f[s \cdot t] \le c_f[s] \cdot c_f[t]$ which implies that the function $c_f[s]$ is upper bounded by a polynomial in $s$ with a degree that depends only on $c_f[2]$.
\begin{lemma}
    For any $x_1, x_2, \ldots, x_{s \cdot t} \ge 0$, 
    \begin{align*}
        f(x_{1} + \cdots + x_{s \cdot t}) \le \frac{c_{f}[s] \cdot c_{f}[t]}{s \cdot t}\left(\sqrt{f(x_1)} + \cdots + \sqrt{f(x_{s \cdot t})}\right)^2.
    \end{align*}
\end{lemma}
\begin{proof}
    By definition of $c_f[s]$, we obtain
    \begin{align*}
        &f(x_1 + \cdots + x_{s \cdot t})\\
        &\le \frac{c_f[s]}{s}\left(\sqrt{f(x_1 + \cdots + x_t)} + \sqrt{f(x_{t+1}) + \cdots + f(x_{2 \cdot t})} + \cdots + \sqrt{f(x_{(s-1)\cdot t + 1} + \cdots + f(x_{s \cdot t}))}\right)^2.
    \end{align*}
    Now we use $c_f[t]$ to expand the internal terms, to get
    \begin{align*}
        f(x_1 + \cdots + x_{s \cdot t}) \le \frac{c_f[s]}{s}\frac{c_f[t]}{t}\left(\sqrt{f(x_1)} + \cdots + \sqrt{f(x_{s \cdot t})}\right)^2.
    \end{align*}
    By definition of $c_f[s \cdot t]$, we obtain $c_f[s \cdot t] \le c_f[s] \cdot c_f[t]$.
\end{proof}
The above lemma upper bounds the growth of the parameter $c_f[s]$. We now lower bound the growth and show that $c_f[s]$ must grow at least linearly in the parameter $s$.
\begin{lemma}
    If $s \ge t$, then $c_f[s] \ge c_f[t] \cdot s / t$.
\end{lemma}
\begin{proof}
    Let $x_1, \ldots, x_t \ge 0$ be arbitrary.
\begin{align*}
    f(x_1 + \cdots + x_t) &= f(x_1 + \cdots + x_t + \underbrace{0 + \cdots + 0}_{s - t})\\
    &\le \frac{c_f[s]}{s}\left(\sqrt{f(x_1)} + \cdots + \sqrt{f(x_t)} + \underbrace{\sqrt{f(0)} + \cdots + \sqrt{f(0)}}_{s-t}\right)^2\\
    &\le \frac{c_f[s]}{s}(\sqrt{f(x_1)} + \cdots + \sqrt{f(x_t)})^2
\end{align*}
where we used $f(0) = 0$ in the last inequality.
Hence by definition of $c_f[t]$, we get $c_f[t]/t \le c_f[s]/s$ from which we obtain that $c_f[s] \ge c_f[t] \cdot s/t$.
\end{proof}
Using the above properties, we obtain that $c_f[s] \le c_f[2^{\ceil{\log_2(s)}}] \le (c_f[2])^{\ceil{\log_2(s)}} \le c_f[2] \cdot s^{\log_2 c_f[2]}$ which shows that $c_f$ only grows at most \emph{polynomially} in the number of servers $s$ and the degree of growth is upper bounded $\log_2 c_f[2]$. For example, if $f(x) = x^k$, then we can show $c_f[2] = 2^{k-1}$ from which we obtain that $c_f[s] \le (2s)^{k-1}$ for all the values of $s$. 

In addition, if $c_f[t] \ge c_f[s]/\alpha$ for some value $\alpha$, using the second property above, we obtain $c_f[s] \ge c_f[t] \cdot s/t \ge (c_f[s]/\alpha) \cdot (s/t)$ which implies $t \ge s/\alpha$. Thus if $c_f[t]$ is ``comparable'' to $c_f[s]$, then $t$ is ``comparable'' to $s$ as well. This is a property that we critically use in the analysis of our algorithm.
\subsection*{Protocol Description and Analysis}
We now recall the overview of the algorithm we presented in the introduction and define some notation that we use throughout our analysis. All the servers together with the coordinator use the shared randomness and sample standard exponential random variables $\be_1, \ldots, \be_n$ and the goal of the coordinator is to compute
\begin{align*}
    \max_i \be_i^{-1}f(x_i).
\end{align*}
Define $i^* \coloneqq \argmax_i \be_i^{-1}f(x_i)$. We have seen that the median of $O(1/\varepsilon^2)$ independent copies of the random variable $\max_i \be_i^{-1}f(x_i)$ can be used to compute a $1 \pm \varepsilon$ approximation to $\sum_i f(x_i)$ with high probability. Throughout the analysis, we condition on the event that
\begin{align*}
    \sum_i \be_i^{-1}f(x_i) \le (C \log^2 n) \cdot \max_i \be_i^{-1}f(x_i).
\end{align*}
We have seen that this event holds with probability $\ge 1 - 1/\poly(n)$ over the exponential random variables. Now fix a server $j \in [s]$. The server $j$ samples $N$ independent copies of the random variable $\bi$ supported in the set $[n]$ and has a distribution defined as
\begin{align*}
    \Pr[\bi = i] = \frac{\be_i^{-1}f(x_i)}{\sum_i \be_i^{-1}f(x_i)}.
\end{align*}
Let $\SC_j$ denote the set of coordinates that are sampled by the server $j$. The server $j$ then sends the set $\SC_j$ along with (i) the sum $\sum_i \be_i^{-1}f(x_i(j))$ and (ii) the values $x_i(j)$ for $i \in \SC_j$ to the coordinator using $O(N)$ words of communication.

The coordinator defines the set $\SC \coloneqq \bigcup_j \SC_j$ as the union of the sets of coordinates that are sampled at different servers. We will first show that the coordinate $i^* \in \SC$ with a large probability if $N$ is large enough.
\subsection{The coordinate \texorpdfstring{$i^*$}{i-star} is sampled}
\begin{lemma}
If $N \ge (C_{\ref{lma:i-star-is-sampled}}\log^3 n) \cdot c_f[s]/s$ for a universal constant $C_{\ref{lma:i-star-is-sampled}}$ large enough, then the coordinate $i^* \in \SC$ with a probability $\ge 1 - 1/n^4$.
    \label{lma:i-star-is-sampled}
\end{lemma}
\begin{proof}
Let $q_{i^*}(j) = \be_{i^*}^{-1}f(x_{i^*}(j))/\sum_{i \in [n]}\be_i^{-1}f(x_i(j))$ be the probability that the random variable $\bi = i^*$ (note that the random variable $\bi$ has a different distribution at each server) and $p_{i^*}(j)$ be the probability that $i^* \in \SC_j$. We have
\begin{align*}
    p_{i^*}(j) = 1 - \left(1 - q_{i^*}(j)\right)^N.
\end{align*}
If $q_{i^*}(j) = (4\log n)/N$, then $p_{i^*}(j) \ge 1 - (1 - q_{i^*}(j))^N \ge 1 - \exp(-q_{i^*}(j)N) \ge 1 - \exp(-4\log n) \ge 1 - 1/n^4$ which implies that $i^* \in \SC_j \subseteq \SC$ with probability $\ge 1 - 1/n^4$ and we are done. Otherwise, 
\begin{align}
    p_{i^*}(j) = 1 - (1 - q_{i^*}(j))^N \ge \frac{Nq_{i^*}(j)}{4\log n}(1 - 1/n^4)
    \label{eqn:lower-bound-p-j}
\end{align}
using the concavity of the function $1 - (1-x)^N$ in the interval $[0, 4\log n/N]$. 

Since each of the servers samples the coordinates independently, the probability that $i^* \in \SC$ is $1 - (1-p_{i^*}(1))(1-p_{i^*}(2)) \cdots (1-p_{i^*}(s)) \ge 1 - \exp(-\sum_j p_{i^*}(j))$. So, showing that the sum $\sum_j p_{i^*}(j)$ is large implies that $i^*$ is in the set $\SC$ with a large probability. Assume that for all $j$, $q_{i^*}(j) < 4\log n/N$ since otherwise we already have that the coordinate $i^* \
\in \SC$ with probability $\ge 1 - 1/n^4$. Now using \eqref{eqn:lower-bound-p-j}
\begin{align*}
 \sum_j p_{i^*}(j) \ge \frac{N}{8\log n}\sum_j q_{i^*}(j) = \frac{N}{8\log n}\sum_j \frac{\be_{i^*}^{-1}f(x_{i^*}(j))}{\sum_i \be_i^{-1}f(x_i(j))}. 
\end{align*}
We will now give a lower bound on the the sum in the above equation using the following simple lemma and the definition of the parameter $c_f[s]$.
\begin{lemma}
    Let $a_1, \ldots, a_s \ge 0$ and $b_1, \ldots, b_s > 0$ be arbitrary. Then,
    \begin{align*}
        \sum_{j \in [s]} \frac{a_j}{b_j} \ge \frac{(\sum_{j \in [s]}\sqrt{a_j})^2}{\sum_j b_j}.
    \end{align*}
\end{lemma}
\begin{proof}
    We prove the lemma by using the Cauchy-Schwarz inequality. Since we assume that $b_j > 0$ for all $j$, we can write
    \begin{align*}
        \sum_{j \in [s]}\sqrt{a_j} = \sum_{j \in [s]}\sqrt{\frac{a_j}{b_j}}\sqrt{b_j}.
    \end{align*}
    Using the Cauchy-Schwarz inequality, we get
    \begin{align*}
        \sum_{j\in [s]}\sqrt{a_j} = \sum_{j \in [s]}\sqrt{\frac{a_j}{b_j}}\sqrt{b_j} \le \sqrt{\sum_{j \in [s]}\frac{a_j}{b_j}}\sqrt{\sum_{j \in [s]}b_j}.
    \end{align*}
    Squaring both sides and using the fact that $\sum_{j \in [s]}b_j > 0$, we get
    \begin{align*}
        \sum_{j \in [s]}\frac{a_j}{b_j} \ge \frac{(\sum_{j \in [s]}\sqrt{a_j})^2}{\sum_{j \in [s]}b_j}. &\qedhere
    \end{align*}
\end{proof}
Letting $a_j = \be_{i^*}^{-1}f(x_{i^*}(j))$ and $b_j = \sum_i \be_i^{-1}f(x_i(j))$ in the above lemma, we get
\begin{align*}
    \sum_j \frac{\be_{i^*}^{-1}f(x_{i^*}(j))}{\sum_i \be_i^{-1}f(x_i(j))} \ge \frac{\be_{i^*}^{-1}(\sum_{j \in [s]}\sqrt{f(x_{i^*}(j))})^2}{\sum_j\sum_i \be_{i}^{-1}f(x_i(j))}. 
\end{align*}
Using the definition of $c_f[s]$, we obtain
\begin{align*}
    \left(\sum_{j \in [s]}\sqrt{f(x_{i^*}(j))}\right)^2 \ge \frac{s}{c_f[s]}(f(x_{i^*}(1) + \cdots x_{i^*}(s))) = \frac{s}{c_f[s]}f(x_{i^*}).
\end{align*}
Using the super-additivity of the function $f$, we get
\begin{align*}
    \sum_j \sum_i \be_i^{-1}f(x_i(j)) = \sum_i \be_i^{-1}\sum_j f(x_i(j)) \le \sum_i \be_i^{-1}f(\sum_j x_i(j)) = \sum_i \be_i^{-1}f(x_i). 
\end{align*}
Since we conditioned on the event that $\sum_i \be_i^{-1}f(x_i) \le (C\log^2 n)\cdot\be_{i^*}^{-1}f(x_{i^*})$, we get
\begin{align*}
    \sum_j\sum_i \be_i^{-1}f(x_i(j)) \le (C\log^2 n)\cdot\be_{i^*}^{-1}f(x_{i^*}). 
\end{align*}
We therefore have
\begin{align*}
    \sum_j \frac{\be_{i^*}^{-1}f(x_{i^*}(j))}{\sum_i \be_{i}^{-1}f(x_i(j))} \ge \frac{s}{c_f[s]}\frac{\be_{i^*}^{-1}f(x_{i^*})}{(C\log^2 n) \cdot \be_{i^*}^{-1}f(x_{i^*})} \ge \frac{s}{C\log^2 n \cdot c_f[s]}. 
\end{align*}
Thus, if $N \ge (32C\log^3 n) \cdot c_f[s]/s$, then $\sum_j p_{i^*}(j) \ge 4\log n$ which implies that $i^*$ is in the set $\SC$ with probability $\ge 1 - 1/n^4$. Letting $C_{\ref{lma:i-star-is-sampled}} \coloneqq 32 C$, we have the proof.
\end{proof}
When $N \gg c_f[s] \log^3 n/s$ as is required by the above lemma, the set $\SC = \bigcup_j \SC_j$ may have a size of $\Omega(c_f[s ]\log^3 n)$ which is quite large. Conditioned on the event that $i^* \in \SC$, we now want to compute a small subset $\PL \subseteq \SC$ such that $i^* \in \PL$ with a large probability.

\subsection{Computing the set \texorpdfstring{$\PL$}{PL}}
Recall we condition on the event that
\begin{align*}
    \sum_i \be_i^{-1}f(x_i) \le (C \log^2 n) \cdot \be_{i^*}^{-1}f(x_{i^*})
\end{align*}
and that $i^* \in \SC$. As we noted in the introduction, we proceed by constructing an estimator $\hat{x}_i$ for each $i \in \SC$ that satisfies the following properties with a probability $1 - 1/\poly(n)$:
\begin{enumerate}
    \item For all $i \in \SC$, $\hat{x}_i \le x_i$ and
    \item $\be_{i^*}^{-1}f(\hat{x}_{i^*}) \ge \alpha \cdot \be_{i^*}^{-1}f(x_{i^*})$ for some $\alpha < 1$.
\end{enumerate}

Fix an index $i \in [n]$. In the remaining part of this section, we will describe a quantity that the coordinator can approximate to obtain $\hat{x}_i$ which satisfies the above properties. 
\subsubsection{Contribution from \textLARGE{} servers}
We define $\textLARGE_i$ to be the set of servers $j$ such that
\begin{align*}
    q_i(j) = \frac{\be_i^{-1}f(x_i(j))}{\sum_i \be_i^{-1}f(x_i(j))} \ge \frac{4\log n}{(c_f[s]\log^3 n)/s}. 
\end{align*}
Note if $i \in \SC_j$, then the coordinator can determine if $j \in \textLARGE_i$ since it has access to both the values $\be_i^{-1}f(x_i(j))$ and $\sum_i \be_i^{-1}f(x_i(j))$. We have the following lemma:
\begin{lemma}
    If $N \ge (c_f[s]\log^3 n)/s$, then with probability $\ge 1 - 1/n^3$, for all the coordinates $i$ and servers $j\in\textLARGE_i$, we have $i \in \SC_j$. In other words, with a large probability, all the coordinates $i$ are sampled at all the servers that are in the set $\textLARGE_i$.
    \label{lma:all-large-are-sampled}
\end{lemma}
\begin{proof}
    Recall that $\SC_j$ denotes the set of coordinates that are sampled at server $j$. Consider a fixed $i \in [n]$. If $j \in \textLARGE_i$, then the probability that $i$ is sampled at server $j$ is $1 - (1 - q_i(j))^N \ge 1 - (1 - 4\log n/(c_f[s]\log^3 n/s))^N \ge 1 - \exp(-4\log n) \ge 1 - 1/n^4$ using $N \ge (c_f[s]\log^3 n)/s$. By a union bound over all the servers $j$ that are \textLARGE{} for $i$, we have the proof.
\end{proof}
Thus with a large probability, for each $i \in [n]$, the contribution to $x_i = \sum_j x_i(j)$ from servers $j \in \textLARGE_i$ can be computed exactly by the coordinator after it receives the samples from the servers. 
\subsubsection{Contribution from \textSMALL{} servers}
We now define $\textSMALL_i$ to be the set of servers $j$ for which
\begin{align*}
    q_i(j) = \frac{\be_i^{-1}f(x_i(j))}{\sum_i \be_i^{-1}f(x_i(j))} \le \frac{\varepsilon_1}{c_f[s] \cdot (C\log^2 n)}.
\end{align*}
We have
\begin{align*}
    \sum_{j \in \textSMALL_{i}}\be_i^{-1}f(x_i(j)) \le \sum_{j \in \textSMALL_{i}}\frac{\varepsilon_1\sum_{i'}\be_{i'}^{-1}f(x_{i'}(j))}{c_f[s] \cdot (C\log^2 n)} \le \frac{\varepsilon_1 \cdot (C\log^2 n) \cdot  (\be_{i^*}^{-1}f(x_{i^*}))}{c_f[s] \cdot (C\log^2 n)} = \frac{\varepsilon_1 \cdot (\be_{i^*}^{-1}f(x_{i^*}))}{c_f[s]}
\end{align*}
where we used the fact that \[\sum_{j \in \textSMALL_i}\sum_{i'}\be_{i'}^{-1}f(x_{i'}(j)) = \sum_{i'}\sum_{j \in \textSMALL_i}\be_{i'}^{-1}f(x_{i'}(j)) \le \sum_{i'}\be_{i'}^{-1}f(x_{i'}) \le (C\log^2 n) \cdot \be_{i^*}^{-1}f(x_{i^*}).\]
By definition of the parameter $c_f[s]$, we then obtain that
\begin{align*}
    \be_i^{-1}f(\sum_{j \in \textSMALL_{i}{}}x_i(j)) \le c_f[s] \cdot \be_i^{-1}\sum_{j \in \textSMALL_{i}{}}f(x_i(j)) \le \varepsilon_1 \cdot \be_{i^*}^{-1}f(x_{i^*})
\end{align*}
which then implies
\begin{align}
    \be_{i^*}^{-1}f(x_{i^*} - \sum_{j \in \textSMALL_{i^*}{}}x_{i^*}(j)) \ge (1-\varepsilon_2) \cdot \be_{i^*}^{-1}f(x_{i^*}).
    \label{eqn:ignoring-small}
\end{align}
Hence we can ignore the contribution of the servers in $\textSMALL_i$ when computing $\hat{x}_i$.
\subsubsection{Contribution from Remaining Servers}
From the above, we have for each coordinate $i \in [n]$, we can compute the contribution of servers $j \in \textLARGE_i$ exactly and ignore the contribution of servers in $\textSMALL_i$  while still being able to satisfy the required properties for $\hat{x}_i$. We will now show how to estimate the contribution of servers $j$ that are neither in $\textLARGE_i$ nor in $\textSMALL_i$. Note that for such servers, the value $q_i(j) = \be_i^{-1}f(x_i(j))/\sum_i \be_i^{-1}f(x_i(j))$ lies in the interval
\begin{align*}
    \left[\frac{\varepsilon_1}{c_f[s] \cdot (C\log^2 n)}, \frac{4s}{c_f[s]\log^2 n}\right].
\end{align*}
We now partition\footnote{We do not require it and so are not too careful about ensuring that the intervals we use are disjoint.} the above interval into intervals with lengths geometrically increasing by a factor of $\sqrt{\theta}$. Let $P_{\start} = \frac{\varepsilon_1}{c_f[s] \cdot (C\log^2 n)}$ and we partition the above interval into intervals $[P_{\start}, \sqrt{\theta}P_{\start}], [\sqrt{\theta}P_{\start},(\sqrt{\theta})^2P_{\start}], \ldots,$ and so on. We note that there are at most 
\begin{align}
 A = O\left(\frac{\log (s/\varepsilon_1)}{\log \theta}\right)
\end{align}
such intervals in the partition.

Let $I^{(a)}_i$ denote the set of servers $j$ such that $q_i(j) \in [(\sqrt{\theta})^{a}P_{\start}, (\sqrt{\theta})^{a+1}P_{\start}]$. If $|I^{(a)}_i|$ is large enough, then the number of servers $j$ in $I^{(a)}_i$ at which the coordinate $i$ is sampled is ``concentrated'' which can then be used to estimate $|I^{(a)}_i|$. But observe that even having an estimate of $|I^{(a)}_i|$ is insufficient since we cannot directly estimate $\sum_{j \in I^{(a)}}x_i(j)$ only given $|I^{(a)}_i|$ as the servers in the set $I^{(a)}_i$ may have quite different values for $\sum_i \be_i^{-1}f(x_i(j))$. So we further partition the servers based on the value of $\sum_i \be_i^{-1}f(x_i(j))$. We first give a lower bound on the values $\sum_i \be_i^{-1}f(x_i(j))$ that we need to consider.
\begin{lemma}
If the server $j \notin \textSMALL_i \cup \textLARGE_i$ and $\sum_i \be_i^{-1}f(x_i(j)) \le \frac{\varepsilon_1(1-\varepsilon_2)\sum_{i, j}\be_i^{-1}f(x_i(j))}{4Cs^2}$, then the contribution from all those servers can be ignored.
\end{lemma}
\begin{proof}
    Let $\Ignore_i$ be the set of servers $j$ that are neither in $\textLARGE_i$ nor in $\textSMALL_i$ and have $\sum_i \be_i^{-1}f(x_i(j)) \le \frac{\sum_{i, j}\varepsilon_1(1-\varepsilon_2)\be_i^{-1}f(x_i(j))}{4Cs^2}$. By definition, we have
    \begin{align*}
        \sum_{j \in \Ignore_i} \be_i^{-1}f(x_i(j)) &\le \sum_{j \in \Ignore_i}\frac{4s}{c_f[s]\log^2 n}\sum_i \be_i^{-1}f(x_i(j))\\
        &\le \frac{4s \cdot |\Ignore_i|}{c_f[s]\log^2 n} \frac{\varepsilon_1(1-\varepsilon_2)\sum_{i, j}\be_i^{-1}f(x_i(j))}{4Cs^2}.
    \end{align*} 
    Using the definition of $c_f[s]$ and the Cauchy-Schwarz inequality, we then obtain
    \begin{align*}
        \be_i^{-1}f(\sum_{j \in \Ignore_i}x_i(j)) &\le \frac{c_{f}[s]}{s} \cdot |\Ignore_i| \cdot \sum_{j \in \Ignore_i}\be_i^{-1}f(x_i(j))\\
        &\le \frac{4|\Ignore_i|^2}{\log^2 n}\frac{\varepsilon_1(1-\varepsilon_2)\sum_{i, j}\be_i^{-1}f(x_i(j))}{4Cs^2}.
    \end{align*}
    Since $\sum_{i, j}\be_i^{-1}f(x_i(j)) \le (C\log^2 n) \cdot \be_{i^*}^{-1}f(x_i^*)$ and $|\Ignore_{i}| \le s$, we get
    \begin{align*}
        \be_i^{-1}f(\sum_{j \in \Ignore_i}x_i(j)) \le \varepsilon_1(1-\varepsilon_2)\be_{i^{*}}^{-1}f(x_i^*). 
    \end{align*}
    Taking $i = i^*$, we get $\be_{i^*}^{-1}f(\sum_{j \in \Ignore_{i^*}}x_{i^*}(j)) \le \varepsilon_1(1-\varepsilon_2)\be_{i^*}^{-1}f(x_{i^*})$ and therefore using \eqref{eqn:ignoring-small} we obtain that
    \begin{align*}
        \be_{i^*}^{-1}f(\sum_{j \in \Ignore_{i^*}}x_{i^*}(j)) \le \be_{i^*}^{-1}\varepsilon_1 f(x_{i^*} - \sum_{j \in \textSMALL_{i^*}}x_{i^*}(j))
    \end{align*}
    from which we then get
    \begin{align*}
        \be_{i^*}^{-1}f(x_{i^*}-\sum_{j \in \textSMALL_{i^*}}x_{i^*}(j) - \sum_{j \in \Ignore_{i^*}}x_{i^*}(j)) \ge (1-\varepsilon_2)^2 \be_{i^*}^{-1}f(x_{i^*}).
    \end{align*}
    Thus the contribution from the servers in the set $\Ignore_{i^*}$ can be ignored as we can make $\be_{i^*}^{-1}f(\hat{x}_{i^*})$ large even without them.
\end{proof}
So we only have to focus on the servers for which the quantity $\sum_i \be_i^{-1}f(x_i(j))$ lies in the interval
\begin{align*}
    \left[\frac{\varepsilon_1(1-\varepsilon_2)}{4Cs^2}\sum_{i,j}\be_i^{-1}f(x_i(j)), \max_j \sum_i \be_i^{-1}f(x_i(j))\right].
\end{align*}
Now define $F_{\start} \coloneqq \frac{\varepsilon_1(1-\varepsilon_2)}{4Cs^2}\sum_{i, j}\be_i^{-1}f(x_i(j))$ and  partition the above interval into intervals $[F_{\start}, (\sqrt{\theta}F_{\start})]$, $[(\sqrt{\theta})F_{\start}, (\sqrt{\theta})^2F_{\start}]$, $\ldots$, and so on. Note that there are at most 
\begin{align}
    B = O\left(\frac{\log(s^2/\varepsilon_1(1-\varepsilon_2))}{\log \theta}\right)
\end{align}
such intervals. Let $I^{(a, b)}_i$ for $a = 0, 1, \ldots, A-1$ and $b = 0, 1, \ldots, B-1$ be the set of servers $j$ for which
\begin{align*}
    \frac{\be_i^{-1}f(x_i(j))}{\sum_{i}\be_i^{-1}f(x_i(j))} \in [(\sqrt{\theta})^{a}P_{\start}, (\sqrt{\theta})^{a+1}P_{\start}],
\end{align*}
and
\begin{align*}
    \sum_i \be_i^{-1}f(x_i(j)) \in [(\sqrt{\theta})^{b} F_{\start}, (\sqrt{\theta})^{b+1}F_{\start}].
\end{align*}

If the set $|I^{(a, b)}_i|$ is large, then the number of servers $j \in I^{(a, b)}_i$ at which $i$ is sampled is concentrated which can then in turn be used to approximate $|I^{(a, b)}_i|$. But if $|I^{(a, b)}_i|$ is too small, then we cannot obtain a good approximation using the number of servers in the set $I^{(a, b)}_i$ that sample $i$. In the following lemma, we show that the contribution from servers in the sets $I^{(a, b)}_i$ needs to be considered only if $|I^{(a, b)}_i|$ is large. Let $\textBAD_i$ denote the set of tuples $(a, b)$ for which
\begin{align}
        c_f[|I^{(a, b)}_i|] \le \frac{c_f[s]}{c_f[A \cdot B] \cdot  (A \cdot B) \cdot (\sqrt{\theta})^{a+1}/(1-\varepsilon_2)^2}. \label{eqn:bad-i-defn}  
\end{align}
Let $\textGOOD_i$ be the set of all the remaining tuples $(a, b)$. We first note that the coordinator cannot determine if a particular tuple $(a, b)$ is in the set $\textBAD_i$. Using the properties of $c_f[s]$, we get that if $(a, b) \in \textGOOD_i$, then
\begin{align*}
    |I_{i}^{(a, b)}| \ge \frac{s}{c_f[A \cdot B] \cdot (A \cdot B) \cdot (\sqrt{\theta})^{a+1}/(1-\varepsilon_2)^2}.
\end{align*}
\begin{lemma}
The contribution from servers $j \in \bigcup_{(a, b) \in \textBAD_i} I^{(a, b)}_i$ can be ignored.
\label{lma:bad-can-be-ignored}
\end{lemma}
\begin{proof}
By definition of the set of servers $I^{(a, b)}_i$, 
\begin{align*}
    \sum_{j \in I^{(a, b)}_i}\be_i^{-1}f(x_i(j)) &\le \sum_{j \in I_i^{(a, b)}}(\sqrt{\theta})^{a+1}P_{\start}\sum_{i' \in [n]}\be_{i'}^{-1}f(x_{i'}(j))\\
    &\le (\sqrt{\theta})^{a+1}P_{\start}\sum_{i' \in [n]}\be_{i'}^{-1}\sum_{j \in I_i^{(a, b)}}f(x_{i'}(j))\\
    &\le (\sqrt{\theta})^{a+1}P_{\start}  \sum_{i' \in [n]}\be_{i'}^{-1}f(x_{i'})\\
    &\le (\sqrt{\theta})^{a+1}P_{\start} \cdot (C\log^2 n) \cdot \be_{i^*}^{-1}f(x_{i^*}).
\end{align*}
By definition of the parameter $c_f$, we have
    \begin{align*}
        &\be_i^{-1}f\left(\sum_{(a, b) \in \textBAD_i}\sum_{j \in I^{(a, b)}_i}x_i(j)\right)\\
        &\le c_f[A \cdot B]\left(\sum_{(a, b) \in \textBAD_i} \be_i^{-1}f(\sum_{j \in I^{(a, b)}_i}x_i(j))\right)\\
        &\le c_f[A \cdot B] \sum_{(a, b) \in \textBAD_i} c_f[|I^{(a, b)}_i|]\sum_{j \in I^{(a, b)}_i}\be_i^{-1}f(x_i(j))\\
        &\le c_f[A \cdot B]\sum_{(a, b) \in \textBAD_i}c_f[|I^{(a, b)}_i|] \cdot (\sqrt{\theta})^{a+1}P_{\start} \cdot (C\log^2 n) \cdot \be_{i^*}^{-1}f(x_{i^*})\\
        &\le c_f[A \cdot B] \sum_{(a, b) \in \textBAD_i}c_f[|I^{(a, b)}_i|] \cdot (\sqrt{\theta})^{a+1} \cdot \frac{\varepsilon_1}{c_f[s] \cdot (C\log^2 n)} \cdot (C\log^2 n) \cdot \be_{i^*}^{-1}f(x_{i^*}).
    \end{align*}
    Using \eqref{eqn:bad-i-defn}, we get
    \begin{align*}
        \be_i^{-1}f\left(\sum_{(a, b) \in \textBAD_i}\sum_{j \in I_i^{(a, b)}}x_i(j)\right) \le (1-\varepsilon_2)^2\varepsilon_1 \cdot \be_{i^*}^{-1}f(x_{i^*}).
    \end{align*}
    Taking $i = i^*$, we get that
    \begin{align*}
        \be_{i^*}^{-1}f(x_{i^*} - \sum_{j \in \textSMALL_{i^*}}x_{i^*}(j) - \sum_{j \in \Ignore_{i^*}}x_{i^*}(j) - \sum_{(a, b) \in \textBAD_i}\sum_{j \in I^{(a, b)}_{i^*}}x_{i^*}(j)) \ge (1-\varepsilon_2)^3 \cdot \be_{i^*}^{-1}f(x_{i^*}).
    \end{align*}
    Thus the contribution from the servers in the set $I^{(a, b)}_i$ for tuples $(a, b) \in \textBAD_i$ can be ignored.
\end{proof}
Now we show that the coordinator can compute $\hat{x}_i$ by essentially approximating the following quantity:
\begin{align*}
    &x_{i} - \sum_{j \in \textSMALL_{i}}x_{i}(j) - \sum_{j \in \Ignore_{i}}x_{i}(j) - \sum_{(a, b) \in \textBAD_i}\sum_{j \in I^{(a, b)}_{i}}x_{i}(j)\\
    &= \sum_{j \in \textLARGE_i}x_i(j) + \sum_{(a, b) \in \textGOOD_i}\sum_{j \in I_{i}^{(a, b)}}x_i(j).
\end{align*}
\subsubsection{Algorithm to compute \texorpdfstring{$\hat{x}_i$}{x-hat-i}}
We will go on to give an algorithm that can approximate $\hat{x}_i$ given the samples that the coordinator receives in the first round. We have already seen that given an index $i$, all the servers in the set $\textLARGE_i$ sample the coordinate $i$ with a large probability and therefore we can compute the quantity $\sum_{j \in \textLARGE_i}x_i(j)$ exactly with high probability. We have also seen that the contribution from servers that are in the set $\textSMALL_i$ and in the set $\Ignore_i$ can be ignored.

The main remaining contribution to $x_i$ that is to be accounted is from tuples $(a, b) \in \textGOOD_i$. An important issue we need to solve for is the fact that the coordinator cannot determine if a given tuple $(a, b)$ is in the set $\textGOOD_i$ or in $\textBAD_i$. We will show that if $(a, b) \in \textGOOD_i$, then $i$ is sampled at many servers in the set $I^{(a, b)}_i$ and depending on the \emph{absolute} number of servers in $I^{(a, b)}_i$ that sample $i$, we mark a tuple $(a, b)$ as ``probably good for $i$''. We argue that, with high probability, all tuples $(a, b) \in \textGOOD_i$ are marked ``probably good for $i$'' and that for the tuples $(a, b)$ in $\textBAD_i$ that are marked ``probably good for $i$'', we will still obtain good approximations for $|I^{(a, b)}_i|$.

The following lemma shows why approximating $|I_i^{(a, b)}|$ is enough to approximate the contributions of the servers in the set $I^{(a, b)}_i$.
\begin{lemma}
The value $|I^{(a, b)}_i|$ can be used to approximate $\sum_{j \in I^{(a, b)}_i}x_i(j)$ up to a factor of $\theta'$.
\label{lma:size-to-contribution}
\end{lemma}
\begin{proof}
    For all servers $j \in I^{(a, b)}_i$, by definition, we have
    \begin{align*}      (\sqrt{\theta})^{a+b}P_{\start}F_{\start} \le \be_i^{-1}(f(x_i(j))) \le (\sqrt{\theta})^{a+b+2}P_{\start}F_{\start}
    \end{align*}
    which further implies using the monotonicity of $f$ that
    \begin{align*}
        f^{-1}(\be_i(\sqrt{\theta})^{a+b}P_{\start}F_{\start})\le x_i(j) \le f^{-1}(\be_i(\sqrt{\theta})^{a+b+2}P_{\start}F_{\start}) \le \theta'f^{-1}(\be_i(\sqrt{\theta})^{a+b}P_{\start}F_{\start}).
    \end{align*}
    Now we note
    \begin{align*}
       \frac{\sum_{j \in I^{(a, b)}_i}x_i(j)}{\theta'}\le |I^{(a, b)}_i| \cdot f^{-1}(\be_i(\sqrt{\theta})^{a+b}P_{\start}F_{\start}) \le \sum_{j \in I^{(a, b)}_i}x_i(j).&\qedhere
    \end{align*}
\end{proof}
We will then show that the expected number of servers in the set $I^{(a, b)}_i$ that sample $i$ can be used to obtain an approximation for $|I^{(a, b)}_i|$. 
\begin{lemma}
    Let $p_i(j) \coloneqq 1 - (1 - q_i(j))^N$ denote the probability that the coordinate $i$ is among the $N$ coordinates sampled at server $j$. For any tuple $(a, b)$, 
    \begin{align*}
        |I^{(a, b)}_i| \cdot \frac{(1 - (1 - (\sqrt{\theta})^{a+1}P_{\start})^N)}{\sqrt{\theta}}\le \sum_{j \in I^{(a, b)}_i}p_i(j) \le |I^{(a, b)}_i| \cdot (1 - (1 - (\sqrt{\theta})^{a+1}P_{\start})^N).
    \end{align*}
    \label{lma:expectation-to-size}
\end{lemma}
\begin{proof}
    Using monotonicity and concavity of the function $1 - (1-x)^N$ in the interval $[0, 1]$, for all $j \in I^{(a, b)}_i$, we have
    \begin{align*}
        (1 - (1 - (\sqrt{\theta})^{a+1}P_{\start})^N) \ge p_i(j) \ge  \frac{(1 - (1 - (\sqrt{\theta})^{a+1}P_{\start})^N)}{\sqrt{\theta}}.
    \end{align*}
    Summing the inequalities over all $j \in I^{(a,b)}_i$ gives the proof.
\end{proof}
If $\sum_{j \in I^{(a, b)}_i}p_j$ is large enough, we obtain using a Chernoff bound that the number of servers in $I^{(a, b)}_i$ at which $i$ is sampled is highly concentrated around the mean $\sum_{j \in I^{(a, b)}_i}p_j$ and using the above lemmas, we can obtain an estimate for $|I^{(a, b)}_i|$ and therefore estimate $\sum_{j \in I^{(a, b)}_i}x_i(j)$. We will follow this approach to approximate $\sum_{j \in I^{(a, b)}_i}x_i(j)$. Let $\bX^{(a, b)}_i$ be the number of servers in $j \in I^{(a, b)}_i$ that sample the coordinate $i$. By linearity of expectation, we have $\E[\bX^{(a, b)}_i] = \sum_{j \in I^{(a, b)}_i}p_j$. We will now use the following standard concentration bounds.
\begin{lemma}
Let $\bY_j$ for $j \in [t]$ be a Bernoulli random variable with $\Pr[\bY_j = 1] = p_j$. Let $\bY_1, \ldots, \bY_t$ be mutually independent and $\bX = \bY_1 + \cdots + \bY_t$. Then the following inequalities hold:
\begin{enumerate}
    \item If $\sum_j p_j \ge 100\log n$, then 
    \begin{align*}
        \Pr[\bX = (1 \pm 1/3)\sum_{j=1}^t p_j] \ge 1 - 1/n^3.
    \end{align*}
    \item For any values of $p_1, \ldots, p_t$, 
    \begin{align*}
        \Pr[\bX < 2\sum_{j=1}^t p_j + 4\log n] \ge 1 - 1/n^4.
    \end{align*}
\end{enumerate}
\label{lma:concentration}
\end{lemma}
\begin{proof}
To prove the first inequality, we use the multiplicative Chernoff bound. We have 
\begin{align*}
    \Pr[\bX = (1 \pm 1/3)\sum_{j=1}^t p_j] \le 2\exp\left(-\frac{\sum_{j=1}^t p_j}{27}\right).
\end{align*}
If $\sum_{j=1}^t p_j \ge 100\log n$, then the RHS is at most $1/n^3$. To prove the second inequality, we use the Bernstein concentration bound. We get
\begin{align*}
    \Pr[\bX \ge 2\sum_{j=1}^t p_j + 4\log n] \le \exp\left(- \frac{(\sum_{j} p_j + 4\log n)^2}{\sum_j p_j + (\sum_j p_j + 4\log n)/3}\right) \le \exp(-4\log n) \le 1/n^4.&\qedhere
\end{align*}
\end{proof}
We note that if $\sum_j p_j \ge 100\log n$, then with probability $\ge 1 - 1/n^3$, we have $\bX = (1 \pm 1/3)\log n \ge 66\log n$ and if $\sum_j p_j \le 30\log n$, then with probability $\ge 1 - 1/n^4$, we have $\bX < 64\log n$. Thus the value of $\bX$ can be used to separate the cases of $\sum_j p_j \ge 100\log n$ or $\sum_j p_j \le 30\log n$ with high probability.

Now we show that if $(a, b) \in \textGOOD_i$ and $N$ is large enough, then $\sum_{j \in I^{(a, b)}_i}p_j \ge 100\log n$.
\begin{lemma}
    If $(a, b) \in \textGOOD_i$ and the number of samples $N$ at each coordinator satisfies
    \begin{align*}
        N \ge O\left(\frac{c_f[s]}{s} \cdot \frac{c_f[A \cdot B] \cdot (A \cdot B) \cdot \sqrt{\theta} \cdot \log^4 n}{\varepsilon_1(1-\varepsilon_2)^2}\right),
    \end{align*}
    then either $(\sqrt{\theta})^{a+1}P_{\start} \ge 4\log n/N$ or $\sum_{j \in I^{(a, b)}_i}p_j \ge 100\log n$.
    \label{lma:good-i-has-concentration}
\end{lemma}
\begin{proof}
Assume $(\sqrt{\theta})^{a+1}P_{\start} < 4\log n/N$. By concavity of the function $1 - (1-x)^N$ in the interval $[0,1]$, 
\begin{align*}
    1 - (1 - (\sqrt{\theta})^{a+1}P_{\start})^N \ge \frac{N(\sqrt{\theta})^{a+1}P_{\start}}{4\log n} (1 - 1/n^4).
\end{align*}
For $(a, b) \in \textGOOD_i$, we then obtain
\begin{align*}
    \sum_{j \in I^{(a, b)}_i}p_j &\ge |I^{(a, b)}_i| \cdot \frac{N(\sqrt{\theta})^{a+1}P_{\start}}{4\log n\sqrt{\theta}}(1 - 1/n^4)\\
    &\ge \frac{s}{c_f[A \cdot B] \cdot (A \cdot B) \cdot (\sqrt{\theta})^{a+1}/(1-\varepsilon_2)^2} \frac{N(\sqrt{\theta})^{a+1}\varepsilon_1}{4C\log^3 n \cdot c_f[s] \cdot \sqrt{\theta}} \cdot (1 - 1/n^4)\\
    &\ge 100\log n. \qedhere
\end{align*}
\end{proof}
With a high probability, for all coordinates $i$ and all servers $j$ in the set $I^{(a, b)}_i$ for some $a$ with $(\sqrt{\theta})^{a+1}P_{\start} \ge 4\log n/N$, the coordinate $i \in \SC_j$. So the contribution from such servers can be computed exactly akin to the servers in the set $\textLARGE_j$. 

Now consider all the tuples $(a, b) \in \textGOOD_i$ with $(\sqrt{\theta})^{a+1}P_{\start} < 4\log n/N$. The above lemma shows that $\sum_{j \in I^{(a, b)}_i}p_j \ge 100\log n$. Thus, if we mark all the tuples $(a, b)$ with $\bX^{(a, b)}_i \ge 66\log n$ as ``probably good for $i$'', then with a probability $\ge 1 - 1/n^2$, all the tuples $(a, b)$ in $\textGOOD_i$ are marked as ``probably good for $i$'' and  any tuple $(a, b) \in \textBAD_i$ marked as ``probably good for $i$'' satisfies
\begin{align*}
    \bX^{(a, b)}_i \le 2.5 \sum_{j \in I^{(a, b)}_i}p_j.
\end{align*}
We now consider the following algorithm for computing $\hat{x}_i$ using which we then compute $\Est_i$:
\begin{enumerate}
    \item Let $\bS_i = \setbuilder{j}{i \in \SC_j}$ be the set of all the servers that sample the coordinate $i$.
    \item Let $\hat{x}_i \gets 0$ denote our initial estimate for $x_i$.
    \item Let $\bL_i = \setbuilder{j \in \bS}{ \frac{\be_i^{-1}f(x_i(j))}{\sum_{i}\be_i^{-1}f(x_i(j))} > \frac{4s}{c_f[s]\log^2 n}}$. This corresponds to the servers in $\textLARGE_i$ that have sampled $i$. With high probability, $\bL_i = \textLARGE_i$. Note that we know the value $x_i(j)$ for all $j \in \bL_i$.
    \item Update $\hat{x}_i \gets \hat{x}_i + \sum_{j \in \bL_i}x_i(j)$.
    \item Update $\bS_i \gets \bS_i \setminus \bL_i$. 
    \item Let $\mathbf{Sm}_i = \setbuilder{j \in \bS_i}{\frac{\be_i^{-1}f(x_i(j))}{\sum_i \be_i^{-1}f(x_i(j))} < \frac{\varepsilon_1}{c_f[s] \cdot (C\log^2 n)}}$. Corresponds to the servers in $\textSMALL_i$ that have sampled $i$ and therefore the contribution from these servers can be ignored.
    \item Update $\bS_i \gets \bS_i \setminus \mathbf{Sm}_i$.
    \item For each $a=0, 1, \ldots, A-1$ and $b = 0, 1, \ldots, B-1$, set $\bS^{(a, b)}_i \gets \setbuilder{j \in \bS_i}{j \in I^{(a, b)}_i}$. Note that since for all $j \in \bS_i$, we know the value of $x_i(j)$ and therefore we can compute the set $\bS^{(a, b)}_i$.
    \item For all tuples $(a, b)$ with $(\sqrt{\theta})^{a+1}P_{\start} > 4\log n/N$, with high probability $\bS^{(a, b)}_i = I^{(a, b)}_i$ and we update $\hat{x}_i \gets \hat{x}_i + \sum_{j \in \bS^{(a, b)}_i}x_i(j)$.
    \item For all other tuples $(a, b)$, if $|\bS^{(a, b)}_i| \ge 66\log n$, we mark $(a, b)$ as ``probably good for $i$'' and ignore the rest of the tuples.
    \item For all tuples $(a, b)$ marked as ``probably good for $i$'', we update
    \begin{align*}
        \hat{x}_i \gets \hat{x}_i + \frac{(2/5)|\bS_i^{(a, b)}|}{1 - (1 - (\sqrt{\theta})^{a+1}P_{\start})^N}f^{-1}(\be_i(\sqrt{\theta})^{a+b}P_{\start}F_{\start}).
    \end{align*}
    \item Compute $\Est_i \coloneqq \be_i^{-1}f(\hat{x}_i)$.
\end{enumerate}
Using all the lemmas we have gone through till now, we will argue that for all $i \in \SC$, $\hat{x}_i \le x_i$ with a high probability, which proves that $\Est_i$ is upper bounded by $\be_i^{-1}f(x_i)$. We then prove $\Est_{i^*}$ is large.
\begin{lemma}
    With probability $\ge 1 - 1/n$, for all $i \in [n]$, $\hat{x}_i \le x_i$ which implies $\Est_i \le \be_i^{-1}f(x_i)$.
    \label{lma:est-i-is-always}
\end{lemma}
\begin{proof}
Consider a fixed $i$. We will argue about contributions from different types of servers to $\hat{x}_i$ separately. By Lemma~\ref{lma:all-large-are-sampled}, we have $\bL_i = \textLARGE_i$ with probability $\ge 1 - 1/n^3$ and the contribution $\sum_{j \in \textLARGE_i}x_i(j)$ to $x_i$ is estimated correctly. We deterministically exclude all the servers in $\textSMALL_i$ and therefore we do not overestimate the contribution of $\sum_{j \in \textSMALL_i} x_i(j)$ to $x_i$.

Now consider the tuples $(a, b)$ for $a = 0, 1, \ldots, A-1$ and $b = 0, 1, \ldots, B-1$. If $a$ is such that $(\sqrt{\theta})^{a+1}P_{\start} \ge 4\log n/N$, we again have $\bS^{(a, b)}_i = I^{(a, b)}_i$ with a large probability similar to the analysis of Lemma~\ref{lma:all-large-are-sampled} and therefore we estimate the contribution of such servers to $x_i$ exactly.

If $(\sqrt{\theta})^{a+1}P_{\start} < 4\log n/N$ and $(a, b) \in \textGOOD_i$, then Lemma~\ref{lma:good-i-has-concentration} shows that $\sum_{j \in I^{(a, b)}_i}p_j \ge 100\log n$. We then obtain using Lemma~\ref{lma:concentration} that $66\log n \le |\bS^{(a, b)}_i| \le (4/3)\sum_{j \in I^{(a, b)}_i}p_j$ with a large probability. Using Lemma~\ref{lma:expectation-to-size}, we get that
\begin{align*}
    \frac{|\bS_i^{(a, b)}|}{(4/3)(1 - (1-(\sqrt{\theta})^{a+1}P_{\start})^N)} \le |I^{(a, b)}_i|.
\end{align*}
Using Lemma~\ref{lma:size-to-contribution}, we then obtain that
\begin{align*}
       &\frac{|\bS_i^{(a, b)}|}{(4/3)(1 - (1-(\sqrt{\theta})^{a+1}P_{\start})^N)}f^{-1}(\be_i(\sqrt{\theta})^{a+b}P_{\start}F_{\start})\\
       &\le |I^{(a, b)}_i|f^{-1}(\be_i (\sqrt{\theta})^{a+b}P_{\start}F_{\start}) \le \sum_{j \in I^{(a, b)}_i}x_i(j). 
\end{align*}
Hence, the contribution of tuples $(a, b) \in \textGOOD_i$ is not overestimated. If $(a, b) \in \textBAD_i$ and $(a, b)$ is marked ``probably good for $i$'', then using Lemma~\ref{lma:concentration}, we get
\begin{align*}
    |\bS_{i}^{(a, b)}| \le \frac{5}{2}\sum_{j \in I^{(a, b)}_i}p_j
\end{align*}
and using the same series of steps as above, we get
\begin{align*}
        &\frac{|\bS_i^{(a, b)}|}{(5/2)(1 - (1-(\sqrt{\theta})^{a+1}P_{\start})^N)}f^{-1}(\be_i(\sqrt{\theta})^{a+b}P_{\start}F_{\start})\\
        &\le |I^{(a, b)}_i|f^{-1}(\be_i (\sqrt{\theta})^{a+b}P_{\start}F_{\start}) \le \sum_{j \in I^{(a, b)}_i}x_i(j)
\end{align*}
which again shows that the contribution of tuples $(a, b)$ in $\textBAD_i$ but are marked ``probably good for $i$'' is also not overestimated. Overall we get that with a probability $\ge 1 - O(1/n^2)$, $\hat{x}_i \le x_i$. Using a union bound we have the proof.
\end{proof}
We now show that $\Est_{i^*}$ is large with a large probability.
\begin{lemma}
With probability $\ge 1 - 1/n^2$,
\begin{align*}
        \Est_{i^*}  \ge \frac{(1 - \varepsilon_2)^2}{\theta''}\be_{i^*}^{-1}f(x_{i^*}).
\end{align*}
\label{lma:est-i-star-is-significant}
\end{lemma}
\begin{proof}
    By Lemma~\ref{lma:all-large-are-sampled}, $\bL_{i^*} = \textLARGE_{i^*}$ and therefore the contribution to $x_{i^*}$ from the servers in $\textLARGE_{i^*}$ is captured by $\hat{x}_{i^*}$. We argued that contribution to $\hat{x}_{i^*}$ from servers in $\textSMALL_{i^*} \cup \Ignore_{i^*}$ can be ignored.

    Now consider the tuples $(a, b)$ for $a = 0, 1, \ldots, A-1$ and $b = 0, 1, \ldots, B-1$. If $(\sqrt{\theta})^{a+1}P_{\start} \ge (4\log n)/N$, then $\bS^{(a, b)}_{i^*} = I^{(a, b)}_{i^*}$ with probability $\ge 1 - 1/n^3$ and therefore the contribution from the servers in $I^{(a, b)}_{i^*}$ is captured exactly by $\hat{x}_{i^*}$. If $(a, b) \in \textBAD_{i^*}$, then we argued in Lemma~\ref{lma:bad-can-be-ignored} that we need not capture the contribution from those servers. So any contribution from servers in $\textBAD_{i^*}$ marked as ``probably good for $i^*$'' will only help in increasing $\hat{x}_{i^*}$.

    Now consider $(a, b) \in \textGOOD_{i^*}$ with $(\sqrt{\theta})^{a+1}P_{\start} < 4\log n/N$. By Lemma~\ref{lma:good-i-has-concentration}, with a large probability  
    \begin{align*}
        |\bS_{i^*}^{(a, b)}| \ge \frac{2}{3}\sum_{j \in I^{(a, b)}_{i^*}}p_j.
    \end{align*}
    By Lemma~\ref{lma:expectation-to-size}, we get 
    \begin{align*}
        |\bS_{i^*}^{(a, b)}| \ge \frac{2}{3}|I^{(a, b)}_{i^*}| \frac{1 - (1 - (\sqrt{\theta})^{a+1}P_{\start})^N}{\sqrt{\theta}}.
    \end{align*}
    Now, using Lemma~\ref{lma:size-to-contribution}, we get
    \begin{align*}
        |\bS_{i^*}^{(a, b)}| \cdot f^{-1}(\be_{i^*}(\sqrt{\theta})^{a+b}P_{\start}F_{\start}) \ge \frac{2}{3 \cdot \theta'} \cdot \frac{1 - (1-(\sqrt{\theta})^{a+1}P_{\start})^N}{\sqrt{\theta}}\sum_{j \in I^{(a, b)}_{i^*}}x_{i^*}(j)
    \end{align*}
    which then implies
    \begin{align*}
        \frac{(2/5)|\bS_{i^*}^{(a, b)}| \cdot f^{-1}(\be_{i^*}(\sqrt{\theta})^{a+b}P_{\start}F_{\start})}{1 - (1 - (\sqrt{\theta})^{a+1}P_{\start})^N} \ge \frac{4}{15 \cdot \theta \cdot \sqrt{\theta}} \sum_{j \in I^{(a, b)}_{i^*}}x_{i^*}(j).
    \end{align*}
Thus, overall, with high probability, we have
\begin{align*}
    \hat{x}_{i^*} &\ge \sum_{j \in \textLARGE_{i^*}}x_{i^*}(j) + \frac{4}{15 \cdot \theta \cdot \sqrt{\theta}}\sum_{(a, b) \in \textGOOD_{i^*}}\sum_{j \in I^{(a, b)}_{i^*}}x_{i^*}(j).
\end{align*}
We therefore have
\begin{align*}
    \Est_{i^*} \coloneqq \be_{i^*}^{-1}f(\hat{x}_{i^*}) \ge \frac{(1 - \varepsilon_2)^3}{\theta''}\be_{i^*}^{-1}f(x_{i^*})
\end{align*}
assuming for all $x$, $f(x/(4 \cdot \theta \cdot \sqrt{\theta})) \ge f(x) / \theta''$.
\end{proof}
\begin{theorem}
Assume we are given a super-additive nonnegative function $f$ that satisfies the ``approximate invertibility'' property with parameters $\theta, \theta', \theta'', \varepsilon_1$ and $\varepsilon_2$. Let $s \ge 1$ be the number of servers. Define $A = O(\log(s/\varepsilon_1)/\log \theta)$ and $B = O(\log(s^2/\varepsilon_1(1-\varepsilon_2))/\log \theta)$. Let $\be_1, \ldots, \be_n$ be independent standard exponential random variables shared across all the servers. Then there is a 2-round protocol in the coordinator model which uses a total communication of 
\begin{align*}
    O\left({c_f[s] \cdot \frac{c_f[A \cdot B] \cdot (A \cdot B) \cdot \sqrt{\theta}\log^4 n}{\varepsilon_1(1-\varepsilon_2)^2}} + s \cdot  \frac{\log^2 n \cdot \theta''}{(1-\varepsilon_2)^3}\right)
\end{align*}
words of communication and with probability $\ge 1 - 1/\poly(n)$ computes $\max_i \be_{i}^{-1}f(x_i)$.
\label{thm:single-exponential}
\end{theorem}
\begin{proof}
First we condition on the event that $\sum_i \be_i^{-1}f(x_i) \le (C\log^2 n) \cdot \max_i \be_i^{-1}f(x_i)$ which holds with probability $\ge 1 - 1/\poly(n)$. Note that the randomness used in sampling is independent of the exponential random variables. Define $i^* \coloneqq \argmax_i \be_i^{-1}f(x_i)$.

Let $N = O((c_f[s]/s) \cdot c_f[A \cdot B] \cdot (A \cdot B) \cdot \sqrt{\theta}\log^4 n/(\varepsilon_1(1-\varepsilon_2)^2))$. Let each server $j$ sample $N$ coordinates independently from its local distribution:
\begin{align*}
    \Pr[\bi = i] = \frac{\be_i^{-1}f(x_i(j))}{\sum_i \be_i^{-1}f(x_i(j))}.
\end{align*}
Let $\SC_j$ be the set of coordinates sampled by the server $j$. Now, each server $j$ sends the set $\SC_j$ along with the values $x_i(j)$ for $i \in \SC_j$ to the coordinator. Additionally, each server also sends the value $\sum_i \be_i^{-1}f(x_i(j))$ to the coordinator. Note that this requires a total communication of $O(N \cdot s)$ words of communication. 

Now, the coordinator computes $\SC = \bigcup_j \SC_j$. By Lemma~\ref{lma:i-star-is-sampled}, we have $i^* \in \SC$ with probability $\ge 1 - 1/\poly(n)$. Now the coordinator computes $\hat{x}_i$ and a value $\Est_i$ for each $i \in \SC$ using the algorithm described above. Lemma~\ref{lma:est-i-is-always} shows that with a probability $\ge 1 - 1/\poly(n)$, for all $i \in \SC$, we have $\Est_i \le \be_i^{-1}f(x_i)$ and Lemma~\ref{lma:est-i-star-is-significant} shows that with probability $\ge 1 - 1/\poly(n)$, we have $\Est_{i^*} \ge \frac{(1-\varepsilon_2)^3}{\theta''}\be_{i^*}^{-1}f(x_{i^*})$. Using a union bound, all these events hold with probability $\ge 1 - 1/\poly(n)$. Condition on all these events. 

Let $\PL$ be the set of coordinates $i \in \SC$ with the $O(C\log^2 n \cdot \theta'' / (1-\varepsilon_2)^3)$ largest values of $\Est_i$. We have $i^* \in \PL$ since we conditioned on all the above events. Now the coordinator queries sends the set $\PL$ to \emph{each} server $j$ and asks for the values $x_i(j)$ for $i \in \PL$. Then the servers all send the requested information in the second round of communication. Note that the total communication required is $O(s \cdot |\PL|)$ words.
Since $i^*$ in $\PL$, we obtain that
\begin{align*}
    \max_{i \in \PL} \be_{i}^{-1}f(\sum_{j \in [s]}x_i(j)) = \be_{i^*}^{-1}f(x_{i^*})
\end{align*}
and the coordinator can compute this value after receiving the required information from the servers in the second round of communication. This proves the theorem.
\end{proof}
We can run the protocol in the above theorem concurrently using $O(1/\varepsilon^2)$ independent copies of the exponential random variables and then obtain a $1 \pm \varepsilon$ approximation for $\sum_i f(x_i)$ with a probability $\ge 99/100$. We note that the overall protocol requires two rounds and a total communication of $O_{\theta, \theta', \theta''}\left(\frac{c_f[s]}{\varepsilon^2}\polylog(n)\right)$ words of communication.
\begin{theorem}
    Let $f$ be a non-negative, increasing, super-additive function that satisfies the ``approximate invertibility'' properties with the parameters $\theta, \theta', \theta''$. Let there be $s$ servers and each of the servers holds a non-negative vector $x(1), \ldots, x(s) \in \R^n$ respectively. Define $A = O(\log(s \cdot \theta'')/\log \theta)$ and $B = O(\log(s^2 \cdot (\theta'')^2)/\log \theta)$.
    Given $\varepsilon < 1/n^c$ for a small constant $c$, there is a two round protocol that uses a total of 
    \begin{align*}
        O\left(\frac{c_f[s]}{\varepsilon^2} \cdot {c_f[A \cdot B] \cdot (A \cdot B) \cdot \sqrt{\theta} \cdot (\theta'')^3 \cdot \log^4 n} + \frac{s}{\varepsilon^2} \cdot {\log^2 n \cdot (\theta'')^4}\right)
    \end{align*}
    words of communication and with probability $\ge 9/10$ computes a $1 \pm \varepsilon$ for the quantity $\sum_{i}f(x_i)$.
    \label{thm:main-estimation}
\end{theorem}
\begin{proof}
    For $k \in [O(1/\varepsilon^2)]$ and $i \in [n]$, let $\be^{(k)}_i$ be an independent standard exponential random variable. Let $i^*(k) \coloneqq \argmax_{i \in [n]}(\be_i^{(k)})^{-1}f(x_{i})$. By a union bound, the following hold simultaneously with probability $\ge 1 - 1/\poly(n)$:
    \begin{align*}
        \text{for all }k,\ \sum_{i}(\be_i^{(k)})^{-1}f(x_i) \le (C \log^2 n) \cdot (\be_{i^*{(k)}}^{(k)})^{-1}f(x_{i^*{(k)}})
    \end{align*}
    and
    \begin{align*}
        \ln(2) \cdot \text{median}_k\, (\be_{i^*(k)}^{(k)})^{-1}f(x_{i^*(k)}) = (1 \pm \varepsilon)\sum_{i \in [n]}f(x_i).
    \end{align*}
    We condition on these events. Now concurrently for each $k$, we use the exponential random variables $\be^{(k)}_1, \ldots, \be^{(k)}_n$ and run the protocol in Theorem~\ref{thm:single-exponential} to obtain the value of $(\be_{i^*(k)}^{(k)})^{-1}f(x_{i^*(k)})$ with probability $\ge 1 - 1/\poly(n)$. We union bound over the success of the protocol for all $k$ and obtain that with probability $\ge 9/10$, we can compute the exact value of 
    \begin{align*}
        \ln(2) \cdot \text{median}_k\, (\be_{i^*(k)}^{(k)})^{-1}f(x_{i^*(k)})
    \end{align*}
    which then gives us a $1 \pm \varepsilon$ approximation of $\sum_{i \in [n]}f(x_i)$. The communication bounds directly follow from Theorem~\ref{thm:single-exponential}.
\end{proof}
We obtain the following corollary for estimating $F_k$ moments.
\begin{corollary}
    Let $k > 2$ be arbitrary. In the coordinator model with $s$ servers that each hold a non-negative vector $x(j) \in \R^n$, there is a randomized two round protocol that uses a total of $\tilde{O}_k(s^{k-1}\polylog(n)/\varepsilon^2)$ bits of communication and approximate $\sum_i (\sum_{j \in [s]}x_i(j))^k$ up to a $1 \pm \varepsilon$ factor with probability $\ge 9/10$.
    \label{cor:fk-estimation}
\end{corollary}
\begin{proof}
    For the function $f(x) = x^k$, we have $c_f[s] = s^{k-1}$ by a simple application of the Holder's inequality. We additionally note that $x^k$ is ``approximately invertible'' with parameters $\theta = 2$, $\theta' = 2^{1/k}$, and $\theta'' = 2 \cdot 8^{k/2}$. Therefore $\varepsilon_1$ and $1 - \varepsilon_2$ can be taken as $1 / (2 \cdot 8^{k/2})$. We now have $A, B = O(k + \log s)$ so that $c_f[A \cdot B] \cdot (A \cdot B) = (k + \log s)^k$. From the above theorem, we therefore obtain that there is a two round protocol that computes a $1 \pm \varepsilon$ approximation of $\sum_{i\in [n]} (\sum_{j\in [s]}x_i(j))^k$ with probability $\ge 9/10$ and uses a total communication of
    \begin{align*}
        O\left(\frac{s^{k-1}}{\varepsilon^2} \frac{(k + \log s)^k}{8^{3k/2}}\polylog(n)\right)
    \end{align*}
    words of communication.
\end{proof}
\subsection{Higher-Order Correlations}
Kannan, Vempala, and Woodruff \cite{kannan2014principal} also study the problem of approximating higher order correlations and list a few applications of the problem in their paper. In this problem, there are $s$ servers and the $j$-th server holds a set of $n$-dimensional vectors $W_j$. Given a parameter $k$, and functions $f: \R_{\ge 0} \rightarrow \R_{\ge 0}$, $g: \R_{\ge 0}^k \rightarrow \R_{\ge 0}$, the coordinator wants to approximate
\begin{align*}
    M(f, g, W_1, \ldots, W_s) \coloneqq \sum_{i_1, i_2, \ldots, i_k\, \text{distinct}}f\left(\sum_{j}\sum_{v \in W_j}g(v_{i_1}, v_{i_2}, \ldots, v_{i_k})\right)
\end{align*}
As they mention, for each server $j$, we can create a vector $w(j)$ with $r = \binom{n}{k}k!$ components (one for each tuple $(i_1, \ldots, i_k)$ with distinct values of $i_1, i_2,\ldots, i_k \in [n]$) defined as
\begin{align*}
    [w(j)]_{(i_1, i_2,\ldots, i_k)} \coloneqq \sum_{v \in W_j} g(v_{i_1}, v_{i_2}, \ldots, v_{i_k}).
\end{align*}
Now running the function sum approximation protocol on the vectors $w(1), \ldots, w(s)$ with the function $f$, we can compute a $1 \pm \varepsilon$ approximation for $M(f, g)$ using a total of 
\begin{align*}
    O_{\theta, \theta', \theta''}\left(\frac{c_f[s]}{\varepsilon^2} \cdot \polylog(r)\right) = O_{\theta, \theta', \theta''}\left(\frac{c_f[s]}{\varepsilon^2} \cdot \poly(k, \log n)\right)
\end{align*}
words of communication. The main issue in implementing this algorithm is that all the servers have to realize the $r = \binom{n}{k}k!$ dimensional vectors $w(j)$ which end up occupying $O(n^k)$ space which is prohibitive when $n$ is large. Using a simple trick, we can show that to execute the protocol in Theorem~\ref{thm:main-estimation} can be implemented without using $O(n^k)$ space. 

We first solve for the issue of sharing $O(n^k)$ exponential random variables across all the servers. In Appendix~\ref{sec:nisan}, we show that the exponential random variables used in the protocol need not be independent but can be generated using Nisan's Pseudorandom Generator (PRG) \cite{nisan1992pseudorandom}. The seed for Nisan's PRG needs to be only of length $O(k^2 \log^2 (n/\varepsilon))$ and hence the shared randomness across the servers is only of this size.

The main reason we need the vectors $w(j)$ is so that the server $j$ can sample independent copies of the random variable $(\bi_1, \ldots, \bi_k)$ with probability distribution
\begin{align*}
    \Pr[(\bi_1, \ldots, \bi_k) = (i_1, \ldots, i_k)] = \frac{(\be_{(i_1, \ldots, i_k)})^{-1} [w(j)]_{(i_1, \ldots, i_k)}}{\sum_{i_1', \ldots, i_k' }(\be_{(i_1', \ldots, i_k')})^{-1} [w(j)]_{(i_1', \ldots, i_k')}}
\end{align*}
where $\be_{(i_1, \ldots, i_k)}$ is an independent standard exponential random variable. But we note that to sample from this distribution, the protocol does not need $O(n^k)$ space. Consider the lexicographic ordering of the tuples $(i_1, \ldots, i_k)$ with all $i_1, \ldots, i_k$ distinct. The algorithm goes over the tuples in the lexicographic orders, computes the value of $(\be_{(i_1, \ldots, i_k)})^{-1} [w(j)]_{(i_1, \ldots, i_k)}$ by generating the random variable $\be_{(i_1, \ldots, i_k)}$ using Nisan's PRG and then uses a \emph{reservoir sampling} algorithm to sample a tuple form the above defined distribution. This entire process can be accomplished using a constant amount of space and hence implementing the function sum approximation protocol on the vectors $w(1), \ldots, w(s)$ can be accomplished without using $\Theta(n^k)$ space at each of the servers. Hence we have the following theorem:
\begin{theorem}
Let there be $s$ servers each holding an arbitrary set of $n$-dimensional non-negative vectors $W_1, \ldots, W_s \subseteq$ respectively. Given a function $f:\R_{\ge 0} \rightarrow \R_{\ge 0}$ that satisfies the ``approximate invertibility'' property with parameters $\theta, \theta', \theta'' > 1$ and a function $g : \R^{k}_{\ge 0} \rightarrow \R^{k}_{\ge 0}$, there is a randomized two round protocol that approximates $M(f, g, W_1, \ldots, W_s)$ up to a $1 \pm \varepsilon$ factor with probability $\ge 9/10$. The protocol uses a total of 
\begin{align*}
        O\left(\frac{c_f[s]}{\varepsilon^2} \cdot {c_f[A \cdot B] \cdot (A \cdot B) \cdot \sqrt{\theta} \cdot (\theta'')^3 \cdot k^4\log^4 n} + \frac{s}{\varepsilon^2} \cdot {k^2\log^2 n \cdot (\theta'')^4}\right)
\end{align*}
words of communication, where $A = O(\log(s \cdot \theta'')/\log \theta)$ and $B = O(\log(s^2 \cdot (\theta'')^2)/\log \theta)$
\label{thm:higher-order}
\end{theorem}

%% file: lowerbounds.tex
\section{Lower Bounds}
\subsection{Lower Bound for Sum Approximation}
For $F_k$ approximation problem, \cite{WZ12} show a $\Omega(s^{k-1}/\varepsilon^2)$ lower bound on the total communication of any protocol that $1+\varepsilon$ approximates the $F_k$ value of a vector that is distributed among $s$ servers. They show the lower bound by reducing from a communication problem called the $k$-BTX. Their proof can be adapted in a straightforward way to obtain the following result for general function approximation for some class of functions $f$.
\begin{theorem}
    Let $f$ be a non-negative, super-additive function and $c_f[s]$ be the parameter such that for all $y_1, \ldots, y_s \ge 0$, then
    \begin{align*}
        f(y_1 + \cdots + y_s) \le \frac{c_f[s]}{s}(\sqrt{f(y_1)} + \cdots + \sqrt{f(y_s)})^2.
    \end{align*}
    Assume that there exists $y^*$ such that the above inequality is tight when $y_1 = \cdots = y_s = y^*$ i.e., $f(sy^*) = s \cdot c_f[s] \cdot f(y^*)$ and let $\beta = f(sy^*)/(2 \cdot f(sy^*/2)) \ge 1$. If $s \ge \Omega(\beta)$, then any protocol that approximates $\sum_i f(x_i)$ in the coordinator model up to a $1 \pm (1/72\beta - 1/72\beta^2)\varepsilon$ factor must use a total communication of $\Omega(c_f[s]/\varepsilon^2)$ bits.
\end{theorem}
While the requirements in the above theorem may seem circular, as in, $\beta = f(sy^*)/2f(sy^*/2)$ and $s \ge \Omega(\beta)$, note that $\beta$ is upper bounded by $\max_x f(x)/(2f(x/2))$ which is independent of the number of servers $s$.
\begin{proof}
We prove the communication lower bound by showing that any protocol which can approximate $\sum_i f(x_i)$ where $f$ is a function that satisfies the properties in the theorem statement can be used to construct a protocol for solving the so-called $s$-BTX (Block-Threshold-XOR) problem on a specific hard input distribution $\nu$.

To define the $s$-BTX communication problem and a hard distribution $\nu$, we first define the $s$-XOR problem and a hard distribution $\psi_n$ for this problem. There are $s$ sites $S_1, \ldots, S_s$. Each site $S_j$ holds a block $b(j) = (b_1(j), \ldots, b_n(j))$ of $n$ bits. The $s$ sites want to compute the following function:
\begin{align*}
    s\text{-XOR}(b(1), \ldots, b(s)) = \begin{cases}
        1, & \text{if there is an index $i\in [n]$ such that}\\
        &\quad \text{ $b_{i}(j) = 1$ for exactly $s/2$ values of $j$},\\
        0, &\text{otherwise.}
    \end{cases}
\end{align*}
Woodruff and Zhang \cite{WZ12} define an input distribution $\varphi_n$ to the $s$-XOR problem as follows. For each coordinate $i \in [n]$, a variable $D_i$ is chosen uniformly at random from the set $\{1, \ldots, s\}$. Conditioned on the value $D_{i}$, all but the $D_i$-th site sets their input to $0$ in the $i$-th coordinate, whereas the $D_i$-th site sets its input in the $i$-th coordinate to $0$ or $1$ with equal probability. Let $\varphi_1$ be this distribution on one coordinate.

Next, a special coordinate $M$ is chosen uniformly at random from $[n]$ and the inputs in the $M$-th coordinate at all $s$ sites  are modified as follows: for the first $s/2$ sites, the inputs in the $M$-th coordinate are replaced with all $0$s with probability $1/2$ and all $1$s with probability $1/2$. Similarly for the last $s/2$ sites, the inputs in the $M$-th coordinate are replaced with all $0$s with probability $1/2$ and all $1$s with probability $1/2$. Let $\psi_1$ denote the input distribution on the special coordinate and $\psi_n$ denote the input distribution that on special coordinate follows $\psi_1$ and follows $\varphi_1$ on the remaining $n-1$ coordinates.

We will now define the $s$-BTX problem and a hard input distribution $\nu$. Again, there are $s$ sites $S_1, \ldots, S_s$. Each site $S_j$ holds an input consisting of $1/\varepsilon^2$ blocks and each block is an input for that site in a corresponding $s$-XOR problem. Concretely, each site $S_j$ holds a length $n/\varepsilon^2$ vector $b(j) = (b^{1}(j), \ldots, b^{1/\varepsilon^2}(j))$ divided into $1/\varepsilon^2$ blocks of $n$ bits each. There are $1/\varepsilon^2$ instances of the $s$-XOR problem with the $\ell$-th instance having the inputs $b^{\ell}(1), \ldots, b^{\ell}(s)$. In the $s$-BTX problem, the sites want to compute the following:
\begin{align*}
    s\text{-BTX}(b(1), \ldots, b(s)) = \begin{cases}
        1, &\text{if}\ |\sum_{\ell \in [1/\varepsilon^2]}s\text{-XOR}(b^{\ell}(1), \ldots, b^{\ell}{(s)}) - 1/2\varepsilon^2|        \ge 2/\varepsilon\\
        0, &\text{if}\ |\sum_{\ell \in [1/\varepsilon^2]}s\text{-XOR}(b^{\ell}(1), \ldots, b^{\ell}{(s)}) - 1/2\varepsilon^2|        \le 1/\varepsilon\\
        *, &\text{otherwise.}
    \end{cases}
\end{align*}
A hard input distribution $\nu$ for this problem is defined as follows: The input of the $s$ sites in each block is independently chosen according to the input distribution $\psi_n$ defined above for the $s$-XOR problem. Let $B$ be the random variable denoting the inputs $(b(1), \ldots, b(s))$ when drawn from input distribution $\nu$ and $M = (M^1, \ldots, M^{1/\varepsilon^2})$ denote the random variable where $M^{\ell}$ denotes the special coordinate in the $\ell$-th block of the inputs and $D$ denotes the special sites for all the coordinates in all $1/\varepsilon^2$ instances of the $s$-XOR problem. \cite{WZ12} prove the following theorem:
\begin{theorem}[{\cite[Theorem~7]{WZ12}}]
    Let $\Pi$ be the transcript of any randomized protocol for the $s$-BTX problem on input distribution $\nu$ with error probability $\delta$ for a sufficiently small constant $\delta$. We have $I(B;\Pi \mid M, D) \ge \Omega(n/s\varepsilon^2)$, where information is measured with respect to the input distribution $\nu$.
\end{theorem}
The theorem essentially states that the transcript of any protocol that solves the $s$-BTX problem on the input distribution $\nu$ with a large probability must have a large amount of ``information'' about the input vectors when conditioned on the random variables $M, D$. Since the randomized communication complexity is always at least the conditional information cost, the above theorem implies that any randomized protocol that solves the $s$-BTX problem on input distribution $\nu$ with error probability $\delta$ has a communication complexity of $\Omega(n/s\varepsilon^2)$.

We show a lower bound on the communication complexity of the function sum estimation problem for $f$ in the theorem statement by reducing the $s$-BTX problem to approximating $\sum_i f(x_i)$ for appropriately chosen vectors $x(1), \ldots, x(s)$ at each of the sites.

Let $n = s \cdot c_f[s]$ so that the communication complexity of a randomized protocol for $s$-BTX on input distribution $\nu$ is $\Omega(c_f[s]/\varepsilon^2)$. Let $(b(1), \ldots, b(s))$ be inputs to the $s$-BTX problem drawn from the distribution $\nu$. Notice that each $b(j)$ is a binary vector with $s \cdot c_f[s]/\varepsilon^2$ coordinates. Now define $b = b(1) + \cdots + b(s)$.

Since the input $(b(1), \ldots, b(s))$ is drawn from the distribution $\nu$, we note the following about vector $b$:
\begin{enumerate}
    \item Each block of $1/\varepsilon^2$ coordinates has exactly one coordinate $i$ in which $b_i = s$ with probability $1/4$, $b_i = s/2$ with probability $1/2$ and $b_i = 0$ with probability $1/4$.
    \item In each block, all other coordinates apart from the one singled out above have a value $0$ with probability $1/2$ and $1$ with probability $1/2$.
\end{enumerate}
Therefore the vector $b$ when $(b(1), \ldots, b(s))$ is sampled from $\nu$ has, in expectation, $\frac{s \cdot c_f - 1}{2\varepsilon^2}$ coordinates with value $1$, $\frac{1}{2\varepsilon^2}$ coordinates with value $s/2$ and $\frac{1}{4\varepsilon^2}$ coordinates with value $s$.

For each $j \in [s]$, define $x(j) = y^* \cdot b(j)$ where $y^*$ is as in the theorem statement and let $x = \sum_j x(j) = y^* \cdot b$. From the above properties of the vector $b$, the vector $x$ has coordinates only with values $0, y^*, sy^*/2, sy^*$ and in expectation it has $\frac{s\cdot c_f[s] - 1}{2\varepsilon^2}$ coordinates with value $y^*$, $1/2\varepsilon^2$ coordinates with value $sy^*/2$ and $1/4\varepsilon^2$ coordinates with value $sy^*$. So, we write
\begin{align*}
    W \coloneqq \sum_i f(x_i) = \left(\frac{s \cdot c_f[s] - 1}{2\varepsilon^2} + Q\right) \cdot f(y^*) + \left(\frac{1}{2\varepsilon^2} + U\right) \cdot f(sy^*/2) + \left(\frac{1}{4\varepsilon^2} + V\right) \cdot f(sy^*)
\end{align*}
where $Q, U, V$ denote the deviations from the means for each type of coordinate. Note that we have $f(0) = 0$ and hence no contribution from such random variables. Now, the $s$-BTX problem is exactly to determine if $|U| \ge 2/\varepsilon$ or $|U| \le 1/\varepsilon$ and we want to show that a protocol to approximate $\sum_i f(x_i)$ can be used to distinguish between the cases.

We now define $x^{\text{left}} = \sum_{j=1}^{s/2}x(j)$ and $x^{\text{right}} = \sum_{j=s/2+1}^s x(j)$. Let $W^{\text{left}} \coloneqq \sum_i f(x^{\text{left}}_i)$ and $W^{\text{right}} \coloneqq \sum_i f(x^{\text{right}}_i)$. We now note that
\begin{align*}
    W^{\text{left}} + W^{\text{right}} = \left(\frac{s \cdot c_f[s] - 1}{2\varepsilon^2} + Q\right) \cdot f(y^*)+ \left(\frac{1}{2\varepsilon^2} + U\right) \cdot f(sy^*/2) + \left(\frac{1}{4\varepsilon^2} + V\right) \cdot 2 \cdot f(sy^*/2).
\end{align*}
Note that for the function $f$, we have $f(sy^*) = \beta \cdot 2 \cdot f(sy^*/2)$ for some $\beta > 1$. Hence,
\begin{align*}
    \beta(W^{\text{left}} + W^{\text{right}}) - W = (\beta-1)\left(\left(\frac{s \cdot c_f[s] - 1}{2\varepsilon^2} + Q\right) \cdot f(y^*)+ \left(\frac{1}{2\varepsilon^2} + U\right) \cdot f(sy^*/2)\right).
\end{align*}
Let $\calP$ be a protocol that can approximate $\sum_i f(x_i)$, up to a $1 \pm \alpha\varepsilon$ factor, when the vector $x$ is distributed across $s$ servers. Let $\widetilde{W}$, $\widetilde{W}^{\text{left}}$ and $\widetilde{W}^{\text{right}}$ be the $1 \pm \alpha\varepsilon$ approximations for $W$, $W^{\text{left}}$ and $W^{\text{right}}$ computed by running the protocol $\calP$ on three different instances of the function sum approximation problem. We first note that for the vector $x$ constructed using the inputs $(b(1), \ldots, b(s))$, we have $\sum_i f(x_i) \le \frac{s \cdot c_f[s]}{\varepsilon^2}f(y^*) + \frac{1}{\varepsilon^2}f(sy^*) \le \frac{2 \cdot f(sy^*)}{\varepsilon^2}$ with probability $1$ where we used the fact that $s \cdot c_f[s] \cdot f(y^*) = f(sy^*)$. Hence,
\begin{align*}
    \widetilde{W} = W \pm \frac{\alpha \cdot 2 \cdot f(sy^*)}{\varepsilon},\quad \widetilde{W}^{\text{left}} = W^{\text{left}} \pm \frac{\alpha \cdot 2 \cdot f(sy^*)}{\varepsilon},\quad \text{and}\quad \widetilde{W}^{\text{right}} = W^{\text{right}} \pm \frac{\alpha \cdot 2 \cdot f(sy^*)}{\varepsilon}
\end{align*}
which then implies
\begin{align*}
    &\beta(\widetilde{W}^{\text{left}} + \widetilde{W}^{\text{right}}) - \widetilde{W}\\
    &= \beta(W^{\text{left}} + W^{\text{right}}) - W \pm \frac{6\alpha\beta}{\varepsilon}f(sy^*)\\
    &= (\beta-1)\left(\left(\frac{s \cdot c_f[s] - 1}{2\varepsilon^2} + Q\right) \cdot f(y^*)+ \left(\frac{1}{2\varepsilon^2} + U\right) \cdot f(sy^*/2)\right) \pm \frac{6\alpha\beta}{\varepsilon}f(sy^*).
\end{align*}
Now, we note that with a large constant probability over the distribution $\nu$, the random variable $Q$ satisfies
\begin{align*}
    |Q| \le \frac{C\sqrt{s \cdot c_f[s]}}{\varepsilon}
\end{align*}
for a large enough constant $C$ by a simple application of a Chernoff bound. Hence, with a union bound on the above event and the correctness of the protocol on inputs $x$,  $x^{\text{left}}$ and $x^{\text{right}}$, we get
\begin{align*}
     \beta(\widetilde{W}^{\text{left}} + \widetilde{W}^{\text{right}}) - \widetilde{W} &= (\beta-1) \cdot \left(\frac{1}{2\varepsilon^2} + U\right) \cdot f(sy^*/2) + (\beta - 1) \cdot \left(\frac{s \cdot c_f[s]-1}{2\varepsilon^2}\right) \cdot f(y^*)\\
     &\quad \pm (\beta-1)\frac{C\sqrt{s \cdot c_f[s]}}{\varepsilon}f(y^*) \pm \frac{6\alpha\beta}{\varepsilon}f(sy^*).
\end{align*}
Dividing the expression by $(\beta-1)$, we get
\begin{align*}
    \frac{ \beta(\widetilde{W}^{\text{left}} + \widetilde{W}^{\text{right}}) - \widetilde{W}}{\beta-1} &= \left(\frac{1}{2\varepsilon^2} + U\right) \cdot f(sy^*/2) +  \left(\frac{s \cdot c_f[s]-1}{2\varepsilon^2}\right)f(y^*)\\
     &\quad \pm \frac{C\sqrt{s \cdot c_f[s]}}{\varepsilon}f(y^*) \pm \frac{6\alpha\beta}{\varepsilon(\beta-1)}f(sy^*) 
\end{align*}
We now use $f(sy^*) = s \cdot c_f[s] \cdot f(y^*)$, $f(sy^*/2) = s \cdot c_f[s] \cdot f(y^*)/2\beta$ to obtain that
\begin{align*}
    \frac{ \beta(\widetilde{W}^{\text{left}} + \widetilde{W}^{\text{right}}) - \widetilde{W}}{\beta-1} = \frac{s \cdot c_f[s] \cdot f(y^*)}{2\beta} \cdot \left(\frac{1}{2\varepsilon^2} + U + \frac{(s \cdot c_f[s] - 1) \cdot 2\beta}{2\varepsilon^2 \cdot s \cdot c_f[s]} \pm \frac{2C \cdot \beta }{\varepsilon\sqrt{s \cdot c_f[s]}} \pm \frac{12\alpha\beta^2}{\varepsilon(\beta-1)}\right).
\end{align*}
If $s \ge C' \cdot \beta$, then $\sqrt{s \cdot c_f[s]} \ge s \ge C' \cdot \beta$ as well. If $C' \ge 8C$, and $\alpha \le (\beta-1)/72\beta^2$,  then
\begin{align*}
    \frac{ \beta(\widetilde{W}^{\text{left}} + \widetilde{W}^{\text{right}}) - \widetilde{W}}{\beta-1} = \frac{s \cdot c_f[s] \cdot f(y^*)}{2\beta} \cdot \left(\frac{1}{2\varepsilon^2} + U + \frac{(s \cdot c_f[s] - 1) \cdot 2\beta}{2\varepsilon^2 \cdot s \cdot c_f[s]} \pm \frac{5}{12\varepsilon}\right).
\end{align*}
Hence, we can distinguish between the case when $|U| \le 1/\varepsilon$ or $|U| \ge 2/\varepsilon$ using the expression on the LHS of the above equality. As, $n = s \cdot c_f[s]$, the lower bound for the $s$-BTX problem implies that any randomized protocol that approximates $\sum_i f(x_i)$ in the coordinator model when the vector $x$ is split between $s$ servers, up to a $1 \pm \left(\frac{1}{72\beta} - \frac{1}{72\beta^2}\right)\varepsilon$ factor, with probability $\ge 1 - \delta$ for a small enough constant $\delta$, must use a total communication of $\Omega(c_f[s]/\varepsilon^2)$ bits.
\end{proof}
\subsection{\texorpdfstring{$F_k$}{Fk} Estimation Lower Bound for 1-round Algorithms}
We use the multi-player set disjointness problem to show that a one round protocol for $F_k$ estimation using shared randomness requires a total of $\tilde{\Omega}(s^{k-1}/\varepsilon^k)$ bits of communication. In the one-way \emph{blackboard} private-coin communication model, it is known that the $t$-player \emph{promise} set disjointness problem, with sets drawn from $[n]$, has a communication lower bound of $\Omega(n/t)$. In this problem, each of the $t$ servers receives a subset of $[n]$ with the promise that the sets received by all the servers are either mutually disjoint or that there is exactly one element that is present in all the subsets.

The one round algorithms in our paper can be \emph{implemented} in the standard $1$-way bloackboard model: In this model, all the servers in a deterministic order \emph{write} the information on a publicly viewable blackboard. The total communication in this model is then the total number of bits written on the blackboard. The one round algorithms in the coordinator model are strictly weaker as each server sends its information to the coordinator without even looking at others bits. The lower bound of $\Omega(n/t)$ in the 1-way blackboard model was shown in \cite{chakrabarti2003near} and later \cite{Gro09} extended the $\Omega(n/t)$ lower bound to an arbitrary number of rounds.
\begin{theorem}
    Given $s \ge 3$ servers each having an $n$ dimensional vector $x(1), \ldots, x(s)$ respectively, any 1-round $F_k$ estimation algorithm, in which the servers send a single message to the coordinator, that approximates $F_k(x)$ up to $1 \pm \varepsilon$ factor with probability $\ge 9/10$ over the randomness in the protocol, must use $\Omega_k(s^{k-1}/\varepsilon^{k}\log(s/\varepsilon))$ bits of total communication.
    \label{thm:Fk-one-round-lowerbound}
\end{theorem}
\begin{proof}
The lower bounds in \cite{chakrabarti2003near,Gro09} hold even with shared randomness, but are not stated that way, so one can also argue as follows to handle shared randomness: suppose there is a public coin algorithm in the $1$-way blackboard communication model using  a total of $c$ bits of communication. Then by Newman's equivalence \cite{Newman} of private-coin vs public-coin protocols up to an additive logarithmic increase in the communication, there is a private coin algorithm in the one-way blackboard communication model using a total of $c + O(\log (nt))$ bits. The first player samples one of the strings pre-shared among all the servers and announces the index of the string on the blackboard and then the remaining servers proceed with the computation using this string as the shared random bits. Hence, $c = \Omega(n/t) - O(\log nt)$ and when $t \le n^\alpha$ for a constant $\alpha < 1$, we obtain that $\Omega(n/t)$ bits is a lower bound on the communication complexity of 1-way public coin protocols in the blackboard model that solve the $s$-player set disjointness problem. 

In our model, the $F_k$ estimation algorithm is even weaker than the $1$-way public coin protocol in the blackboard model as all the servers send their bits to the coordinator without looking at others bits. Hence, the lower bound of $\Omega(n/t)$ bits can be used to lower bound the communication complexity.

Let $n = s^k/\varepsilon^k$ and $t=s/2$. Consider the instance of a $t$ player set-disjointness problem. We will encode the problem as approximating the $F_k$ moment of an $n$ dimensional vector distributed over $s$ servers.

For $j=1,\ldots,s/2$, the player $j$ encodes the subset $S_j \subseteq [n]$ they receive as an $n$ dimensional vector $x(j)$ by putting $1$ in the coordinates corresponding to the items in the set and $0$ otherwise.

Now each of the $s/2$ servers runs a $(1/Cn)$-error protocol in the coordinator model (as in the protocol fails with probability at most $1/Cn$) to approximate  $\|\sum_{j=1}^s x(j)\|_k^k$ and sends the transcript to the central coordinator. The central coordinator chooses appropriate vectors $x(s/2+1), \ldots, x(s)$ and using the transcripts from the $s/2$ servers finds a $1+\varepsilon$ approximation to $\|\sum_{j=1}^s x(j)\|_k^k$ by running the $(1/Cn)$-error protocol for estimating $F_k$ moments. 

Let $\|\sum_{j=1}^{s/2} x(j)\|_k^k = T$. Fix an index $i \in [n]$. The central coordinator creates the vectors $x{(s/2 + 1)}, \ldots, x{(s/2 + s/2)}$ to be all be equal and have a value of $2/\varepsilon$ in coordinate $i$ and remaining positions have value $0$. Now consider a NO instance for the set disjointness problem. Then $\|\sum_{j=1}^s x(j)\|_k^k \le (T-1) + (s/\varepsilon + 1)^k$.

Let $T'$ be such that $(1-\varepsilon')T \le T' \le (1+\varepsilon')T$. Then,
\begin{align*}
    (1+\varepsilon')\|\sum_{j=1}^s x(j)\|_k^k \le \frac{1+\varepsilon'}{1-\varepsilon'}T' - (1+\varepsilon') + (1+\varepsilon')(s/\varepsilon + 1)^k.
\end{align*}
Now consider a YES instance. If all the sets intersect in $i$, then $\|\sum_{j=1}^s x(j)\|_k^k = T - (s/2)^k + (s/\varepsilon + s/2)^k$. Now,
\begin{align*}
    (1-\varepsilon')\|\sum_{i=1}^s x(j)\|_k^k \ge \frac{1-\varepsilon'}{1+\varepsilon'}T' - (1-\varepsilon')(s/2)^k + (1-\varepsilon')(s/\varepsilon + s/2)^k.
\end{align*}
If 
\begin{align*}
    \frac{1-\varepsilon'}{1+\varepsilon'}T' - (1-\varepsilon')(s/2)^k + (1-\varepsilon')(s/\varepsilon + s/2)^k > \frac{1+\varepsilon'}{1-\varepsilon'}T' - (1+\varepsilon') + (1+\varepsilon')(s/\varepsilon + 1)^k, 
\end{align*}
we have a test for set disjointness. The above is implied by
\begin{align*}
    (1-\varepsilon')(s/\varepsilon + s/2)^k - (1+\varepsilon')(s/\varepsilon + 1)^k - (1-\varepsilon')(s/2)^k \ge \frac{4\varepsilon'}{1-(\varepsilon')^2}T'
\end{align*}
which is further implied by 
$
    (1-\varepsilon')(s/\varepsilon + s/2)^k - (1+\varepsilon')(s/\varepsilon + 1)^k - (1-\varepsilon')(s/2)^k \ge {8\varepsilon'}T.
$
As $T \le (s/\varepsilon)^k + (s/2)^k$, we obtain that the above is implied by
\begin{align*}
    (1-\varepsilon')(1/\varepsilon + 1/2)^k - (1+\varepsilon')(1/\varepsilon + 1/s)^k - (1-\varepsilon')(1/2^k) \ge 8\varepsilon'(1/\varepsilon^k + 1/2^k).
\end{align*}
For $s \ge 3$,
\begin{align*}
    \left(\frac{1/\varepsilon + 1/2}{1/\varepsilon + 1/s}\right)^k \ge \left(\frac{1/\varepsilon + 1/2}{1/\varepsilon + 1/3}\right)^k \ge (1+\varepsilon/8).
\end{align*}
Hence, setting $\varepsilon' = \varepsilon/C$ for a large enough constant implies that 
\begin{align*}
    (1-\varepsilon')(1/\varepsilon + 1/2)^k - (1+\varepsilon')(1/\varepsilon + 1/s)^k \ge \frac{\varepsilon}{16}(1/\varepsilon + 1/2)^k
\end{align*}
For $k \ge 2$, we further get 
\begin{align*}
    (1-\varepsilon')(1/\varepsilon + 1/2)^k - (1+\varepsilon')(1/\varepsilon + 1/s)^k - (1-\varepsilon')(1/2^k) \ge \frac{\varepsilon}{32}(1/\varepsilon + 1/2)^k.   
\end{align*}
By picking $C$ large enough, we obtain $(\varepsilon/32)(1/\varepsilon + 1/2)^k \ge 8\varepsilon'(1/\varepsilon^k + 1/2^k)$. Thus, if a $1\pm \varepsilon'$ approximation of $\|\sum_{j=1}^s x(j)\|_k^k$ for any $i \in [n]$ (note that the vectors $x(s/2+1), \ldots, x(s)$ depend on which $i$ we are using) exceeds $(1+\varepsilon')T'/(1-\varepsilon') - (1+\varepsilon') + (1+\varepsilon')(s/\varepsilon+1)^k$, then we can output YES to the set disjointness instance and otherwise output NO. Note that we needed to union bound over the $n+1$ instances of the problem, i.e., that we compute $T'$ such that $(1-\varepsilon')T \le T' \le (1+\varepsilon')T$ and later for each $i \in [n]$, we want a $1\pm \varepsilon'$ approximation to the appropriately defined $\|\sum_{j=1}^s x(j)\|_k^k$ and hence we use a $1/Cn$ error protocol.

Thus, any distributed protocol which outputs a $1+\varepsilon/C$ approximation to the $F_k$ approximation problem with probability $\ge 1 - \varepsilon^k/Cs^k$ must use a total communication of $\Omega_k(s^{k-1}/\varepsilon^k)$ bits. Consequently, an algorithm which succeeds with a probability $\ge 9/10$ must use $\Omega_k(s^{k-1}/\log(s/\varepsilon)\varepsilon^k)$ bits of total communication since the success probability of such an algorithm can be boosted to a failure probability $O(\varepsilon^k/s^k)$ by simultaneous independent copies of the protocol.
\end{proof}

%% file: framework.tex
\section{Neighborhood Propagation via Composable Sketches}
We define composable sketches and show how using a neighborhood propagation algorithm, composable sketches can be used so that all nodes in  a graph with arbitrary topology can simultaneously compute statistics of the data in a distance $\Delta$ neighborhood of the node. Typically the distance parameter $\Delta$ is taken to be a small constant but can be as large as the diameter of the underlying graph.

We use $\calA$ to denote a dataset. Each item in the dataset is of the form $(\key, \val)$ where the keys are drawn from an arbitrary set $T$ and the values are $d$-dimensional vectors. We use the notation $\calA.\vals$ to denote the matrix with rows given by the values in the dataset. We say two datasets $\calA$ and $\calB$ are \emph{conforming} if for all $\key$s present in both the datasets, the corresponding $\val$s in both the datasets are the same. We use the notation $\calA \cup \calB$ to denote the union of both the datasets. In the following, we assume that all the relevant datasets are conforming. A composable sketch $\sk(\mathcal{A})$ is a summary of the data items $\mathcal{A}$.
 The sketch $\sk(\cdot)$ must support the following three operations:
 \begin{enumerate}
     \item \textsc{Create}$(\mathcal{A})$: given data items $\mathcal{A}$, generate a sketch $\sk(\mathcal{A})$.
     \item
     \textsc{Merge}$(\sk(\mathcal{A}_1), \sk(\mathcal{A}_2), \cdots, \sk(\mathcal{A}_k))$:
     given the sketches $\sk(\mathcal{A}_1), \cdots, \sk(\mathcal{A}_k)$ for sets of data items  $\mathcal{A}_1,\cdots,\mathcal{A}_k$ which may have overlaps,
     generate a composable sketch $\sk(\mathcal{A}_1 \cup \cdots \cup \mathcal{A}_k)$  for the union of data items $\mathcal{A}_1 \cup \cdots \cup \mathcal{A}_k$.
     \item \textsc{Solve}$(\sk(\mathcal{A}))$: given a sketch $\sk(\mathcal{A})$ of data items $\mathcal{A}$, compute a solution with for a pre-specified problem with respect to $\mathcal{A}$.
\end{enumerate}
Note that $\sk(\calA)$ need not be unique and randomization is allowed during the construction of the sketch and merging. We assume that $\textsc{Create}$, $\textsc{Merge}$ and $\textsc{Solve}$ procedures have access to a shared uniform random bit string.

A core property of the above composable sketch definition is that it handles  duplicates. Consider the following problem over a graph $G = (V, E)$. Each vertex of the graph represents a user/server. For a node $u$, we represent their dataset with $\calS_u$, a set of $(\key,\val)$ pairs. Given a parameter $\Delta$, \emph{each} node in the graph wants to compute statistics or solve an optimization problem over the data of all the nodes within a distance $\Delta$ from the node. For example, with $\Delta = 1$, each node $u$ may want to solve a regression problem defined by the data at node $u$ and all the neighbors $u$.

As composable sketches handle duplicates, the following simple algorithm can be employed to solve the problems over the $\Delta$ neighborhood of each node $u$.
\begin{enumerate}
\item Each node $u$ computes $\sk(\mathcal{S}_u)$ and communicates to all its neighbors. 
\item Repeat $\Delta$ rounds: in round $i$, each node $u$ computes 
\begin{align*}
\sk(\mathcal{S}^i_u)&=\sk(\bigcup_{v:\set{v, u} \in E}\mathcal{S}^{i-1}_v) =\textsc{Merge}(\sk(\mathcal{S}_{v_1}^{i-1}),\cdots, \sk(\mathcal{S}_{v_k}^{i-1}))
\end{align*}
and sends the sketch to all its neighbors.
Here we use $\sk(\calS_u^0)$ to denote $\sk(\calS_u)$.
\item Each node $u$ in the graph outputs a solution over its $\Delta$ neighborhood via first computing $\textsc{Merge}(\sk(\calS_u^{0}), \sk(\calS_u^{1}), \sk(\calS_u^{2}), \ldots, \sk(\calS_u^{\Delta}))$ and then using the $\textsc{Solve}(\cdot)$ procedure.
\end{enumerate}
Notice that the capability of handling duplicates is crucial for the above neighborhood propagation algorithm to work. For example, a node $v$ at a distance $2$ from $u$ maybe connected to $u$ through two disjoint paths and hence $u$ receives the sketch of $u$'s data from two different sources. So it is necessary for the sketch to be duplicate agnostic to not overweight data of vertices that are connected through many neighbors. Another nice property afforded by composable sketches is that a node sends the same ``information'' to all its neighbors meaning that a node does not perform different computations determining what information is to be sent to each of its neighbors. 

In the following section, we give a composable sketch for computing an $\ell_p$ subspace embedding and show that it can be used to solve $\ell_p$ regression problems as well as the low rank approximation problem.

%% file: composable_leverage.tex
\section{Composable Sketches for Sensitivity Sampling}\label{sec:sensitivity-sampling}
We assume $\calA_1,\ldots,\calA_s$ are conforming datasets. Let $\calA \coloneqq \calA_1 \cup \cdots \cup \calA_s$. We give a composable sketch construction such that using $\sk(\calA)$, we can compute an $\ell_p$ subspace embedding for the matrix $\calA.\vals$. Another important objective is to make the size of the sketch $\sk(\calA)$ as small as possible so that sketches can be efficiently communicated to neighbors in the neighborhood propagation algorithm.

Given a matrix $A \in \R^{n \times d}$, we say that a matrix $M \in \R^{m \times d}$ is an $\varepsilon$ $\ell_p$-subspace embedding for $A$ if for all $x \in \R^d$,
\begin{align*}
    \lp{Mx}^p = (1 \pm \varepsilon) \lp{Ax}^p.
\end{align*}
$\ell_p$ subspace embeddings have numerous applications and are heavily studied in the numerical linear algebra literature. We will now define the so-called $\ell_p$ sensitivities and how they can be used to compute subspace embeddings.
\subsection{\texorpdfstring{$\ell_p$}{lp} Sensitivity Sampling}
The $\ell_p$ sensitivities are a straightforward generalization of the leverage scores. Given a matrix $A$ and a row $a$ of the matrix, the $\ell_p$ sensitivity of $a$ w.r.t. the matrix $A$ is defined as
\begin{align*}
    \tau^{\ell_p}_A(a) \coloneqq \max_{x : Ax \ne 0}\frac{|\la a, x\ra|^p}{\lp{Ax}^p}.
\end{align*}
The $\ell_p$ sensitivities measure the importance of a row to be able to estimate $\|Ax\|_p^p$ given any vector $x$. Suppose that a particular row $a$ is orthogonal to all the other rows of the matrix $A$, we can see that $a$ is very important to be able to approximate $\|Ax\|_p^p$ up to a multiplicative factor. It can be shown that if the matrix $A$ has $d$ columns, then the sum of $\ell_p$ sensitivities $\sum_{a \in A}\tau_A^{\ell_p}(a) \le d^{\max(p/2, 1)}$ \cite{MMWY22}. Now we state the following sampling result which shows that sampling rows of the matrix $A$ with probabilities depending on the sensitivities and appropriately rescaling the sampled rows gives an $\ell_p$ subspace embedding.
\begin{theorem}
	Given a matrix $A$ and a vector $v \in [0,1]^n$ such that for all $i \in [n]$, $v_i \ge \beta \tau_A^{\ell_p}(a_i)$ for some $\beta \le 1$, let a random diagonal matrix $\bS$ be generated as follows: for each $i \in [n]$ independently, set $\bS_{ii} = (1/p_i)^{1/p}$ with probability $p_i$ and $0$ otherwise. If $p_i \ge \min(1, C_1\beta^{-1}v_i(C_2d\log(d/\varepsilon) + \log(1/\delta))/\varepsilon^2)$ for large enough constants $C_1$ and $C_2$, then with probability $\ge 1 - \delta$, for all $x \in \R^d$,
	\begin{align*}
		\lp{\bS Ax} = (1 \pm \varepsilon)\lp{Ax}.
	\end{align*}
	\label{thm:sensitivity-sampling}
\end{theorem}
Given constant factor approximations for the $\ell_p$ sensitivities we can define the probabilities $p_i$ such that the matrix $\bS$ has at most $O(d^{\max(p/2,1)}(d\log (d/\varepsilon)+ \log 1/\delta)/\varepsilon^2)$ non-zero entries with a large probability. The proof of the above theorem proceeds by showing that for a fixed vector $x$, the event $\lp{\bS A x} = (1 \pm \varepsilon)\lp{Ax}$ holds with a high probability and then using an $\varepsilon$-net argument to extend the high probability guarantee for a single vector $x$ to a guarantee for all the vectors $x$. For $p = 2$, we can show that in the above theorem $p_i \ge \min(1, C_1\beta^{-1}v_i(C_2\log(d/\varepsilon) + \log 1/\delta)\varepsilon^{-2})$ suffices to construct a subspace embedding. So, in all our results for the special case of $2$, only $\tilde{O}(d)$ rows need to be sampled.

We will now show a construction of a composable sketch $\sk(\calA)$ given a dataset $\calA$. The composable sketch $\sk(\calA)$ can be used to construct an $\ell_p$ subspace embedding for the matrix $\calA.\vals$. Importantly, we note that given composable sketches $\sk(\calA)$ and $\sk(\calB)$, the sketches can be merged only when $\calA$ and $\calB$ are conforming and the sketches $\sk(\calA)$ and $\sk(\calB)$ are constructed using the same randomness in a way which will become clear after we give the sketch construction.

We parameterize our sketch construction with an integer parameter $t$ that defines the number of times a sketch can be merged with other sketches. We denote the sketch by $\sk_t(\calA)$ if it is ``mergeable'' $t$ times. Merging $\sk_t(\calA)$ and $\sk_{t'}(\calA)$ gives $\sk_{\min(t,t')-1}(\calA \cup \calB)$. Naturally, the size of the sketch increases with the parameter $t$. We will first show how the sketch $\sk_{t}(\calA)$ is created.
\subsection{Sketch Creation}
Given a dataset $\calA$ and a parameter $t$, we pick $t$ independent \emph{fully random} hash functions $h_1,\ldots,h_t$ mapping keys to uniform random variables in the interval $[0,1]$. Thus for each $\key$, the value $h_i(\key)$ is an independent uniform random variable in the interval $[0,1]$. Given such hash functions $h_1,h_2,\ldots,h_t$, first for each $(\key, \val) \in \calA$, we compute $\tilde{\tau}_{\key}$ that satisfies 
    \begin{align*}
         (1+\varepsilon)^t \tau^{\ell_p}_{\calA.\vals}(\val) \le \tilde{\tau}_{\key} \le (1+\varepsilon)^{t+1}\tau_{\calA.\vals}^{\ell_p}(\val).
    \end{align*}
Note that we are free to choose $\tilde{\tau}_{\key}$ to be any value in the above interval. To allow randomness in computing the values of $\tilde{\tau}_{\key}$, we introduce another parameter $\gamma$. We assume that with probability $1 - \gamma$, for all $\key \in \calA.\text{keys}$, $\tilde{\tau}_{\key}$ satisfies the above relation. When creating the sketch from scratch, as we can compute exact $\ell_p$ sensitivities, we can take $\gamma$ to be $0$. The only requirement is that the value of $\tilde{\tau}_{\key}$ must be computed independently of the hash functions $h_1,\ldots,h_t$. 

For each $\key \in \calA.\text{keys}$, let $p_{\key} = C\tilde{\tau}_{\key}(d\log d/\varepsilon + \log 1/\delta)\varepsilon^{-2}$ and now for each $h_i$, define
    \begin{align*}
        \levSample(\calA, h_i, t) \coloneqq \setbuilder{(\key, \val, \min(p_{\key}, 1))}{(\key, \val)\in \calA, h_i(\key) \le p_{\key}}.
    \end{align*}
The sketch $\sk_{t,0}(\calA)$ is now defined to be the collection $(\levSample(\calA, h_1, t), \ldots, \levSample(\calA, h_t, t))$. The procedure is described in Algorithm~\ref{alg:sketch-creation}. Note that for each $(\key, \val) \in \calA$, 
\begin{align*}
    \Pr_{h_i}[(\key, \val, *) \in \levSample(\calA, h_i, t)] &= \min(p_{\key}, 1)\\
    &\ge \min(C\tau_{\calA.\vals}^{\ell_p}(\val)(d\log (d/\varepsilon) + \log 1/\delta)\varepsilon^{-2}, 1).
\end{align*}
Hence, the construction of the set $\levSample(\calA, h_i, t)$ is essentially performing $\ell_p$ sensitivity sampling as in Theorem~\ref{thm:sensitivity-sampling} and for sampled rows it also stores the probability with which they were sampled. Thus, a matrix constructed appropriately using $\levSample(\calA, h_i, t)$ will be a subspace embedding for the matrix $\calA.\vals$ with probability $\ge 1 - \delta$.

Throughout the construction, we ensure that the sketch $\sk_{t, \gamma}(\calA) = (\levSample(\calA, h_1, t)$, \ldots, $\levSample(\calA, h_t, t))$ satisfies the following definition.
\begin{definition}
A sketch $(\levSample(\calA, h_1, t), \ldots, \levSample(\calA, h_t, t))$ is denoted $\sk_{t,\gamma}(\calA)$ if with probability $\ge 1 - \gamma$ (over randomness independent of $h_1,\ldots,h_t$), for each $(\key, \val) \in \calA$, there exist values $\tilde{\tau}_{\key}$ (computed independently of the hash functions $h_1,\ldots,h_t$) such that
\begin{align}
     (1+\varepsilon)^t\tau^{\ell_p}_{\calA.\vals}(\val) \le \tilde{\tau}_{\key} \le (1+\varepsilon)^{t+1} \tau^{\ell_p}_{\calA.\vals}(\val)
     \label{eqn:tau-between-powers}
\end{align}
and for $p_{\key} = C\tilde{\tau}_{\key}(d\log d/\varepsilon + \log 1/\delta)\varepsilon^{-2}$,
\begin{align}
        \levSample(\calA, h_i, t) = \setbuilder{(\key, \val, \min(p_{\key}, 1))}{(\key, \val)\in \calA, h_i(\key) \le p_{\key}}.
        \label{eqn:sensample-as-function-of-tau}
\end{align}
\end{definition}
Note that using the bounds on the sum of $\ell_p$ sensitivities, we obtain that with probability $\ge 1 - \gamma - \exp(-d)$, the size of the sketch $\sk_{t,\gamma}(\calA)$ is $O(t(1+\varepsilon)^{t+1}d^{\max(p/2,1)}(d\log d/\varepsilon + \log 1/\delta)\varepsilon^{-2})$. By Theorem~\ref{thm:sensitivity-sampling}, we obtain that given a sketch $\sk_{t,\gamma}(\calA)$, Algorithm~\ref{alg:compute-subspace-embedding} computes a subspace embedding for the matrix $\calA.\vals$. Thus, we have the following theorem.
\begin{theorem}
    Given $\sk_{t,\gamma}(\calA)$ constructed with parameters $\varepsilon, \delta$, Algorithm~\ref{alg:compute-subspace-embedding} returns a matrix that with probability $\ge 1 - \gamma - \delta$ satisfies, for all $x$,
    \begin{align*}
        \|Mx\|_p^p = (1 \pm \varepsilon)\|\calA.\vals \cdot x\|_p^p.
    \end{align*}
\end{theorem}

We now show how to compute a sketch for $\calA_1 \cup \cdots \cup \calA_s$ given sketches $\sk_{t_1,\gamma_1}(\calA_1), \ldots, \sk_{t_s, \gamma_s}(\calA_s)$ for $s$ conforming datasets $\calA_1,\ldots,\calA_s$.

\subsection{Merging Sketches}
\begin{theorem}
Let $\calA_1,\ldots,\calA_s$ be conforming datasets. Given sketches $\sk_{t_1,\gamma_1}(\calA_1), \ldots,\sk_{t_s,\gamma_s}(\calA_s)$ constructed using the same hash functions $h_1,\ldots$ and parameters $\varepsilon, \delta > 0$, Algorithm~\ref{alg:merging-sketches} computes $\sk_{\min_i(t_i)-1, \delta+\gamma_1+\cdots+\gamma_s}(\calA_1 \cup \cdots \cup \calA_s)$. 
\end{theorem}
\subsubsection{Proof Outline}
In the original sketch creation procedure, we compute approximations to the $\ell_p$ sensitivities which we use to compute a value $p_{\key}$ and keep all the $(\key, \val)$ pairs satisfying $h(\key) \le p_{\key}$. We argued that the original sketch creation is essentially an implementation of the $\ell_{p}$ sensitivity sampling algorithm in Theorem~\ref{thm:sensitivity-sampling}. Now, given sketches of $\calA_1,\ldots,\calA_s$, we want to \emph{simulate} the $\ell_p$ sensitivity sampling of the rows in the matrix $(\calA_1 \cup \cdots \cup \calA_s).\vals$ to create the sketch $\sk(\calA_1 \cup \cdots \cup \calA_k)$. An important property of the $\ell_p$ sensitivities is the \emph{monotonicity} -- the $\ell_p$ sensitivity of a row only goes down with adding new rows to the matrix. Suppose we have a way to compute $\ell_p$ sensitivities of the rows of the matrix $(\calA_1 \cup \cdots \cup \calA_s).\vals$. Suppose a row $a \in \calA_{1}.\vals$. Then the probability that it has to be sampled when performing $\ell_p$ sensitivity sampling on the matrix $(\calA_1 \cup \cdots \cup \calA_s).\vals$ is smaller than the probability that the row has to be sampled when performing $\ell_p$ sensitivity sampling on the matrix $\calA_1.\vals$. Thus, the rows that we ignored when constructing $\sk(\calA_1)$ ``don't really matter'' as the $\ell_p$ sensitivity sampling of the rows of $(\calA_1 \cup \cdots \cup \calA_s).\vals$ when performing using the same hash function $h$ would also not have sampled that row since $h(\key)$ was already larger than the probability that $\calA_1$ assigned to the row $a$ which is in turn larger than the probability that $\calA_1 \cup \cdots \cup \calA_t$ assigned to the row $a$.

The above argument assumes that we have a way to approximate the $\ell_p$ sensitivity of a row with respect to the matrix $\calA_1 \cup \cdots \cup \calA_s$ and sensitivity sampling requires that these approximations be independent of the hash function $h$ we are using to simulate sensitivity sampling. We now recall that each $\sk_{t,\gamma}(\calA)$ has $t$ independent copies of the $\levSample$ data structure. We show that one of the copies can be used to compute approximate sensitivities and then perform the $\ell_p$ sensitivity sampling on the other copies. Thus, each time we merge a $\sk_{t,\gamma}(\cdot)$ data structure, we lose a copy of the $\levSample$ data structure in the sketch which is why the sketch $\sk_{t,\gamma}(\cdot)$ can be merged only $t$ times in the future.
\subsubsection{Formal Proof}
\begin{proof}
Let $\calA \coloneqq \calA_1 \cup \cdots \cup \calA_s$ and $t = \min(t_1,\ldots,t_s)$. Recall that each $\sk_{t_j, \gamma_j}(\calA_j)$ is a collection of the data structures $\levSample(\calA_j, h_1, t_j), \ldots, \levSample(\calA_j, h_{t_j}, t_j)$ and that by definition of $\sk_{t,\gamma}(\calA)$, for each $j=1,\ldots,s$, with probability $1 - \gamma_j$ (over independent randomness $h_1,\ldots,h_{t_j}$) for each $(\key, \val) \in \calA_j$, there exists $\tilde{\tau}^{(j)}_{\key}$ for which
\begin{align}    
(1+\varepsilon)^{t_j}\tau^{\ell_p}_{\calA_j.\vals}(\val) \le \tilde{\tau}^{(j)}_{\key} \le (1+\varepsilon)^{t_j+1}\tau^{\ell_p}_{\calA_j.\vals}(\val)
\label{eqn:satisfied-by-tau}
\end{align}
and for $p_{\key} = C\tilde{\tau}_{\key}^{(j)}(d\log (d/\varepsilon) + \log 1/\delta)\varepsilon^{-2}$ and $i=1,\ldots,t_j$,
\begin{align*}
        \levSample(\calA_j, h_i, t) = \setbuilder{(\key, \val, \min(p_{\key}, 1))}{(\key, \val)\in \calA, h_i(\key) \le p_{\key}}.
\end{align*}
By a union bound, with probability $\ge 1 - (\gamma_1 + \cdots + \gamma_s)$, we have $\tilde{\tau}^{(j)}_{\key}$ as in \eqref{eqn:satisfied-by-tau} for all $j = 1,\ldots,s$ and $\key \in \calA_j.\keys$. Condition on this event. 

We now show that the matrix $M$ constructed by the algorithm is a subspace embedding for $(\calA_1 \cup \cdots \cup \calA_s).\vals$. Note that, in constructing the matrix $M$, the algorithm uses $\levSample$ data structures all constructed using the same hash function $h_t$.

If $(\key, \val, *)$ is in \emph{any} of the sets $\levSample(\calA_1, h_t, t_1), \ldots, \levSample(\calA_s, h_t, t_s)$, let $p_{\key}^{\text{merge}}$ be the maximum ``probability value'' among all the tuples with $(\key, \val, *)$. Let $S$ be the set formed by all the tuples $(\key, \val, p_{\key}^{\text{merge}})$. For each $(\key, \val) \in \calA$, define
\begin{align*}
    \tilde{\tau}_{\key}^{\text{merge}} = \max_{(\key, \val) \in \calA_j}\tilde{\tau}^{(j)}_{\key}.
\end{align*}
Now, for each $(\key, \val) \in \calA_1 \cup \cdots \cup \calA_s$,
\begin{align*}
    \Pr[(\key, \val, *) \in S] &= \Pr[h_t(\key) \le \max_{j : (\key, \val) \in \calA_j}C\tilde{\tau}^{\text{merge}}_{\key}(d\log d/\varepsilon + \log 1/\delta)\varepsilon^{-2}]\\
    &=p_{\key}^{\text{merge}}.
\end{align*}
By monotonicity of $\ell_p$ sensitivities, if $(\key, \val) \in \calA_j$, then
\begin{align*}
    \tau^{\ell_p}_{(\calA_1 \cup \cdots \cup \calA_s).\vals}(\val) \le \tau^{\ell_p}_{\calA_j.\vals}(\val) \le \tilde{\tau}^{(j)}_{\key} \le \tilde{\tau}_{\key}^{\text{merge}}.
\end{align*}
Hence, with probability $\ge 1 - \delta$ the set $S$ is a leverage score sample of the rows of the matrix $\calA.\vals$. By a union bound, with probability $\ge 1 - (\delta + \gamma_1 + \cdots + \gamma_s)$, the matrix $M$ with rows given by $1/(p_{\key}^{\text{merge}})^{1/p} \cdot \val$ for $(\key, \val, p_{\key}^{\text{merge}}) \in S$ is an $\ell_p$ subspace embedding for the matrix $\calA.\vals$ and satisfies for all $x$,
\begin{align*}
    \lp{Mx}^p = (1 \pm \varepsilon/4)\lp{\calA.\vals \cdot x}^p.
\end{align*}
For each $(\key, \val) \in \calA$, we can compute
\begin{align*}
    \tilde{\tau}_{\key}^{\text{approx}} = (1+\varepsilon)^{t-1}(1+\varepsilon/4)\max_{x}\frac{|\la \val, x\ra|^p}{\lp{Mx}^p}.
\end{align*}
Conditioned on $M$ being a subspace embedding for $\calA$, we have that 
$(1+\varepsilon)^{t-1}\tau_{\calA.\vals}^{\ell_p}(\val) \le \tilde{\tau}_{\key}^{\text{approx}} \le (1+\varepsilon)^t \tau_{\calA.\vals}^{\ell_p}(\val)$. For each $(\key,\val) \in \calA_j$, we have
\begin{align}
    \tilde{\tau}_{\key}^{\text{approx}} \le (1+\varepsilon)^t \tau^{\ell_p}_{\calA.\vals}(\val) \le (1+\varepsilon)^t \tau^{\ell_p}_{\calA_j.\vals}(\val) \le \tilde{\tau}^{(j)}_{\key} \le \tilde{\tau}^{\text{merge}}_{\key}.
    \label{eqn:using-different-powers-of-t}
\end{align}
Thus, with probability $\ge 1 - (\delta + \gamma_1 + \cdots + \gamma_s)$, for all $(\key, \val) \in \calA$, 
\begin{align*}
    (1+\varepsilon)^{t-1}\tau^{\ell_p}_{\calA.\vals}(\val) \le \tilde{\tau}_{\key}^{\text{approx}} \le (1+\varepsilon)^t \tau^{\ell_p}_{\calA.\vals}(\val).
\end{align*}
Now we define $\tilde{p}_{\key} = C\tilde{\tau}^{\text{approx}}_{\key}(d\log d/\varepsilon + \log 1/\delta)\varepsilon^{-2}$ and
\begin{align*}
    \levSample(\calA, h_i, t-1) = \setbuilder{(\key, \val, \min(1, \tilde{p}_{\key}))}{(\key, \val) \in \calA, h_i(\key) \le \tilde{p}_{\key}}
\end{align*}
and have
\begin{align*}
    \Pr_{h_i}[(\key, \val, *) \in \levSample(\calA, h_i, t-1)] = \min(1, \tilde{p}_{\key}).
\end{align*}
Note that while the above definition says to construct the set by looking at each $(\key, \val) \in \calA$, as $\tilde{\tau}_{\key}^{\text{approx}} \le \max_{j:(\key, \val) \in \calA_j}\tilde{\tau}^{(j)}_{\key}$  by definition, we only have to look at the elements of the set $\levSample(\calA_1, h_i, t-1), \ldots, \levSample(\calA_s, h_i, t-1)$ as all other missing elements from $\calA$ would not have been included in the set anyway. Here the property that the $\tilde{\tau}$ values satisfy \eqref{eqn:satisfied-by-tau} becomes crucial. 

Thus, we have that the algorithm constructs $\sk_{t-1,\delta+\gamma_1+\cdots+\gamma_s}(\calA)$. 
\end{proof}
\subsection{Neighborhood Propagation}
As described in the previous section, the neighborhood propagation algorithm using the composable sketches lets each node compute a subspace embedding for the matrix formed by the data of the matrices in a neighborhood around the node. We will now analyze the setting of the $\delta$ parameter in the $\ell_p$ composable sketch construction. 

We have that merging the sketches $\sk_{t_1,\gamma_1}(\calA_1), \ldots, \sk_{t_s, \gamma_s}(\calA_s)$, we obtain $\sk_{\min_i t_i - 1, \delta + \gamma_1 + \cdots + \gamma_s}(\calA_1 \cup \cdots \cup \calA_s)$. Let $s$ be the total number of nodes in the graph. The sketches that each neighborhood obtains are merged at most $\Delta$ times. Hence setting $\delta = \delta'/(2s)^{\Delta}$, each node in the graph computes a sketch for the data in its neighborhood with the probability parameter $\delta'$. Further setting $\delta' = 1/10s$, we obtain by a union bound that with probability $\ge 9/10$, all the nodes in the graph compute an $\ell_p$ subspace embeddings for the data in their $\Delta$ neighborhoods. Thus we have the following theorem.
\begin{theorem}
    Suppose $G = (V,E)$ is an arbitrary graph with $|V| = s$. Each node in the graph knows and can communicate only with its neighbors. Given a distance parameter $\Delta$ and accuracy parameter $\varepsilon < 1/\Delta$, there is a neighborhood propagation algorithm that runs for $\Delta$ rounds such that at the end of the algorithm, with probability $\ge 9/10$, each vertex $u$ in the graph computes an $\varepsilon$ $\ell_p$ subspace embedding for the matrix formed by the data in the $\Delta$ neighborhood of $u$.

    In each of the $\Delta$ rounds, each node communicates at most $O(\Delta \cdot d^{\max(p/2,1)}(d\log d + \Delta \log s)\varepsilon^{-2})$ rows along with additional information for each row to all its neighbors. For $p=2$, each node communicates $O(\Delta \cdot d(\log d + \Delta \log s)\varepsilon^{-2})$ rows to each of its neighborts in each round.
    \label{thm:main-theorem-congest-embeddings}
\end{theorem}
Since in many problems of interest, the parameter $\Delta$ is a small constant, the algorithm is communication efficient.

\begin{algorithm}
\caption{Creating the sketch $\sk_{t, 0}$ given $\calA$}\label{alg:sketch-creation}
\KwIn{A dataset $\calA$ of pairs $(\key, \val)$, an integer parameter $t \ge 1$, $\varepsilon$, $\delta$}
\KwOut{A sketch $\sk_{t,0}(\calA)$}
Let $h_1,\ldots,h_t$ be independent fully random hash functions with $h_i(\key)$ being a uniform random variable from $[0,1]$\;
For each $(\key, \val) \in \calA$, $\tau^{\ell_p}_{\calA.\vals}(\val) \gets \max_{x}|\la \val, x\ra|^p/\|\calA.\vals \cdot x\|_p^p$\;
For each $(\key, \val) \in \calA$, $p_{\key} \gets C\tau^{\ell_p}(d\log d + \log 1/\delta)\varepsilon^{-2}$\;
\For{$i=1,\ldots,t$}{
    $\levSample(\calA, h_i, t) \gets \emptyset$\;
    \For{$(\key, \val) \in \calA$}{
        \If{$h_i(\key) \le p_\key$}{
            $\levSample(\calA, h_i, t) \gets \levSample(\calA, h_i, t) \cup \set{(\key, \val, \min(1, p_{\key}))}$\;
        }
    }
}
$\sk_{t, 0}(\calA) \gets (\levSample(\calA, h_1, t), \ldots, \levSample(\calA, h_t, t))$\;
\end{algorithm}

\begin{algorithm}
    \caption{Computing a subspace embedding from a sketch}\label{alg:compute-subspace-embedding}
    \KwIn{Sketch $\sk_{t,\gamma}(\calA)$ constructed with parameters $\varepsilon, \delta$}
    \KwOut{A matrix $M$ that is an $\varepsilon$ $\ell_p$ subspace embedding}
    Note $\sk_{t,\gamma}(\calA) = (\levSample(\calA), h_1, t), \ldots, \levSample(\calA), h_t, t))$\;
    $M \gets $ matrix with rows given by $(1/p_{\key})^{1/p} \cdot \val$ for $(\key, \val, p_{\key}) \in \levSample(\calA, h_1, t)$\;
    \Return{$M$}
\end{algorithm}

\begin{algorithm}
    \caption{Merging Sketches}\label{alg:merging-sketches}
    \KwIn{Sketches $\sk_{t_1, \gamma_1}(\calA_1), \ldots, \sk_{t_s, \gamma_s}(\calA_s)$ constructed with the same parameters $\varepsilon, \delta$ and the same hash functions $h_1,\ldots,$}
    \KwOut{Sketch $\sk_{\min_it_i-1,\delta+\sum_i \gamma_i}(\calA_1 \cup \cdots \calA_k)$}
    Let $h_1, h_2, \ldots,$ be the hash functions used in the construction of the sketches\;
    $t \gets \min_i t_i$\;
    $\calA \gets \calA_1 \cup \cdots \cup \calA_s$\tcp*{Only notational}
    $\text{merge} \gets \setbuilder{(\key, \val)}{\exists j \in [s], (\key, \val, *) \in \levSample(\calA_j, h_t, t_j)}$\;
    For each $(\key, \val) \in \text{merge}$, $p_{\key}^{\text{merge}} \gets $ max $p$ with $(\key, \val, p) \in \cup_j \levSample(\calA_j, h_t, t_j)$\;
    $M \gets $ matrix with rows given by $(1/p_{\key}^{\text{merge}})^{1/p} \cdot \val$ for $(\val, \key) \in \text{merge}$\;
    \For{$i=1,\ldots,t-1$}{
        $\levSample(\calA, h_i, t-1) \gets \emptyset$\;
        $\text{merge}_i \gets \setbuilder{(\key, \val)}{\exists j \in [s], (\key, \val, *) \in \levSample(\calA_j, h_i, t_j)}$\;
        For each $(\key, \val) \in \text{merge}_i$, $p^{(i)}_{\key} \gets $ max $p$ with $(\key, \val, p) \in \cup_j \levSample(\calA_j, h_i, t_j)$\;
        \For{$(\key, \val) \in \textnormal{merge}_i$}{
            $\tilde{\tau}_{\key}^{\text{approx}} \gets (1+\varepsilon)^{t-1}(1+\varepsilon/4)\max_x \frac{|\la \val, x \ra|^p}{\|Mx\|_p^p}$\;
            $p_{\key} \gets C\tilde{\tau}^{\text{approx}}_{\key}(d\log d/\varepsilon + \log 1/\delta)\varepsilon^{-2}$\;
            \If{$\min(1,p_{\key}) > p_{\key}^{(i)}$}{
                Output FAIL\;
            }
            \If{$h_i(\key) \le p_{\key}$}{
            $\levSample(\calA, h_i, t-1) \gets \levSample(\calA, h_i, t-1) \cup \set{(\key, \val, \min(1, p_{\key}))}$\;
            }
        }
    }
    $\sk_{t-1, \delta+\gamma_1+\cdots+\gamma_s}(\calA) \gets (\levSample(\calA, h_1, t-1), \ldots, \levSample(\calA, h_{t-1}, t-1))$\;
\end{algorithm}
\subsection{Applications to \texorpdfstring{$\ell_p$}{lp} Regression}\label{subsec:congest-regression}
Let $\calA$ be a dataset. In a $(\key, \val)$ pair with $\val$ being a $d$ dimensional vector, we treat the first $d-1$ coordinates as the features and the last coordinate as the label. Then the $\ell_p$ linear regression problem on a dataset $\calA$ is
\begin{align*}
    \min_{x \in \R^{d-1}}\lp{\calA.\vals \begin{bmatrix}x \\ -1\end{bmatrix}}^p.
\end{align*}
Thus, if the matrix $M$ is an $\varepsilon$ subspace embedding for the matrix $\calA.\vals$, then
\begin{align*}
    \tilde{x} = \argmin_x \lp{M \begin{bmatrix}x \\ -1\end{bmatrix}}^p,
\end{align*}
then 
\begin{align*}
    \lp{\calA.\vals \begin{bmatrix}\tilde{x} \\ -1\end{bmatrix}}^p \le (1+O(\varepsilon))\min_x\lp{\calA.\vals \begin{bmatrix}x \\ -1\end{bmatrix}}.
\end{align*}
Thus, composable sketches for constructing $\ell_p$ subspace embeddings can be used to solve $\ell_p$ regression problems.
\subsection{Low Rank Approximation}\label{subsec:congest-lra}
We consider the Frobenius norm low rank approximation. Given a matrix $A$, a rank parameter $k$ we want to compute a rank $k$ matrix $B$ such that $\frnorm{A - B}^2$ is minimized. The optimal solution to this problem can be obtained by truncating the singular value decomposition of the matrix $A$ to its top $k$ singular values. As computing the exact singular value decomposition of a matrix $A$ is slow, the approximate version of low rank approximation has been heavily studied in the literature \cite{clarkson2017low}. In the approximate version, given a parameter $\varepsilon$, we want to compute a rank-$k$ matrix $B$ such that
\begin{align*}
    \frnorm{A - B}^2 \le (1+\varepsilon)\min_{\text{rank-}k\, B}\frnorm{A-B}^2.
\end{align*}
As the number of rows in $A$ is usually quite large, the version of the problem which asks to only output a $k$ dimensional subspace $V$ of $\R^d$ is also studied:
\begin{align*}
    \frnorm{A(I-\Proj_V)}^2 \le (1+\varepsilon)\min_{\text{rank-}k\, B}\frnorm{A-B}^2.
\end{align*}
Here $\Proj_V$ denotes the orthogonal projection matrix onto the subspace $V$.

We show that using composable sketches for $\ell_2$ sensitivity sampling, we can solve the low rank approximation problem. While the composable sketch for $\ell_2$ sensitivity sampling has $\tilde{O}(d)$ rows, we will show that for solving the low rank approximation problem, the composable sketch need only have $\tilde{O}(k)$ rows. We use the following result. 
\begin{theorem}[{\cite[Theorem~4.2]{clarkson2009numerical}}]
    If $A$ is an $n \times d$ matrix and $\bR$ is a $d \times m$ random sign matrix for $m = O(k\log(1/\delta)/\varepsilon)$, then with probability $\ge 1 - \delta$,
    \begin{align*}
        \min_{\text{rank-}k\, X}\frnorm{A\bR X - A}^2 \le (1+\varepsilon)\min_{\text{rank-}k\, B}\frnorm{A - B}^2.
    \end{align*}
\end{theorem}
Using the affine embedding result of \cite{clarkson2017low}, if $\bL$ is now a leverage score sampling matrix, meaning that $\bL$ is a diagonal matrix with the entry $1/\sqrt{p_i}$ for the rows that are sampled by the leverage score sampling algorithm as in Theorem~\ref{thm:sensitivity-sampling}, then for all matrices $X$ 
\begin{align*}
    \frnorm{\bL A\bR X - \bL A}^2 = (1 \pm \varepsilon)\frnorm{A\bR X - A}^2.
\end{align*}
Hence, if 
\begin{align*}
    \tilde{X} = \argmin_{\text{rank-}k\, X}\frnorm{\bL A\bR X - \bL A }^2,
\end{align*}
then $\frnorm{A\bR\tilde{X} - A}^2 \le (1+O(\varepsilon))\min_{\text{rank-}k\, B}\frnorm{A-B}^2$ which implies that $\frnorm{A(I-\Proj_{\text{rowspace}(\bR\tilde{X})})}^2 \le (1+O(\varepsilon))\min_{\text{rank-k}\, B}\frnorm{A-B}^2$.

Thus, if $\bR$ is a random sign matrix with $O(k\log(1/\delta)/\varepsilon)$ rows, then $\sk_{t,\gamma}(\calA.\vals \cdot \bR)$, along with the corresponding rows in $\calA.\vals$ for the rows in $\sk_{t,\gamma}(\calA.\vals \cdot \bR)$, can be used to compute a $1+\varepsilon$ approximation to the low rank approximation problem. As the matrix $\calA.\vals \cdot \bR$ has only $\tilde{O}(k)$ rows, the composable sketch $\sk_{t, \gamma}(\calA. \vals \cdot \bR)$ has a number of rows that depends only on $k$ as well.

%% file: appendix.tex
\section{Gap in the analysis of \texorpdfstring{\cite{kannan2014principal}}{KVW}}\label{app:bug}
In Theorem~1.6 of \cite{kannan2014principal}, the authors claim an $F_k$ estimation algorithm that uses $\tilde{O}(\varepsilon^{-3}(s^{k-1} + s^3)(\ln s)^3)$ bits of total communication. In the proof of Theorem~1.6, in the inequalities used to bound the quantity $\E(Y^2)/(\E(Y))^2$, the last inequality seems to use that $\rho_i \ge \beta/e^{\varepsilon}$ but the inequality holds only when $i \in S_{\beta}$ (in their notation). But the question of if $i \in S_{\beta}$ is exactly what they are trying to find out from the analysis and hence it cannot be assumed that $\rho_i \ge \beta/e^{\varepsilon}$.
\section{Huber Loss Function}
Given a parameter $\tau$, the Huber loss function $f$ is defined as $f(x) = x^2/(2\tau)$ if $|x| \le \tau$ and $f(x) = |x| - \tau/2$ if $|x| \ge \tau$. In this work we consider only the values of $x \ge 0$. We will now examine various properties of the Huber loss function.
\subsection{Super-additivity}
One can verify that the huber loss function is convex and $f(0) = 0$. Consider arbitrary $x, y \ge 0$. Since $f$ is convex in the interval $[0, x+y]$, we get
\begin{align*}
    f(x) \le \frac{x}{x+y}f(x+y)\quad
    \text{and}\quad
    f(y) \le \frac{y}{x+y}f(x+y).
\end{align*}
Adding both the inequalities, we get $f(x) + f(y) \le f(x+y)$. Thus we have the following lemma:
\begin{lemma}
    If $f$ is convex and $f(0) = 0$, then for any $x, y \ge 0$, we have $f(x+y) \ge f(x) + f(y)$.
\end{lemma}
\subsection{Bounding \texorpdfstring{$c_{f}[s]$}{cfs}}
We note that when $f$ denotes the Huber loss function with parameter $\tau$, the function $\sqrt{f}$ is concave on the interval $[0, \infty)$. Now consider arbitrary $x_1, \ldots, x_s \ge 0$. By concavity of $\sqrt{f}$ in the interval $[0, x_1 + \cdots + x_s]$, we get
\begin{align*}
    \sqrt{f(x_j)} \ge \frac{x_j}{x_1 + \cdots + x_s}\sqrt{f(x_1 + \cdots + x_s)}
\end{align*}
for all $j = 1, \ldots, s$. By adding all the inequalities,
\begin{align*}
    \sqrt{f(x_1)} + \cdots + \sqrt{f(x_s)} \ge \sqrt{f(x_1 + \cdots + x_s)}
\end{align*}
which implies
\begin{align*}
    f(x_1 + \cdots + x_s) \le \left(\sqrt{f(x_1)} + \cdots + \sqrt{f(x_s)}\right)^2
\end{align*}
and therefore that $c_f[s] \le s$ when $f$ is the Huber loss function.
\section{Derandomizing Exponential Random Variables using Nisan's PRG}\label{sec:nisan}
Our algorithm for estimating higher-order correlations assumes that we have access to $O(n^k)$ independent exponential random variables, which raises the question how these are stored since the protocol later requires the values of these random variables. We now argue that Nisan's PRG \cite{nisan1992pseudorandom} can be used to derandomize the exponential random variables and that all the required exponential random variables can be generated using a short seed of length $O(k\log^2 (n/\varepsilon))$. 

Note that using $O(k\log(n/\varepsilon))$ bits of precision, we can sample from a discrete distribution that approximates the continuous exponential random variables up to a $1 \pm \varepsilon$ factor since by a simple union bound if we sample $O(n^k/\varepsilon^2)$ exponential random variables, then they all lie in the interval $[\poly(\varepsilon) \cdot n^{-O(k)}, O(k\log n/\varepsilon)]$ with a $1-1/\poly(n)$ probability. Let $b = O(k\log (n/\varepsilon))$. We can use $b$ uniform bits to sample from this discrete distribution and store the sampled discrete random variables using $b$ bits as well while ensuring that the discrete random variable has all the properties we use of the continuous random variable. Let $\be$ be the random variable drawn from this discrete distribution and let $\be_1, \ldots, \be_n$ independent copies of this discrete random variable. Since the distribution of $\be$ is obtained by discretizing the continuous exponential random variable into powers of $1 + \varepsilon/4$, we have that $\max(f_i/\be_i)$ has the same distribution of $(\sum_i f_i)/\be$ up to a $1 \pm \varepsilon/4$ factor and with probability $\ge 1-1/\poly(n)$, $\sum_i \be_i^{-1}f_i \le (C\log^2 n) \cdot \max_i f_i/\be_i$ still holds with a slightly larger value of $C$.

To run the protocol for approximating higher-order correlations, we need to generate the same $r' = O(\binom{n}{k} \cdot k! \cdot \varepsilon^{-2})$ exponential random variables at all the servers and the coordinator. Suppose that the exponential random variables are generated using a random string of length $b \cdot r'$ as follows: we use the first $b \cdot \binom{n}{k} \cdot k!$ bits to generate the first set of exponential random variables, one for each coordinate of the form $(i_1, \ldots, i_k)$ for distinct $i_1, \ldots, i_k \in [n]$. We use the second set of $b \cdot \binom{n}{k}k!$ random bits to generate second set of exponential random variables and so on we generate $m = O(1/\varepsilon^2)$ sets of exponential random variables necessary for implementing the protocol. Let $\be^{(t)}_{(i_1, i_2, \ldots, i_k)}$ be the discrete exponential random variable corresponding to the coordinate $(i_1, \ldots, i_k)$ in the $t$-th set of random variables.

Let $f$ be a vector with coordinates of the form $(i_1, \ldots, i_k)$ for distinct $i_1, \ldots, i_k \in [n]$. Now consider the following simple ``small-space'' algorithm \textsf{Alg}. It makes a pass on the length $b \cdot r$ string reading $b$ blocks at a time and maintains the following counts:
\begin{enumerate}
    \item \texttt{CountLess}: The number of values $t$ such that \[\max_{i_1, \ldots, i_k} (\be_{(i_1, \ldots, i_k)}^{(t)})^{-1}f_{(i_1, \ldots, i_k)} \ge (1 + \varepsilon)\frac{\sum_{i_1, \ldots, i_k} f_{(i_1,\ldots, i_k)}}{\ln 2},\]
    \item \texttt{CountMore}: The number of values $t$ such that \[\max_{i_1, \ldots, i_k} (\be_{(i_1, \ldots, i_k)}^{(t)})^{-1}f_{(i_1, \ldots, i_k)} \le (1 - \varepsilon)\frac{\sum_{i_1, \ldots, i_k} f_{(i_1,\ldots, i_k)}}{\ln 2},\]
    \item \texttt{CountHeavy}: The number of values $t$ such that 
    \[\sum_{(i_1, \ldots, i_k)}(\be^{(t)}_{(i_1, \ldots, i_k)})^{-1}f_{(i_1, \ldots, i_k)} \le (C\log^2 n)\max_{i_1, \ldots, i_k}(\be_{(i_1, \ldots, i_k)}^{(t)})^{-1}f_{i_1, \ldots, i_k}.\]
\end{enumerate}
Note that the algorithm can keep track of all these random variables only using $O(b)$ bits of space as follows: When processing the first set of discrete exponential random variables, the algorithm keeps track of cumulative sum and cumulative max corresponding to those set of random variables and at the end updates the counts appropriately. It discards the stored cumulative sum and cumulative max values and starts processing the second set of random variables and so on. 

We now note using the properties of continuous exponential random variables that when the discrete exponential random variables are sampled in the above defined manner using a fully random string, then with probability $\ge 99/100$, using the union bound over the properties of continuous exponential random variables, the following happen:
\begin{enumerate}
    \item $\texttt{CountLess} < m/2$,
    \item $\texttt{CountMore} < m/2$, and
    \item $\texttt{CountHeavy} = m$.
\end{enumerate}
The first two properties from the fact that the median of $O(1/\varepsilon^2)$ independent copies of the random variable $(\sum_{(i_1, \ldots, i_k)} f_{(i_1, \ldots, i_k)})/\be$ concentrates in the interval $[(1-\varepsilon)\sum_{(i_1, \ldots, i_k)} f_{(i_1, \ldots, i_k)}/\ln 2,(1 + \varepsilon)\sum_{(i_1, \ldots, i_k)} f_{(i_1, \ldots, i_k)}/\ln 2]$ and hence both the counts are at most $m/2$. The third property follows from using a union bound on the event in Lemma~\ref{lma:max-is-significant}.
Since the algorithm $\textsf{Alg}$ uses only a space of $O(b)$ bits and the number of required random bits is $\exp(b)$, if the discrete exponential random variables are constructed using a pseudorandom string drawn from Nisan's PRG with a seed length of $O(b^2)$ bits, the above properties continue to hold with probability $\ge 98/100$. Hence, the protocol run with discrete exponential random variables constructed using Nisan's PRG continues to succeed in outputting a $1 \pm \varepsilon$ approximation to the higher-order correlation defined by the functions $f$ and $g$ also succeeds with probability $\ge 98/100$.